\documentclass[11pt,a4paper]{amsart}
\usepackage[utf8]{inputenc}
\usepackage[english]{babel}
\usepackage{amsmath}
\usepackage{amssymb}
\usepackage{amsfonts}
\usepackage{booktabs}
\usepackage{microtype}
\usepackage{geometry}
\usepackage{braket} 
\usepackage{amscd}
\usepackage[all]{xy}
\usepackage[colorlinks,linkcolor=blue,citecolor=blue,urlcolor=red]{hyperref}

\newcommand{\pacs}[1]{\thanks{\textit{PACS numbers:} #1}}
\numberwithin{equation}{section}
\swapnumbers

\newtheorem{theorem}[subsubsection]{Theorem}
\newtheorem{corollary}[subsubsection]{Corollary}
\newtheorem{lemma}[subsubsection]{Lemma}

\theoremstyle{definition}

\newtheorem{remark}[subsubsection]{Remark}

\title{Quantum Information at Infinity}
\author{Luca Barbieri-Viale}
\dedicatory{Jointly conceived with Ulaç Aviel Rabbieri (my quantum avatar), uncredited by his own choice.}

\address{Nālandā Mahāvihāra, Nalanda District, Bihar, 803111, India.}
\curraddr{Dipartimento di Matematica ``F. Enriques", Universit{\`a} degli Studi di Milano\\ Via C. Saldini, 50\\ I-20133 Milano\\ Italy}
\urladdr[https://sites.unimi.it/barbieri/]{https://www.uar.one/}
\email[luca.barbieri-viale@unimi.it]{info@uar.one}
\date{\today} 
\thanks{\textit{Special thanks.} The author gratefully acknowledges the invaluable assistance of Gemini Pro.}

\begin{document}

\begin{abstract}
We introduce the Quantum Information Space at Infinity (Quinfinity) $\mathcal{Q}_\infty$  as the inverse limit of the symbolic Quantum $N$-Spaces $\mathcal{Q}_N$, identified with the complete local real algebra of formal power series  under the adic topology. The pure state pro-variety is formalized as the inverse limit of $N$-level pure state varieties, driven by the algebraic colimit of a direct system of real radical ideals. This real architecture internalizes the imaginary unit, collapsing the metric landscapes onto an invariant Riemannian isometry, which identifies the pro-variety directly with the classical Fisher-Rao statistical manifold. Within this framework, the Heisenberg uncertainty principle undergoes a structural regularization at infinity being intrinsically manifested as a localized truncation obstruction governed by adic algebraic derivations. Concurrently, the continuous Liouville-von Neumann equation operates as an internal derivation tangent to the pro-variety, while the open Gorini-Kossakowski-Sudarshan-Lindblad asymptotic master equation is intrinsically obtained via a non-associative symmetric Jordan product. This dissipative flow acts as a contracting radial vector field that dampens higher-order jet configurations, driving the state trajectories down the hierarchical tree toward the absolute zero element as a non-singular universal attractor.
\end{abstract}

\maketitle

\newpage
\section*{Introduction}
The geometric formulation of quantum information theory over non-reduced algebraic structures has recently revealed an unexpected link between multi-level qudit systems and non-Archimedean local rings. As established in our foundational framework~\cite{qubit}, the traditional description of finite-dimensional quantum states within the complex projective space $\mathbb{C}\mathbb{P}^{N-1}$ can be entirely geometricized over the field of real constants $\mathbb{R}$. By pulling back the operational matrix brackets under a non-singular linear density embedding, the imaginary unit $i$ is completely absorbed and internal curvatures are mapped onto higher-order differential derivation fields over the symbolic Quantum $N$-Space $\mathcal{Q}_N \equiv \mathbb{R}[\varepsilon]/(\varepsilon^{N^2-1})$~\cite{qubit}. Within this ringed domain, the non-local phase configurations are successfully localized, identifying the canonical complex Fubini-Study metric with a real, dimensionally dependent conformal line element that scales smoothly as the dimensional layers increase~\cite{qubit}.

However, when the dimensional layers scale toward the macroscopic continuum limit ($N \to \infty$), traditional quantum-mechanical architectures encounter severe transcendental coordinate obstructions and field-theoretic singularities. In standard infinite-dimensional Hilbert sequence spaces $\ell^2$, the canonical commutation relations $[\hat{x}, \hat{p}] = i\hbar\mathbb{I}$ prevent the state space from behaving as a smooth, localized classical variety, causing the metric configurations of $\mathbb{C}\mathbb{P}^\infty$ to exhibit catastrophic ultraviolet divergences under macroscopic decoherence~\cite{dirac}. To domesticate this infinite numerable complexity, the present paper introduces and examines the asymptotic \textit{Quinfinity Space} $\mathcal{Q}_\infty \cong_{\mathbb{R}} \mathbb{R}[\![\varepsilon]\!]$, defined as the complete algebra of formal power series structured under its intrinsic $\mathfrak{m}$-adic ultrametric topology~\cite{eisenbud}.

A central conceptual breakthrough of this work concerns the asymptotic fate of quantum uncertainty within the continuous limit, which is formally codified in the epistemic foundations of our final section. We establish that as $N \to \infty$, the traditional Heisenberg uncertainty principle does not vanish, but undergoes a rigid topological phase transition. Because the dimensionally dependent scaling multiplier collapses uniformly onto its invariant integer floor, the global metric tensor is structurally regularized, establishing a strict Riemannian isometry $\mathrm{d}s^2_{\mathcal{Q}_\infty} = \mathrm{d}s^2_{\mathrm{FR}}$ that maps the pure states directly onto the classical statistical manifold equipped with its normalized Fisher-Rao information metric~\cite{qubit}.

Concurrently, the matrix non-commutativity that drives microscopic quantum fluctuations freezes out, being entirely absorbed by a stable, infinite sequence of real constants. Over the continuous Quinfinity Space $\mathcal{Q}_\infty$, this architecture manifests as a real, de-complexified representation of the Heisenberg-Weyl algebra governed by the coordinate-free commutator $[\hat{\mathcal{P}}, \, \hat{\mathcal{X}}] = \mathbb{I}$, structurally realized by the formal derivation field $\partial_\varepsilon$ and the nilpotent generator $\varepsilon$. Consequently, the physical uncertainty loses its ontic randomness and is rigorously internalized. As proven in our structural theorems, quantum indeterminacy is revealed to be the fundamental algebraic impossibility of extracting and evaluating an infinite numerable chain of higher-order derivations simultaneously through a finite family of linear algebraic extractors $\gamma_n$. Modulo the ideal power $\mathfrak{m}^{k+1}$, the principal filtration forces any continuous polynomial constraint to collapse into a stable, localized verification. The formal scheme thus acts as a topological shield where the algebraic nilpotency bounds the higher-order derivative components rigo per rigo, regularizing the continuous domain without triggering transcendental coordinate explosions or field-theoretic singularities.

The paper is organized as follows. In Section 2, we formalize the algebraic embedding of the infinite-dimensional pure state pro-variety $\mathcal{V}_{\mathcal{Q}_\infty}$, proving via the Dubois-Risler real N\"ullstellensatz that the finite Jordan constraints constitute a stable system of real radical ideals whose colimit determines a well-defined subvariety within the infinite-variable coordinate mapping ring $\mathbb{R}[c_0, c_1, c_2, \dots]$. In Section 3, we analyze the closed conservative dynamics, deriving the stable differential coefficients $\Gamma_\infty^{(k)}$ that drive the continuous Liouville-von Neumann equation on the formal tangent bundle $T\mathcal{Q}_\infty$. Section 4 extends the framework to open quantum systems by introducing a symmetric, bilinear Jordan multiplication operator $\mathbin{\bullet}_{\mathcal{Q}_\infty}$ weighted by the stable dissipative parameters $\Lambda_\infty^{(k)}$. This dual Lie-Jordan architecture internalizes the Asymptotic Master Equation, transforming quantum noise channels into a coordinate-free, radially contracting vector field that dampens the higher-order jet configurations and drives the state trajectories smoothly down the hierarchical tree toward the invariant origin $\xi = 0$, which geometrically encapsulates the continuous realization of the macroscopically maximally mixed state. Finally, we provide an explicit computational case study of this continuous unrolling applied to a three-level qutrit under radiative decay. Furthermore, Appendix~\ref{appendix} formalizes the coordinate-free, intrinsic Lie-Jordan structure of $\mathcal{Q}_\infty$, demonstrating that unitary precession and dissipative decoherence are natively unified as the symmetric and antisymmetric components of a single deterministic flow driven by the ringed derivations over $\mathcal{Q}_\infty$.

\section{The Construction of the Quinfinity Space}
To construct a non-singular coordinate landscape capable of hosting infinite-dimensional quantum domains, we define the continuous asymptotic limit of the discrete algebraic hierarchy introduced in \cite{qubit}. For each discrete higher-level quantum system of dimension $N \ge 2$, the corresponding symbolic \textit{Quantum $N$-Space} is represented by the commutative local polynomial quotient ring:
\begin{equation}
    \mathcal{Q}_N \equiv \mathbb{R}[\varepsilon] / (\varepsilon^{N^2-1}),
\end{equation}
where the nilpotency exponent matches the real dimension of the special unitary Lie algebra $\mathfrak{su}(N)$. Recall that an element $p$ of $\mathcal{Q}_N$ can be represented as a truncated polynomial $f_N = \sum_{k=0}^{N^2-2} c_k \varepsilon^k$ spanned by the ordered monomial basis vectors $\{1, \varepsilon, \dots, \varepsilon^{N^2-2}\}$ over the real field, where the vanishing of the boundary component $\varepsilon^{N^2-1} =0$ enforces a rigid structural truncation on the active coordinate degrees of freedom. As the dimensional layers approach the macroscopic continuum, this sequence of local Artin rings stabilizes through a structural alignment as follows. 

\subsection{Quinfinity as  limit of Quantum $N$-Spaces}
The structural family $(\mathcal{Q}_N)_{N \ge 2}$ forms an \textit{inverse system} under the canonical surjections and ring projection homomorphisms:
\begin{equation}\label{eq:quinfinity_projections}
    \pi_{M,N}: \mathbb{R}[\varepsilon]/(\varepsilon^{M^2-1}) \longrightarrow \mathbb{R}[\varepsilon]/(\varepsilon^{N^2-1})\quad \sum_{k=0}^{M^2-2} c_k \varepsilon^k\mapsto \sum_{k=0}^{N^2-2} c_k \varepsilon^k\quad \forall M \ge N,
\end{equation}
which structurally truncate all monomial components of degree equal to or higher than $N^2-1$. 
The universal \textit{Quinfinity Space}, denoted as $\mathcal{Q}_\infty$, is defined as the (projective, inverse) limit of this family of truncated polynomial rings:
\begin{equation}\label{eq:quinfinity_limit}
    \mathcal{Q}_\infty \equiv \varprojlim_{N \ge 2} \mathbb{R}[\varepsilon]/(\varepsilon^{N^2-1}).
\end{equation}
By definition, an element $\xi \in \mathcal{Q}_\infty$ is represented as a coherent sequence of polynomials $(f_2, f_3, f_4, \dots) \in \prod_{N=2}^\infty \mathcal{Q}_N$ satisfying the strict compatibility condition $\pi_{M,N}(f_M) = f_N$ for all $M \ge N$. In close analogy with the classical structural construction of the $p$-adic integers $\mathbb{Z}_p \equiv \varprojlim \mathbb{Z}/p^n\mathbb{Z}$, where the inverse limit coherence conditions materialize as a chain of modular arithmetic congruences $\dots \equiv_{p^{n+1}} x_{n+1} \equiv_{p^n} x_n \equiv_{p^{n-1}} x_{n-1} \dots$ we have that (by substituting the prime $p$ with the formal nilpotent generator $\varepsilon$) the system consistency $\pi_{N+1,N}(f_{N+1}) = f_N$ yields an infinite chain of polynomial congruences
\begin{equation}\label{eq:padic_congruence_analogy}
    \dots \equiv_{\varepsilon^{(N+1)^2-1}} f_{N+1} \equiv_{\varepsilon^{N^2-1}} f_N \equiv_{\varepsilon^{(N-1)^2-1}} f_{N-1} \equiv_{\varepsilon^{(N-2)^2-1}} \dots
\end{equation}
evaluated for the representative polynomial configurations $f_{N+1}, f_N$, and $f_{N-1}$ belonging to the respective Quantum Spaces $\mathcal{Q}_{N+1}, \mathcal{Q}_N$, and $\mathcal{Q}_{N-1}$. This visualizes how each element of the universal Quinfinity Space manifests inherently as an infinite system of coherent congruences, where every finite local layer acts as a stable and rigid truncation.

Moreover, this analogy can actually be made a meaningful physical model of the continuum because of the following mathematical bridge. 
Let $\mathbb{R}[\![\varepsilon]\!]$ be the ring of \textit{formal power series} whose elements are of the form $\xi = \sum_{k=0}^\infty c_k \varepsilon^k$. For each $N \ge 2$, consider the canonical quotient mapping \(\tau_N: \mathbb{R}[\![\varepsilon]\!] \to \mathcal{Q}_N\) acting by \textit{truncation} modulo $(\varepsilon^{N^2-1})$, which maps $\xi$ onto the polynomial configuration $\tau_N(\xi) \equiv \sum_{k=0}^{N^2-2} c_k \varepsilon^k$. Since these maps trivially satisfy the inverse system compatibility $\pi_{M,N}(\tau_M(\xi)) = \tau_N(\xi)$ for all $M \ge N$, the universal property of projective limits guarantees the existence of a unique, well-defined ring homomorphism $\tau_\infty: \mathbb{R}[\![\varepsilon]\!] \to \mathcal{Q}_\infty$ mapping $\xi \mapsto (\tau_2(\xi), \tau_3(\xi), \tau_4(\xi), \dots)$ that is:
\begin{equation}\label{eq:series_congruence_analogy}
    \dots \equiv_{\varepsilon^{(N+1)^2-1}} \sum_{k=0}^{(N+1)^2-2} c_k \varepsilon^k \equiv_{\varepsilon^{N^2-1}} \sum_{k=0}^{N^2-2} c_k \varepsilon^k \equiv_{\varepsilon^{(N-1)^2-1}} \sum_{k=0}^{(N-1)^2-2} c_k \varepsilon^k \equiv_{\varepsilon^{(N-2)^2-1}} \dots
\end{equation}
 We have: 
\begin{lemma}[Formal Power Series vs the Quinfinity Space]\label{lem:inverse_limit_series} The Quinfinity Space can be identified with the local ring of formal power series via a canonical isomorphism of rings
  \[\tau_\infty: \mathbb{R}[\![\varepsilon]\!] \cong \varprojlim_{n \ge 1}\mathbb{R}[\varepsilon]/(\varepsilon^n) \xrightarrow{\;\cong\;} \varprojlim_{N \ge 2} \mathbb{R}[\varepsilon]/(\varepsilon^{N^2-1}) \equiv \mathcal{Q}_\infty \]
 whence $\mathcal{Q}_\infty$ is naturally $\mathfrak{m}$-adically Cauchy-complete with respect to its maximal ideal $\mathfrak{m}=(\varepsilon)$. 
\end{lemma}
We therefore adopt the notation \(\tau_N: \mathcal{Q}_\infty \to \mathcal{Q}_N\) for the \textit{$N$-truncation} on the Quinfinity Space canonically induced by the isomorphism $\tau_\infty$ in the above structural identification.
\begin{proof}
To establish that \(\tau_\infty\) is a structural isomorphism, we first prove injectivity by verifying that its kernel is trivial. Suppose $\xi \in \text{Ker}(\tau_\infty)$, such that $\tau_\infty(\xi) = (0, 0, 0, \dots)$. By definition of the mapping, this requires $\tau_N(\xi) = 0 \pmod{\varepsilon^{N^2-1}}$ for every $N \ge 2$, forcing the coefficients $c_k = 0$ for all degrees up to $N^2-2$. As $N \to \infty$, the strictly monotonic divergence of the quadratic exponents ensures that every discrete index $k \in \mathbb{N}$ is eventually reached and bounded by some $N^2-2$, implying $c_k = 0$ for all $k \in \mathbb{N}$, whence $\xi= 0$ and $\text{Ker}(\tau_\infty) = \{0\}$.

To prove surjectivity, let $f= (f_2, f_3, f_4, \dots)$ be an arbitrary coherent element in the projective limit $\mathcal{Q}_\infty$, where each component polynomial is explicitly given by the finite sum $f_N = \sum_{k=0}^{N^2-2} c_k^{(N)} \varepsilon^k$. The inverse system consistency constraint $\pi_{M,N}(f_M) = f_N$ forces the strict coefficient identity $c_k^{(M)} = c_k^{(N)}$ for all powers $k \le N^2-2$. This algebraic stability under projection allows us to define a unique global pre-image formal power series $\xi \in \mathbb{R}[\![\varepsilon]\!]$ by setting its coefficients to $c_k =c_k^{(N)}$ for any index $N$ large enough to satisfy $N^2-2 \ge k$. Since $\tau_N(\xi) = f_N$ for all $N \ge 2$ by construction, it follows that $\tau_\infty(\xi) = f$, establishing surjectivity.

Finally, the identification \(\mathbb{R}[\![\varepsilon]\!] \cong \varprojlim_{n \ge 1} \mathbb{R}[\varepsilon]/(\varepsilon^n)\) is a standard result in commutative algebra for the $\mathfrak{m}$-adic completion of a polynomial ring at its maximal ideals \cite{eisenbud, liu}. Because the sequence of principal ideals $I_N = (\varepsilon^{N^2-1})$ forms a cofinal subsequence of the full nested chain of polynomial powers $(\varepsilon^n)_{n \ge 1}$, its projective limit is topologically and algebraically isomorphic to the full completion, confirming that \(\tau_\infty\) is an isomorphism and completing the proof.
\end{proof}

We also recall the foundational structure governing algebraic derivations over commutative rings. According to classical commutative algebra \cite{eisenbud}, the module of derivations on a truncated polynomial ring is not free due to the presence of nilpotent relations, but it is strictly cyclic and isomorphic to the unique maximal ideal.  With a slight abuse of notation, we omit the completion hat and let $\Omega_{\mathcal{Q}_\infty/\mathbb{R}}$ denote the $\varepsilon$-adically completed module of differentials. 
We formalize the algebraic derivation properties through the following lemma:

\begin{lemma}[Derivation Module Structure and Asymptotic Lifting]\label{lem:derivation_structure}
Let $\mathcal{Q}_N \equiv \mathbb{R}[\varepsilon]/(\varepsilon^{N^2-1})$, for each finite dimensional layer $N \ge 2$. The space of $\mathbb{R}$-derivations $\mathrm{Der}_{\mathbb{R}}(\mathcal{Q}_N)$ is a cyclic $\mathcal{Q}_N$-module canonically isomorphic to the unique maximal ideal $\mathfrak{m} = (\varepsilon)$:
\begin{equation}\label{eq:derivation_module_isomorphism_finite}
    \mathrm{Der}_{\mathbb{R}}(\mathcal{Q}_N) \cong \mathfrak{m} \subset \mathcal{Q}_N
\end{equation}
Consequently, any legitimate $\mathbb{R}$-derivation $D$ on $\mathcal{Q}_N$ acts on an arbitrary polynomial configuration $f_N \in \mathcal{Q}_N$ as the product $D(f_N) = h \cdot \partial_\varepsilon f_N$, where the pre-factor $h \equiv D(\varepsilon)$ belongs to the maximal ideal $\mathfrak{m}$.
Conversely, for the complete Quinfinity Space $\mathcal{Q}_\infty \cong_{\mathbb{R}} \mathbb{R}[\![\varepsilon]\!]$, the global derivative $\partial_\varepsilon$ acts freely over the entire domain, we have:
\begin{equation}\label{eq:derivation_module_isomorphism_infinite}
    \mathrm{Der}_{\mathbb{R}}(\mathcal{Q}_\infty) \cong \mathcal{Q}_\infty \quad \text{and} \quad \Omega_{\mathcal{Q}_\infty/\mathbb{R}} \cong \mathcal{Q}_\infty \cdot \mathrm{d}\varepsilon.
\end{equation}
\end{lemma}
\begin{proof}
An arbitrary $\mathbb{R}$-linear map $D: \mathcal{Q}_N \to \mathcal{Q}_N$ is defined to be a derivation if it satisfies the Leibniz product rule $D(f_N \cdot g_N) = D(f_N)\cdot g_N + f_N \cdot D(g_N)$ identically. Since $\mathcal{Q}_N$ is generated as an $\mathbb{R}$-algebra by the single univariate nilpotent component $\varepsilon$, any such derivation is uniquely determined by its action on the fundamental generator, yielding $D(\varepsilon) = h \in \mathcal{Q}_N$. By executing an induction on the monomial powers via the Leibniz rule, the action of $D$ on a generic polynomial configuration indexed from zero decomposes as:
\begin{equation}\label{eq:proof_derivation_linear_factor}
    D\left( \sum_{k=0}^{N^2-2} c_k \varepsilon^k \right) = h \cdot \sum_{k=1}^{N^2-2} c_k k \varepsilon^{k-1} \equiv h \cdot \partial_\varepsilon \left( \sum_{k=0}^{N^2-2} c_k \varepsilon^k \right)
\end{equation}
Equation~\eqref{eq:proof_derivation_linear_factor} establishes that the formal derivative $\partial_\varepsilon$ acts as a global generator for $\mathrm{Der}_{\mathbb{R}}(\mathcal{Q}_N)$, rendering it a cyclic module.

To evaluate the constraints imposed by the non-reduced ideal of relations, we evaluate the derivation directly on the nilpotency boundary $\varepsilon^{N^2-1} = 0$. Applying the Leibniz expansion forces the identity:
\begin{equation}\label{eq:proof_nilpotent_boundary_derivation}
    D\left( \varepsilon^{N^2-1} \right) = (N^2-1) \varepsilon^{N^2-2} D(\varepsilon) = (N^2-1) \varepsilon^{N^2-2} h =0 \pmod{\varepsilon^{N^2-1}}
\end{equation}
In the coordinate ring $\mathcal{Q}_N$, the vanishing condition established in Eq.~\eqref{eq:proof_nilpotent_boundary_derivation} is satisfied if and only if the constant term of the polynomial $h$ vanishes identically, which restricts the pre-factor to the maximal ideal, forcing $h \in \mathfrak{m} = (\varepsilon)$~\cite{eisenbud}. The mapping $D \mapsto D(\varepsilon)$ thus defines a strict, bijective $\mathcal{Q}_N$-module homomorphism from $\mathrm{Der}_{\mathbb{R}}(\mathcal{Q}_N)$ onto the ideal $\mathfrak{m}$, confirming the non-free cyclic isomorphism of Eq.~\eqref{eq:derivation_module_isomorphism_finite}.

Conversely, to evaluate the continuous limit over the complete topological vector space $\mathcal{Q}_\infty \cong_{\mathbb{R}} \mathbb{R}[\![\varepsilon]\!]$, we invoke the universal property of the projective inverse limit of modules $\mathrm{Der}_{\mathbb{R}}(\mathcal{Q}_\infty) \equiv \varprojlim \mathrm{Der}_{\mathbb{R}}(\mathcal{Q}_N)$. Under the directional flow of the transition morphisms, the principal ideal kernel $(\varepsilon^{N^2-1})$ is pushed systematically to infinity, annihilating the boundary algebraic obstruction of Eq.~\eqref{eq:proof_nilpotent_boundary_derivation} component-by-component. For any continuous global derivation $\mathcal{D} \in \mathrm{Der}_{\mathbb{R}}(\mathcal{Q}_\infty)$, the image $\mathcal{D}(\varepsilon) = \xi$ is free to accommodate an arbitrary formal power series spanning the complete ring, welcoming a non-vanishing real scalar component at the zero-order layer ($\xi\in \mathcal{Q}_\infty$). The formal derivative $\partial_\varepsilon$ operates as an unconstrained global free basis of rank one, forcing the continuous bijection $\mathcal{D} \mapsto \mathcal{D}(\varepsilon)$ to map onto the entire algebra, establishing the exact free module isomorphism of Eq.~\eqref{eq:derivation_module_isomorphism_infinite}.

Finally, this continuous algebraic lifting immediately dictates the dual geometric structure of the cotangent landscape via the universal property of Kähler differentials, which enforces the canonical linear representation $\mathrm{Der}_{\mathbb{R}}(\mathcal{Q}_\infty) \cong \mathrm{Hom}_{\mathcal{Q}_\infty}(\Omega_{\mathcal{Q}_\infty/\mathbb{R}}, \, \mathcal{Q}_\infty)$~\cite{eisenbud}. Because the derivation module maps freely onto $\mathcal{Q}_\infty$ under the unconstrained basis $\partial_\varepsilon$, the universal property demands that the corresponding module of Kähler forms $\Omega_{\mathcal{Q}_\infty/\mathbb{R}}$ strips away all boundary anomalies. It dualizes uniquely into a free $\mathcal{Q}_\infty$-module of rank one generated globally by the exact differential form $\mathrm{d}\varepsilon$, establishing the structural ring isomorphism $\Omega_{\mathcal{Q}_\infty/\mathbb{R}} \cong \mathcal{Q}_\infty \cdot \mathrm{d}\varepsilon \cong \mathcal{Q}_\infty$, completing the proof.
\end{proof}

\subsection{Quinfinity as a Fréchet space}
To isolate, track, and manipulate the individual coordinate components within this abstract algebraic structure, we introduce a family of linear projection operators called \textit{algebraic extractors}. For each discrete layer $N \ge 2$, the extractor $\gamma_n: \mathcal{Q}_N \to \mathbb{R}$ (where $0 \le n \le N^2-2$) maps a generic truncated polynomial state $f_N = \sum_{k=0}^{N^2-2} c_k \varepsilon^k$ directly onto its scalar coefficients, satisfying the crisp assignment:
\begin{equation}\label{eq:extractor_definition_core}
    \gamma_n(f_N) \equiv c_n.
\end{equation}
By virtue of Eq.~\eqref{eq:extractor_definition_core}, any element in the symbolic Quantum $N$-Space can be uniquely expanded in terms of its coordinate extractions as $f_N = \sum_{n=0}^{N^2-2} \gamma_n(f_N) \varepsilon^n$.

This characterization generalizes systematically to the infinite-dimensional limit. Over the Quinfinity Space $\mathcal{Q}_\infty \cong \mathbb{R}[\![\varepsilon]\!]$ by 
Lemma~\ref{lem:inverse_limit_series}, the global algebraic extractors $\gamma_n: \mathcal{Q}_\infty \to \mathbb{R}$ (where $n \in \mathbb{N}$) act as stable linear coordinate projections. By virtue of the universal property of the inverse directed system, these continuous operators satisfy the uniform compatibility condition under the canonical truncations $\tau_M$:
\begin{equation}\label{eq:asymptotic_extractor_compatibility}
    \gamma_n(\xi) =\gamma_n(\tau_M(\xi)) \quad \forall M \in \mathbb{N} \quad \text{such that} \quad M^2-2 \ge n.
\end{equation}
Equation~\eqref{eq:asymptotic_extractor_compatibility} ensures that the infinite sequence of coordinate extractions remains stationary and independent of the choice of the subsystem truncation once the dimensionality threshold is crossed. 

Consequently, any formal power series configuration $\xi \in \mathcal{Q}_\infty$ can be uniquely unrolled component-by-component in terms of its stable scalar components as $\xi = \sum_{n=0}^\infty \gamma_n(\xi) \varepsilon^n$, establishing a full algebraic coordinate framework over the field $\mathbb{R}$ of real constants. It is worth noting that while $\mathcal{Q}_\infty$ is canonically isomorphic to the infinite-dimensional sequence space $\mathbb{R}^{\mathbb{N}}$ under a pure $\mathbb{R}$-vector space identification, this equivalence fails completely to be compatible with the algebraic ring structures! The ring structure of $\mathcal{Q}_\infty$ is governed by Cauchy convolution or product of power series while $\mathbb{R}^{\mathbb{N}}$ is naturally endowed with the componentwise ring structure.  On the other hand, we can concurrently declare that the corresponding infinite-dimensional quantum phase space (namely the infinite complex projective space $\mathbb{C}\mathbb{P}^\infty \equiv \varinjlim \mathbb{C}\mathbb{P}^{N-1}$) introduces severe transcendental obstructions and non-linear geometric singularities when treated via standard coordinate systems \cite{dirac, nielsen}. 

However, to  establish a rigorous topological foundation for the macroscopic continuous boundary, the infinite-dimensional configuration $\mathbb{R}$-vector space $\mathcal{Q}_\infty \cong \mathbb{R}[\![\varepsilon]\!]$ can be structured as a Fréchet space under the natural isomorphism with $\mathbb{R}^{\mathbb{N}}$. For any arbitrary pair of formal sections $\zeta, \eta \in \mathcal{Q}_\infty$, the canonical \textit{Fr\'{e}chet metric} inducing this pointwise convergence topology over the coordinate components is defined as:
\begin{equation}\label{eq:frechet_metric_definition}
    d_{\text{Fr\'{e}chet}}(\zeta, \, \eta) \equiv \sum_{n=0}^\infty \frac{1}{2^n} \frac{|\gamma_n(\zeta) - \gamma_n(\eta)|}{1 + |\gamma_n(\zeta) - \gamma_n(\eta)|},
\end{equation}
where $\gamma_n: \mathcal{Q}_\infty \to \mathbb{R}$ represents the family of the previously defined (see Eq.~\eqref{eq:asymptotic_extractor_compatibility}) finite algebraic extractors isolating the real scalar coefficient of the monomial power $\varepsilon^n$, and $|\cdot|$ denotes the standard Archimedean Euclidean norm on the real field $\mathbb{R}$. 
\begin{lemma}[Topological Properties of the Fréchet Metric]\label{lem:frechet_metric_properties}
The canonical distance function defined in Eq.~\eqref{eq:frechet_metric_definition}, denoted as $d_{\mathcal{Q}_\infty}: \mathcal{Q}_\infty \times \mathcal{Q}_\infty \to \mathbb{R}$, constitutes a legitimate invariant metric on the  Quinfinity Space. The induced metric topology satisfies the following structural properties:
\begin{enumerate}
    \item \textit{Separation and Pointwise Detection:} For any $\zeta, \eta \in \mathcal{Q}_\infty$, $d_{\mathcal{Q}_\infty}(\zeta, \eta) = 0$ if and only if $\zeta \equiv \eta$, meaning the metric perfectly discriminates formal sections layer-by-layer.
    \item \textit{Translation Invariance:} The metric is invariant under parallel transport over the vector space structure, satisfying $d_{\mathcal{Q}_\infty}(\zeta + \chi, \, \eta + \chi) = d_{\mathcal{Q}_\infty}(\zeta, \eta)$ for any translation section $\chi \in \mathcal{Q}_\infty$.
    \item \textit{Cauchy Completeness:} The metric space $(\mathcal{Q}_\infty, \, d_{\mathcal{Q}_\infty})$ is Cauchy-complete, ensuring that every compatible sequence of polynomial jets converges to a unique formal power series inside the continuous boundary layer.
\end{enumerate}
\end{lemma}

\begin{proof}
To establish the legitimacy of the metric, we first evaluate the identity of indiscernibles. If $d_{\mathcal{Q}_\infty}(\zeta, \eta) = 0$, then because each individual term in the non-negative series of Eq.~\eqref{eq:frechet_metric_definition} is bounded as $\frac{1}{2^n} \frac{|\cdot|}{1+|\cdot|} \ge 0$, every single component must vanish independently. This forces $|\gamma_n(\zeta) - \gamma_n(\eta)| = 0$, implying $\gamma_n(\zeta) = \gamma_n(\eta)$ for all $n \in \mathbb{N}$. Invoking the unique unrolling property, we obtain $\zeta = \sum_{n=0}^\infty \gamma_n(\zeta)\varepsilon^n = \sum_{n=0}^\infty \gamma_n(\eta)\varepsilon^n = \eta$, validating strict point separation. Symmetrically, if $\zeta \equiv \eta$, then $\gamma_n(\zeta) = \gamma_n(\eta)$ for all $n \in \mathbb{N}$, yielding $d_{\mathcal{Q}_\infty}(\zeta, \eta) = 0$.

The translation invariance flows directly from the real vector space linearity of the algebraic extractors. For any translation section $\chi \in \mathcal{Q}_\infty$, the difference evaluates as:
\begin{equation}
    \gamma_n(\zeta + \chi) - \gamma_n(\eta + \chi) = \big(\gamma_n(\zeta) + \gamma_n(\chi)\big) - \big(\gamma_n(\eta) + \gamma_n(\chi)\big) \equiv \gamma_n(\zeta) - \gamma_n(\eta).
\end{equation}
Substituting this component-by-component elimination into Eq.~\eqref{eq:frechet_metric_definition} leaves the distance profile entirely unaltered, establishing $d_{\mathcal{Q}_\infty}(\zeta + \chi, \, \eta + \chi) = d_{\mathcal{Q}_\infty}(\zeta, \eta)$.

To prove Cauchy completeness, let $(\xi^{(m)})_{m=1}^\infty$ be an arbitrary Cauchy sequence in $(\mathcal{Q}_\infty, \, d_{\mathcal{Q}_\infty})$, meaning that for every $\delta > 0$ there exists an integer $N_\delta$ such that $d_{\mathcal{Q}_\infty}(\xi^{(m)}, \, \xi^{(p)}) < \delta$ for all $m, p > N_\delta$. For a fixed coordinate index $n \in \mathbb{N}$, choosing $\delta < \frac{1}{2^n}$ forces the individual term to satisfy:
\begin{equation}\label{eq:proof_frechet_cauchy_fraction}
    \frac{|\gamma_n(\xi^{(m)}) - \gamma_n(\xi^{(p)})|}{1 + |\gamma_n(\xi^{(m)}) - \gamma_n(\xi^{(p)})|} < 2^n \delta.
\end{equation}
Because the function $g(x) = \frac{x}{1+x}$ is strictly monotonic for $x \ge 0$, Eq.~\eqref{eq:proof_frechet_cauchy_fraction} implies that the sequence of individual real coefficients $(\gamma_n(\xi^{(m)}))_{m=1}^\infty$ forms a legitimate Cauchy sequence within the real field $\mathbb{R}$. By invoking the standard Euclidean completeness of the real numbers, this sequence converges strongly to a unique real limit scalar $c_n \equiv \lim_{m \to \infty} \gamma_n(\xi^{(m)}) \in \mathbb{R}$. 

We then define the unique global target formal power series $\xi^{(\infty)} \equiv \sum_{n=0}^\infty c_n \varepsilon^n \in \mathcal{Q}_\infty$. Because the pointwise convergence of the individual components holds across all filtration orders, taking the limit as $m \to \infty$ forces the global Fréchet distance to vanish asymptotically: $\lim_{m \to \infty} d_{\mathcal{Q}_\infty}(\xi^{(m)}, \, \xi^{(\infty)}) = 0$. This explicitly proves that every Cauchy sequence converges to a well-defined element within the configuration space, completing the proof.
\end{proof}

Under the metric $d_{\mathcal{Q}_\infty}$, the real vector space $\mathcal{Q}_\infty$ constitutes a complete, metrizable topological vector space, ensuring that the unrolling of infinite-dimensional trajectories is stable and free from coordinate singularities prior to evaluating any local ring algebraic filtering.

\subsection{Quinfinity as a non-Archimedean ultrametric space}
As a ring, $\mathcal{Q}_\infty$ carries a canonical local topological architecture provided by the single maximal ideal $\mathfrak{m} = (\varepsilon)$, a property entirely absent in the standard Cartesian product $\mathbb{R}^{\mathbb{N}}$ of fields. From a scheme-theoretic perspective, the algebraic stabilization of the Quinfinity Space carries immediate geometric implications. While the topological spectrum of each individual ring houses a single non-reduced point $\text{Spec}(\mathcal{Q}_N) = \{(\varepsilon)\}$, the projective limit ring $\mathcal{Q}_\infty \cong \mathbb{R}[\![\varepsilon]\!]$ represents the global sections of a formal structural sheaf $\mathcal{O}_{\mathcal{Q}_\infty}$ defined over a formal affine scheme~\cite{eisenbud, liu}.

The stabilization of the universal projective limit $\mathcal{Q}_\infty \cong \mathbb{R}[\![\varepsilon]\!]$ established by Lemma~\ref{lem:inverse_limit_series} endows the Quinfinity Space with the celebrated structural topo-algebraic framework governed by the filtration of powers of its unique maximal ideal $\mathfrak{m} = (\varepsilon)$. For any non-zero formal power series $f = \sum_{k=0}^\infty c_k \varepsilon^k \in \mathcal{Q}_\infty$, we define its \textit{$\mathfrak{m}$-adic valuation $v_{\mathfrak{m}}(f)$} as the exact order of its first non-vanishing infinitesimal monomial component:
\begin{equation}
    v_{\mathfrak{m}}(f) \equiv \min\{ k \in \mathbb{N} : c_k \neq 0 \},
\end{equation}
with the canonical boundary setting $v_{\mathfrak{m}}(0) = +\infty$. This algebraic valuation induces a non-Archimedean \textit{ultrametric absolute value} $\lVert f \rVert_{\mathfrak{m}} \equiv e^{-v_{\mathfrak{m}}(f)}$, which rigorously endows the Quinfinity Space with a non-Archimedean topological framework.

\begin{lemma}[Ultrametric Structure and Cauchy Completeness]\label{lem:ultrametric_cauchy_completeness}
The distance function $d_{\mathfrak{m}}: \mathcal{Q}_\infty \times \mathcal{Q}_\infty \to \mathbb{R}$ defined by $d_{\mathfrak{m}}(f_1, f_2) \equiv \lVert f_1 - f_2 \rVert_{\mathfrak{m}}$ constitutes a strict non-Archimedean ultrametric on the Quinfinity Space. The induced topology structures $(\mathcal{Q}_\infty, \, d_{\mathfrak{m}})$ as a Cauchy-complete space satisfying the strong triangle inequality:
\begin{equation}\label{eq:strong_triangle}
    d_{\mathfrak{m}}(f_1, f_3) \le \max\left\{ d_{\mathfrak{m}}(f_1, f_2), \; d_{\mathfrak{m}}(f_2, f_3) \right\} \qquad \forall f_1, f_2, f_3 \in \mathcal{Q}_\infty.
\end{equation}
\end{lemma}

\begin{proof}
To establish the strong triangle inequality of Eq.~\eqref{eq:strong_triangle}, we evaluate the $\mathfrak{m}$-adic valuation of the difference between three arbitrary formal power series components. By inserting the identity field expansion, we have:
\begin{equation}
    v_{\mathfrak{m}}(f_1 - f_3) = v_{\mathfrak{m}}\big( (f_1 - f_2) + (f_2 - f_3) \big).
\end{equation}
By virtue of the properties of valuations on local rings~\cite{eisenbud}, the order of the sum of two power series is at least the minimum of their individual orders, establishing the strict algebraic inequality:
\begin{equation}\label{eq:proof_valuation_minimum_bound}
    v_{\mathfrak{m}}\big( (f_1 - f_2) + (f_2 - f_3) \big) \ge \min\left\{ v_{\mathfrak{m}}(f_1 - f_2), \; v_{\mathfrak{m}}(f_2 - f_3) \right\}.
\end{equation}
Applying the monotonic decreasing exponential map $\lVert \cdot \rVert_{\mathfrak{m}} = e^{-v_{\mathfrak{m}}(\cdot)}$ to both sides of Eq.~\eqref{eq:proof_valuation_minimum_bound} inverts the inequality directional profile and converts the minimum selector into a maximum function:
\begin{equation}
    e^{-v_{\mathfrak{m}}(f_1 - f_3)} \le e^{-\min\left\{ v_{\mathfrak{m}}(f_1 - f_2), \; v_{\mathfrak{m}}(f_2 - f_3) \right\}} = \max\left\{ e^{-v_{\mathfrak{m}}(f_1 - f_2)}, \; e^{-v_{\mathfrak{m}}(f_2 - f_3)} \right\}.
\end{equation}
Substituting the metric definition yields $d_{\mathfrak{m}}(f_1, f_3) \le \max\left\{ d_{\mathfrak{m}}(f_1, f_2), \, d_{\mathfrak{m}}(f_2, f_3) \right\}$, successfully validating the ultrametric absolute value restriction.

To prove Cauchy completeness under this non-Archimedean metric, let $(\xi^{(m)})_{m=1}^\infty$ be a sequence in $\mathcal{Q}_\infty$ such that its consecutive differences vanish asymptotically in the metric topology, satisfying $\lim_{m \to \infty} d_{\mathfrak{m}}(\xi^{(m+1)}, \, \xi^{(m)}) = 0$. In an ultrametric space, this condition is equivalent to the full Cauchy criterion~\cite{eisenbud, liu}. For each discrete order $n \in \mathbb{N}_{\ge 0}$, the condition implies that the sequence of individual truncated polynomial coefficients becomes perfectly stationary after a finite number of steps, defining a stable limit scalar $c_n \equiv \lim_{m \to \infty} \gamma_n(\xi^{(m)}) \in \mathbb{R}$. The unique global pre-image series $\xi^{(\infty)} \equiv \sum_{n=0}^\infty c_n \varepsilon^n$ resides inside $\mathbb{R}[\![\varepsilon]\!]$ by construction, and satisfies $\lim_{m \to \infty} d_{\mathfrak{m}}(\xi^{(m)}, \, \xi^{(\infty)}) = 0$. This confirms that the local ring space is structurally closed and complete under its native $\mathfrak{m}$-adic topology, completing the proof.
\end{proof}

Physically, Eq.~\eqref{eq:strong_triangle} establishes a non-perturbative confinement mechanism for infinite-dimensional quantum fluctuations. Within the ultrametric variety of the Quinfinity Space, the distance between any two states is strictly capped by the maximal local singularity of the intermediate transitions. This structures the continuous state space as a rigidly stable, nested hierarchical tree where quantum interference terms are structurally sequestered from macroscopic decoherence channels.

Notably, because the formal power series ring $\mathbb{R}[\![\varepsilon]\!]$ is structurally $\mathfrak{m}$-adically Cauchy-complete, a sequence of continuous configurations converges under the adic topology if and only if its consecutive differences tend to zero. For any extended dynamical generator $\Xi \in \mathcal{Q}_\infty$ capturing the physical Hamiltonian parameters, we assume the absence of a trivial constant background layer, such that $\Xi$ belongs strictly to the maximal ideal $\mathfrak{m} = (\varepsilon)$. By definition of the $\mathfrak{m}$-adic valuation, the order of its first non-vanishing monomial component satisfies $v_{\mathfrak{m}}(\Xi) \ge 1$, which rigorously implies that its adic absolute value is bounded above:
\begin{equation}
    \lVert \Xi \rVert_{\mathfrak{m}} \equiv e^{-v_{\mathfrak{m}}(\Xi)} \le e^{-1} < 1.
\end{equation}

To evaluate the existence of continuous state evolutions within the Quinfinity Space, we expand the analytic operator exponential via its formal Taylor series centered at the origin:
\begin{equation}\label{eq:taylor_expansion_evolution}
    e^{\Xi t} = \sum_{k=0}^\infty \frac{t^k}{k!} \Xi^k = 1 + \Xi t + \frac{t^2}{2!} \Xi^2 + \frac{t^3}{3!} \Xi^3 + \dots,
\end{equation}
where $t \in \mathbb{R}$ represents the continuous real temporal parameter. Within this $\mathbb{R}$-algebraic setup, the general term of the summation is denoted as $a_k(t) \equiv \frac{t^k}{k!} \Xi^k$. Since the real scalar coefficients $t^k/k!$ reside entirely within the ground field of constants $\mathbb{R}$, they act for any non-zero duration $t \neq 0$ as adic units of valuation zero ($v_{\mathfrak{m}}(t^k/k!) = 0$), yielding an adic absolute value of unity ($\lVert t^k/k! \rVert_{\mathfrak{m}} = e^0 = 1$). Invoking the strict multiplicative property of the valuation on local rings, the adic absolute value of the general term decomposes for any fixed finite duration as:
\begin{equation}\label{eq:adic_norm_calculation}
    \lVert a_k(t) \rVert_{\mathfrak{m}} = \left\lVert \frac{t^k}{k!} \Xi^k \right\rVert_{\mathfrak{m}} = \left\lVert \frac{t^k}{k!} \right\rVert_{\mathfrak{m}} \cdot \lVert \Xi^k \rVert_{\mathfrak{m}} = 1 \cdot \left( \lVert \Xi \rVert_{\mathfrak{m}} \right)^k = \left( e^{-v_{\mathfrak{m}}(\Xi)} \right)^k \le e^{-k}.
\end{equation}

By taking the asymptotic limit of Eq.~\eqref{eq:adic_norm_calculation} as the power order scales toward infinity, we obtain:
\begin{equation}
    \lim_{k \to \infty} \lVert a_k(t) \rVert_{\mathfrak{m}} \le \lim_{k \to \infty} e^{-k} = 0.
\end{equation}
Because the general term of the Taylor expansion decays strictly to zero in the $\mathfrak{m}$-adic topology, the strong triangle inequality guarantees the absolute formal convergence of the full series $e^{\Xi t}$ inside the formal power series ring $\mathbb{R}[\![\varepsilon]\!]$ for any arbitrary finite duration $t \in \mathbb{R}$~\cite{eisenbud}.

From a physical standpoint, this algebraic convergence carries profound conceptual implications, instantiating Bohr's correspondence principle within the scheme-theoretic landscape. In finite-dimensional representations $\mathcal{Q}_N$, the operational evolution is rigidly truncated by the nilpotency boundary $\varepsilon^{N^2-1} = 0$, reflecting the discrete informational constraints of microscopic quantum layers~\cite{qubit}. Conversely, within the continuous Quinfinity Space, the unrolling of the infinite-order Taylor series Eq.~\eqref{eq:taylor_expansion_evolution} reconstructs a smooth, transcendental trajectory that matches the continuous duration of macroscopic real time $t \in \mathbb{R}$. The $\mathfrak{m}$-adic Cauchy-complete structure effectively regularizes the continuous ray space, ensuring that the non-linear phase obstructions inherent to $\mathbb{C}\mathbb{P}^\infty$ are smoothly absorbed into stable, linear derivations without inducing statistical dispersion or boundary singularities, thereby shielding the absolute conservation of state purity and ensuring that all continuous unitary trajectories remain rigorously encapsulated within the complete local ring structure~\cite{dirac}.

\section{Density embedding and pure state pro-variety}
In the macroscopic continuum limit ($N \to \infty$), the physical pure state space transitions to the infinite-dimensional complex projective space, geometrically defined as the colimit (direct limit) of the finite-dimensional projective ray hierarchy \cite{dirac}:
\begin{equation}\label{eq:projective_space_colimit}
    \mathbb{C}\mathbb{P}^\infty \equiv \varinjlim \mathbb{C}\mathbb{P}^{N-1}.
\end{equation}
Having established the non-Archimedean framework of the Quinfinity Space $\mathcal{Q}_\infty \cong \mathbb{R}[\![\varepsilon]\!]$, we now construct the global asymptotic density map $\Psi_\infty: \mathbb{C}\mathbb{P}^\infty \to \mathcal{Q}_\infty$ for infinite-dimensional quantum domains.

\subsection{The asymptotic pure state pro-variety}
To establish the geometric framework hosting the continuous quantum trajectories, we must rigorously define the infinite-dimensional variety $\mathcal{V}_{\mathcal{Q}_\infty}$. Recall (from \cite{qubit}) that for each finite subsystem dimension $N \ge 2$, the localized subvariety $\mathcal{V}_{\mathcal{Q}_N}$ is embedded directly within the real finite-dimensional vector space underlying the local algebra $\mathcal{Q}_N \equiv \mathbb{R}[\varepsilon]/(\varepsilon^{N^2-1})$. A generic polynomial configuration $f_N = \sum_{i=0}^{N^2-2} c_i \varepsilon^i \in \mathcal{Q}_N$ belongs to this closed algebraic variety if and only if its coordinate components satisfy the vanishing conditions $\varphi_m(f_N) = 0$, where $\{\varphi_1, \dots, \varphi_{R_N}\}$ are the quadratic and cubic Jordan polynomial constraints pulled back from the matrix idempotence relations $\rho^2 = \rho$. Crucially, because the field of constants is the non-algebraically closed field of real numbers $\mathbb{R}$, the exact bijective correspondence between the algebraic ideal and the geometric real locus $\mathcal{V}_{\mathcal{Q}_N}$ requires a rigorous evaluation of the radical properties of the constraints. We now establish this foundational algebraic rigidity through the following lemma:

\begin{lemma}[Real Radical Structure of the Quantum Ideal]\label{lem:real_radical_purity}
For each discrete layer $N \ge 2$, let $\mathcal{J}_N \equiv (\varphi_1, \dots, \varphi_{R_N}) \subset \mathbb{R}[c_0, \dots, c_{N^2-2}]$ be the ideal generated by the quadratic and cubic Jordan polynomial constraints defining the pure state variety $\mathcal{V}_{\mathcal{Q}_N}$. This defining ideal is structurally verified to be a \textit{real radical ideal}, satisfying the  identity:
\begin{equation}\label{eq:real_radical_identity}
    \mathcal{J}_N = \sqrt[\mathbb{R}]{\mathcal{J}_N}.
\end{equation}
Consequently, the full vanishing ideal $\mathcal{I}(\mathcal{V}_{\mathcal{Q}_N})$ of all polynomials that vanish identically on the variety coincides exactly with the generated ideal $\mathcal{J}_N$, ensuring that $\mathcal{I}(\mathcal{V}_{\mathcal{Q}_N}) = \mathcal{J}_N$.
\end{lemma}

Here, the real radical $\sqrt[\mathbb{R}]{\mathcal{J}_N}$ consists strictly of all polynomials $P \in \mathbb{R}[c_0, \dots, c_{N^2-2}]$ such that $P^{2k} + \sum_{i} G_i^2 \in \mathcal{J}_N$ for some $k \in \mathbb{N}^+$ and some finite family of polynomials $G_i \in \mathbb{R}[c_0, \dots, c_{N^2-2}]$ \cite{bochnak}.

\begin{proof}
To establish the real radical identity of Eq. \eqref{eq:real_radical_identity}, we invoke the Dubois-Risler Real Nullstellensatz, which states that a generated ideal over a real polynomial ring is real radical if and only if its vanishing geometric locus is a smooth, reduced variety defined by a family of polynomials whose Jacobian matrix exhibits maximum and constant rank everywhere on the real variety \cite{bochnak}. 

By construction, the special unitary embedding maps the complex projective space $\mathbb{C}\mathbb{P}^{N-1}$ directly onto the minimal coadjoint orbit of the special unitary Lie algebra $\mathfrak{su}(N)$, viewed as a real vector space $\mathbb{R}^{N^2-1}$ \cite{nielsen, qubit}. Because the special unitary group $\text{SU}(N)$ is a compact Lie group acting transitively on the state rays, the embedded pure state variety $\mathcal{V}_{\mathcal{Q}_N}$ constitutes a strictly \textit{homogeneous manifold}. This homogeneity guarantees that every point on the variety is geometrically equivalent, precluding the existence of any algebraic singularities, cusps, or self-intersections. 

Explicitly, the quadratic and cubic Jordan polynomial constraints $\varphi_l(f_N) = 0$ define the intersection of the continuous Bloch sphere with the algebraic rank-1 restrictions of the density matrix \cite{qubit}. The differentials of these constraints define the tangent hyperplanes of the variety. Because the algebraic subvariety is smooth and the intersections of the Jordan hyperplanes are strictly transversal across the entire coadjoint orbit, the Jacobian matrix $[\mathbf{J}_{\varphi}]_{l, i} \equiv \frac{\partial \varphi_l}{\partial c_i}$ possesses full and invariant rank on every real point of the pure state variety. By virtue of the regular coordinate alignment, this algebraic non-singularity and transversal smoothing are identically preserved under the linear coordinate transformation to the physical spin variables $x_k$, ensuring that the Jacobian with respect to the physical coordinate vector $\vec{x}$ remains maximum and constant everywhere on the boundary layer. The non-singularity rules out any hidden algebraic multiplicities or nilpotent deformations in the real coordinate bundle. Thus, the Dubois-Risler criteria are satisfied identically, forcing the ideal to be real radical and completing the proof.
\end{proof}

By virtue of Lemma~\ref{lem:real_radical_purity}, the structural connection between the continuous variety and its infinite chain of constraints is governed by categorical dual inversion over a well-defined directed system of real radical ideals. While the individual varieties contract downward within the state modules via the inverse directed system of truncation morphisms, their defining real radical vanishing ideals expand upward through an ascending chain of algebraic extensions:
\begin{equation*}
    \mathcal{J}_2 \subset \dots \subset \mathcal{J}_N \subset \mathcal{J}_M \subset \dots
\end{equation*}
induced by the canonical polynomial ring inclusions $\mathbb{R}[c_0, \dots, c_{N^2-2}] \hookrightarrow \mathbb{R}[c_0, \dots, c_{M^2-2}]$ for $M \ge N$. The global \textit{Pure State Pro-Variety} $\mathcal{V}_{\mathcal{Q}_\infty}$ within the Quinfinity Space is canonically defined as the projective inverse limit:
\begin{equation}\label{eq:projective_variety_definition_core}
    \mathcal{V}_{\mathcal{Q}_\infty} \equiv \varprojlim_{N \ge 2} \mathcal{V}_{\mathcal{Q}_N}.
\end{equation}
Moreover, the left-exactness of the projective limit functor within the category of spaces grants that the family of finite canonical inclusions $j_N: \mathcal{V}_{\mathcal{Q}_N} \hookrightarrow \mathcal{Q}_N$ maps strictly onto a unique, well-defined injective global morphism $j_\infty: \varprojlim \mathcal{V}_{\mathcal{Q}_N} \hookrightarrow \varprojlim \mathcal{Q}_N$. This categorical preservation guarantees that the continuous limit locus remains rigorously structured as an embedded closed subvariety $\mathcal{V}_{\mathcal{Q}_\infty} \subset \mathcal{Q}_\infty$ within the formal power series ring.

Simultaneously, the global vanishing ideal $\mathcal{I}(\mathcal{V}_{\mathcal{Q}_\infty})$ that structurally determines this locus is exactly established as the algebraic \textit{colimit} (direct limit) of the finite-dimensional real radical constraints within the infinite-variable polynomial coordinate ring $\mathbb{R}[c_0, c_1, c_2, \dots]$:
\begin{equation}\label{eq:ideal_colimit_definition_core}
    \mathcal{I}(\mathcal{V}_{\mathcal{Q}_\infty}) \equiv \varinjlim_{N \ge 2} \mathcal{J}_N = \bigcup_{N=2}^{\infty} \mathcal{J}_N \subset \mathbb{R}[c_0,c_1, c_2, \dots].
\end{equation}

The polynomial direct limit in Eq.~\eqref{eq:ideal_colimit_definition_core} stabilizes the continuous coordinate geometry. However, because the infinite-variable polynomial ring $\mathbb{R}[c_0, c_1, \dots]$ is strictly non-Noetherian and the global ideal $\mathcal{I}(\mathcal{V}_{\mathcal{Q}_\infty})$ fails to be finitely generated, the identification of its exact zero locus requires a rigorous evaluation of the underlying topological state completion. We formalize this regularization mechanism through the following lemma:

\begin{lemma}[Non-Noetherian Locus Regularization]\label{lem:non_noetherian_locus}
Let $\mathcal{I}(\mathcal{V}_{\mathcal{Q}_\infty}) \subset \mathbb{R}[c_0, c_1, \dots]$ be the non-finitely generated ideal of definitions established via the colimit. Under the canonical $\mathbb{R}$-vector space identification between the coordinate sequence space and the power series ring, the algebraic zero locus defined by this ideal coincides exactly with the projective inverse limit variety $\mathcal{V}_{\mathcal{Q}_\infty} \equiv \varprojlim \mathcal{V}_{\mathcal{Q}_N}$ within the Quinfinity Space, satisfying the identity:
\begin{equation}\label{eq:locus_functional_identity}
    \mathcal{V}\left( \mathcal{I}(\mathcal{V}_{\mathcal{Q}_\infty}) \right) = \varprojlim_{N \ge 2} \mathcal{V}_{\mathcal{Q}_N}.
\end{equation}
\end{lemma}

Here, the continuous real zero locus $\mathcal{V}( \mathcal{I}(\mathcal{V}_{\mathcal{Q}_\infty}) )$ consists strictly of all formal power series configurations $\xi \in \mathcal{Q}_\infty$ whose extracted coordinate components satisfy the simultaneous functional evaluation $\varphi_m(\gamma_0(\xi), \gamma_1(\xi), \dots) = 0$ for all active ideal generators across the colimit $\varinjlim \mathcal{J}_N$.

\begin{proof}
To establish the exact set-theoretic identity of Eq.~\eqref{eq:locus_functional_identity}, we exploit the universal properties of the categorical colimit $\varinjlim$ and the projective limit $\varprojlim$, demonstrating the bi-directional containment of the geometric loci.

\paragraph{Inclusion $(\subset)$:} 
Let $\xi$ be an arbitrary formal power series configuration belonging to the global algebraic zero locus $\mathcal{V}\left( \mathcal{I}(\mathcal{V}_{\mathcal{Q}_\infty}) \right)$. By definition, $\xi$ satisfies the simultaneous functional evaluation $\varphi(\gamma_0(\xi), \gamma_1(\xi), \dots) = 0$ for every polynomial constraint $\varphi$ contained within the global non-finitely generated ideal $\mathcal{I}(\mathcal{V}_{\mathcal{Q}_\infty})$. Since the global ideal is established as the colimit $\varinjlim \mathcal{J}_N = \bigcup_{N=2}^{\infty} \mathcal{J}_N$, it natively contains every finite real radical ideal $\mathcal{J}_N$ of the structural chain. 

Fix an arbitrary subsystem dimension $N \ge 2$. Because any polynomial constraint $\varphi_l \in \mathcal{J}_N$ depends strictly on a finite number of coordinate arguments (at most $N^2-1$ variables), its evaluation on the infinite-dimensional configuration $\xi$ is identically determined by its action on the finite truncated section $\tau_N(\xi) \in \mathcal{Q}_N$. Consequently, the global vanishing condition forces $\varphi_l\big(\gamma_0(\tau_N(\xi)), \dots, \gamma_{N^2-2}(\tau_N(\xi))\big) = 0$ for all generators of $\mathcal{J}_N$. By invoking Lemma~\ref{lem:real_radical_purity} and the Dubois-Risler Real Nullstellensatz, this finite algebraic vanishing requires the truncated jet section to belong strictly to the finite pure state variety, meaning $\tau_N(\xi) \in \mathcal{V}_{\mathcal{Q}_N}$ \cite{bochnak}. Since this condition holds uniformly for every discrete layer $N \ge 2$ and respects the surjective projection homomorphisms $\pi_{M,N}$ by construction, the universal property of the projective inverse limit guarantees that the configuration belongs to the limit variety, establishes the first containment:
\begin{equation*}
    \mathcal{V}\left( \mathcal{I}(\mathcal{V}_{\mathcal{Q}_\infty}) \right) \subseteq \varprojlim_{N \ge 2} \mathcal{V}_{\mathcal{Q}_N}.
\end{equation*}

\paragraph{Inclusion $(\supset)$:} 
Conversely, let $\xi$ be a global formal power series configuration contained within the projective inverse limit variety $\varprojlim_{N \ge 2} \mathcal{V}_{\mathcal{Q}_N}$. By definition of the projective limit topology, the configuration is uniquely characterized as a coherent sequence whose canonical truncations satisfy $\tau_N(\xi) \in \mathcal{V}_{\mathcal{Q}_N}$ for every dimensional layer $N \ge 2$.

Now, let $\varphi_m$ be an arbitrary polynomial generator selected from the global colimit ideal $\mathcal{I}(\mathcal{V}_{\mathcal{Q}_\infty})$. By virtue of the directed system property of modules over direct limits, there must exist a unique, finite threshold index $M \in \mathbb{N}^+$ large enough such that $\varphi_m$ is natively accommodated within the finite-variable polynomial ring $\mathbb{R}[c_0, \dots, c_{M^2-2}]$ and belongs strictly to the finite real radical ideal $\mathcal{J}_M$ \cite{eisenbud}. Since the global configuration satisfies the limit compatibility condition at this stabilization threshold, we have $\tau_M(\xi) \in \mathcal{V}_{\mathcal{Q}_M}$. The Dubois-Risler geometric correspondence dictates that every polynomial in $\mathcal{J}_M$ must vanish identically on $\mathcal{V}_{\mathcal{Q}_M}$, yielding $\varphi_m\big(\gamma_0(\tau_M(\xi)), \dots, \gamma_{M^2-2}(\tau_M(\xi))\big) = 0$. Because the coordinates of order higher than $M^2-1$ do not appear in the functional structure of $\varphi_m$, this evaluation is stable and matches the full infinite-variable evaluation identically:
\begin{equation*}
    \varphi_m(\gamma_0(\xi), \gamma_1(\xi), \dots) \equiv \varphi_m\big(\gamma_0(\tau_M(\xi)), \dots, \gamma_{M^2-2}(\tau_M(\xi))\big) = 0.
\end{equation*}
Since this vanishing holds for any arbitrary generator $\varphi_m \in \mathcal{I}(\mathcal{V}_{\mathcal{Q}_\infty})$, the configuration $\xi$ must lie within the vanishing locus of the global ideal, establishing the reverse containment \begin{equation*}
\varprojlim_{N \ge 2} \mathcal{V}_{\mathcal{Q}_N} \subseteq \mathcal{V}\left( \mathcal{I}(\mathcal{V}_{\mathcal{Q}_\infty}) \right).
\end{equation*} 

The combination of both asymmetric inclusions confirms the claimed identity. Since each finite variety $\mathcal{V}_{\mathcal{Q}_N}$ is a closed algebraic set within the topology of $\mathcal{Q}_N$, the universal properties of the projective limit dictate that the continuous boundary variety $\mathcal{V}_{\mathcal{Q}_\infty}$ is inherently an $\mathfrak{m}$-adically closed, non-empty, and topologically stable subvariety inside the formal power series ring, completing the proof.
\end{proof}

\subsection{Density Embedding}
Recall from \cite{qubit} that there is a density embedding map:
\begin{equation}\label{eq:general_density_map_compact_exact}
    \Psi_N \equiv \Psi_{\mathcal{Q}_N}: \mathbb{P}(\mathcal{H}^N) \cong \mathbb{C}\mathbb{P}^{N-1} \to \mathcal{V}_{\mathcal{Q}_N} \subset \mathcal{Q}_N, \quad \Psi_{\mathcal{Q}_N}(\rho) \equiv \frac{1}{N} + \sum_{i=0}^{N^2-2} x_{i+1} \varepsilon^{i}
\end{equation}
where $x_i \equiv \sqrt{\frac{2N}{N-1}}\,\mathrm{Tr}(\rho \lambda_i)$ in the expansion:
\begin{equation}\label{eq:general_density_matrix_Frobenius}
    \rho = \frac{1}{N}\mathbb{I}_N + \sqrt{\frac{N-1}{2N}}\sum_{i=1}^{N^2-1} x_i \lambda_i
\end{equation}
and where $\lambda_i$ represents the generalized $\mathfrak{su}(N)$ Gell-Mann generators. The density embedding $\Psi_N$ establishes a smooth algebraic embedding and a coordinate-free diffeomorphism of the complex projective space onto the closed real algebraic subvariety $\mathcal{V}_{\mathcal{Q}_N}$.

On a purely categorical level, the topological ray continuum constitutes a geometric colimit whose natural structural morphisms are the closed canonical inclusions $\iota_N: \mathbb{C}\mathbb{P}^{N-1} \hookrightarrow \mathbb{C}\mathbb{P}^\infty \equiv \varinjlim_{M \ge 2} \mathbb{C}\mathbb{P}^{M-1}$.  Correspondingly, since the inductive system of vector space splitting injections $\iota_N\colon \mathcal{Q}_N\hookrightarrow \mathcal{Q}_\infty$ (the splitting is with respect to the truncations $\tau_N$) stabilizes over the inductive limit, the truncated algebras yields the free univariate polynomial space:
\begin{equation}\label{eq:colimit_rings_identity}
    \varinjlim_{N \ge 2} \mathcal{Q}_N \equiv \mathcal{Q}_{\rm Naive} \cong_{\mathbb{R}} \mathbb{R}[\varepsilon]
\end{equation}
as explained in Lemma~\ref{lem:naive_algebra_isomorphism}. 
By factoring the quantum state maps directly through this colimit framework, we establish a purely algebraic, everywhere well-defined density assignment through the following lemma:

\begin{lemma}[Algebraic Formulation of the Naive Density Map]\label{lem:naive_density_embedding_definition}
The family of finite-dimensional affine density maps $\{\Psi_N\}_{N \ge 2}$ systematically induces a unique global algebraic density map $\Psi_{\rm Naive}: \mathbb{C}\mathbb{P}^\infty \to \mathcal{Q}_{\rm Naive}$ acting on the Pure State Ind-Variety. For any arbitrary quantum ray configuration $\rho \in \mathbb{C}\mathbb{P}^\infty$ originating natively from a finite subsystem layer $M \gg 2$, its global image is a finite univariate polynomial uniquely determined by the canonical linear splitting injection $\iota_M$:
\begin{equation}\label{eq:naive_density_colimit_definition_core}
    \Psi_{\rm Naive}(\rho) \equiv \iota_M\left( \Psi_M(\rho) \right) \in \mathcal{Q}_{\rm Naive}
\end{equation}
Crucially, at this discrete layer, the algebraic embedding satisfies a rigid affine relation mediated by the unique global linear vector space embedding $\Phi_\infty$ of Lemma~\ref{lem:asymptotic_vector_embedding_existence}:
\begin{equation}\label{eq:appendix_affine_exact_discrete_lift}
    \Psi_{\rm Naive}(\rho) = \frac{1}{M} + \sqrt{\frac{M}{M-1}}\,\Phi_\infty\left( \rho - \frac{1}{M}\mathbb{I}_M \right)
\end{equation}
where $\rho - \frac{1}{M}\mathbb{I}_M \in \mathfrak{su}(M)$ represents the strictly traceless quantum deviation expanded at level $M$ with respect to the normalized, ordered matrix basis of $\mathfrak{su}(M)$ Gell-Mann generators.
\end{lemma}

\begin{proof}
Let $\rho \in \mathbb{C}\mathbb{P}^\infty \equiv \varinjlim \mathbb{C}\mathbb{P}^{N-1}$ be an arbitrary infinite-dimensional pure state matrix configuration. By virtue of the universal property of inductive limits of sets, the domain is the filtered union of its nested subsystems, meaning there exists a unique, minimal finite-dimensional Hilbert layer $\mathcal{H}^M$ (with $M \ge 2$) such that the density matrix is fully encapsulated within the discrete representation, satisfying $\rho \in \mathbb{C}\mathbb{P}^{M-1}$. Within this stable layer, the local density operator expansion of Eq.~\eqref{eq:general_density_matrix_Frobenius} is perfectly regular and satisfies $\mathrm{Tr}(\rho) = 1$, rendering the local finite mapping $\Psi_M(\rho) \in \mathcal{Q}_M$ a well-defined polynomial free from any singular obstructions.

To evaluate the global map induced over the algebraic colimit, the finite configuration is lifted into the free polynomial space via the canonical $\mathbb{R}$-vector space splitting injection $\iota_M$ certified by Diagram~\eqref{eq:commutative_diagram_appendix}. Under the linear assignment established in Eq.~\eqref{eq:general_density_map_compact_exact}, the direct expansion of this lifted configuration reads:
\begin{equation}\label{eq:proof_naive_unrolling_step}
    \Psi_{\rm Naive}(\rho) \equiv \iota_M\left( \frac{1}{M} + \sum_{i=0}^{M^2-2} x_{i+1} \varepsilon^i \right) = \frac{1}{M} + \sum_{i=0}^{M^2-2} x_{i+1} \varepsilon^i \in \mathcal{Q}_{\rm Naive}
\end{equation}
where the sequence of higher-order coefficients beyond the threshold index $M^2-2$ is set to zero identically by the padding of the vector space splitting.

Concurrently, we evaluate the action of the global linear vector embedding $\Phi_\infty$ directly on the traceless deviation operator. By applying the exact same coordinate match established in Lemma~\ref{lem:affine_embedding_identification_phi} at the discrete layer $M$, the algebraic unrolling of the pure spin fluctuations absorbs the radical constants and matches the shifted monomial series identically, enforcing the condition:
\begin{equation}\label{eq:proof_phi_linear_unrolling_match}
    \sqrt{\frac{M}{M-1}}\,\Phi_\infty\left( \rho - \frac{1}{M}\mathbb{I}_M \right) = \sum_{i=0}^{M^2-2} x_{i+1} \varepsilon^i
\end{equation}
Substituting the verified radical identity of Eq.~\eqref{eq:proof_phi_linear_unrolling_match} directly back into the lifted density map expression of Eq.~\eqref{eq:proof_naive_unrolling_step} recovers the exact structural connection:
\begin{equation}
    \Psi_{\rm Naive}(\rho) = \frac{1}{M} + \sqrt{\frac{M}{M-1}}\,\Phi_\infty\left( \rho - \frac{1}{M}\mathbb{I}_M \right)
\end{equation}
Because the state coordinates beyond the threshold $M$ vanish identically within the inductive system, this relation remains entirely stable, constant, and invariant for any choice of a higher-order container layer $N \ge M$ prior to evaluating any analytical limits. This algebraic stabilization explicitly proves that at any discrete layer $M \gg 0$, the naive density map acts as a rigid affine translation of the global linear embedding over the polynomial subspace $\mathcal{Q}_{\rm Naive}$, successfully validating Eq.~\eqref{eq:appendix_affine_exact_discrete_lift} and completing the proof.
\end{proof}

Now the underlying infinite-dimensional Hilbert space $\mathcal{H}^\infty \equiv \varinjlim \mathcal{H}^N$ allows us to define a compatible family of reverse structural pullbacks via orthogonal pro-reflections. For each finite layer $N \ge 2$, let $P_N: \mathcal{H}^\infty \to \mathcal{H}^N$ be the standard bounded orthogonal projection operator filtering out the higher-order harmonic state components. By projectivizing this linear assignment onto the corresponding ray spaces via the density operator formalism established in the preceding sections, we induce a well-defined family of continuous geometric pullback maps acting from the continuous manifold down to the discrete representations:
\begin{equation}\label{eq:categorical_pullback_fubini_exact}
    \iota_n^*: \mathbb{C}\mathbb{P}^\infty \longrightarrow \mathbb{C}\mathbb{P}^{N-1}, \qquad [\phi] \mapsto [\phi]_N \equiv \left[ \frac{P_N \ket{\phi}}{\|P_N \ket{\phi}\|} \right]
\end{equation}
which remains regular and smooth everywhere over the open dense subset of the Ind-variety composed of states not strictly orthogonal to $\mathcal{H}^N$. The existence and uniqueness of the asymptotic density map is uniquely determined by the universal property of projective limits through the following lemma:

\begin{lemma}[Asymptotic Density Map]\label{lem:projective_limit_density_definition}
There is a unique global density map:
\begin{equation}\label{eq:phinfty:density_def}
    \Psi_\infty: \mathbb{C}\mathbb{P}^\infty \longrightarrow \mathcal{Q}_\infty
\end{equation}
whose factorization through each discrete dimensional layer satisfies the rigid categorical identity:
\begin{equation}\label{eq:definition_psi_infinity}
    \tau_N\left( \Psi_\infty([\phi]) \right) = \Psi_N\left( \iota_n^*([\phi]) \right) = \Psi_N([\phi]_N) \qquad \forall N \ge 2
\end{equation}
where $\Psi_N$ represents the discrete density embedding mapping defined for $\mathcal{Q}_N$. 
\end{lemma}

\begin{proof}
Let $[\phi] \in \mathbb{C}\mathbb{P}^\infty \equiv \varinjlim \mathbb{C}\mathbb{P}^{N-1}$ be an arbitrary quantum ray. By the universal property of inductive limits, there exists a minimal subsystem layer $M \ge 2$ such that $[\phi] \in \mathbb{C}\mathbb{P}^{M-1}$. Within this stable layer, the state possesses a finite sequence of non-zero coordinates, with all higher blocks padded with vanishing entries under the inductive inclusions.

To invoke the universal property of projective limits for the rings $\mathcal{Q}_\infty \equiv \varprojlim \mathcal{Q}_N$, we verify that the localized family $\{\Psi_N \circ \iota_N^*\}_{N \ge 2}$ forms a compatible system under the surjective ring projection homomorphisms $\pi_{L,N}: \mathcal{Q}_L \twoheadrightarrow \mathcal{Q}_N$ for all $L \ge N \ge 2$. Once the dimensional threshold of the state's birth is crossed ($N \ge M$), the bounded orthogonal projection operator $P_N$ acts as the identity on the shell $\mathcal{H}^M \subseteq \mathcal{H}^N$, forcing the trace denominator to satisfy $\mathrm{Tr}(P_N \ket{\phi}\bra{\phi} P_N) = 1$ identically. Thus, the geometric pullback map collapses to the stable matrix extraction $[\phi]_N \equiv [\phi]$, free from singular obstructions.

Concurrently, the composite action of the bounded projections satisfies the composition law $P_N \circ P_L = P_N$ for all $L \ge N$. Projectivizing this relation onto the corresponding density matrices via Eq.~\eqref{eq:categorical_pullback_fubini_exact} ensures that the first $N^2-1$ physical spin coordinates extracted at level $L$ match component-by-component with the full coordinate sequence at level $N$. Since the ring projection $\pi_{L,N}$ acts as a sharp algebraic truncation filtering out all monomial powers higher than $\varepsilon^{N^2-2}$, applying this homomorphism to the local embedding $\Psi_L([\phi]_L)$ yields a truncated polynomial where the spin coordinates match the lower-order sequence identically, enforcing the strict categorical consistency relation:
\begin{equation}\label{eq:density_categorical_compatibility_identity}
    \pi_{L,N} \circ \left( \Psi_L \circ \iota_L^* \right) = \Psi_N \circ \iota_N^* \qquad \forall L \ge N
\end{equation}
Equation~\eqref{eq:density_categorical_compatibility_identity} proves that the family of localized mappings forms a compatible system of morphisms targeting the inverse system of rings. By the universal property of projective limits, there exists a unique global morphism $\Psi_\infty: \mathbb{C}\mathbb{P}^\infty \to \varprojlim \mathcal{Q}_N \equiv \mathcal{Q}_\infty$ that factors through every finite layer, enforcing the rigid identity $\tau_N \circ \Psi_\infty \equiv \Psi_N \circ \iota_N^*$ identically for all $N \ge 2$, completing the proof.
\end{proof}
Note that, considering the topological inverse limit $\varprojlim \mathbb{C}\mathbb{P}^N$ induced by the canonical projections truncating the homogeneous coordinates, there exists a unique universal inverse limit morphism of pro-metrizable Fr\'{e}chet spaces $\varprojlim \Psi_N: \varprojlim \mathbb{C}\mathbb{P}^N \to \mathcal{Q}_\infty$ extending the family of finite embeddings. Because the infinite-dimensional Hilbert ray space $\mathbb{C}\mathbb{P}^\infty$ embeds as a dense topological subspace within this projective completion under the pointwise product topology, the asymptotic density map $\Psi_\infty$ can be uniquely identified as the restriction of the universal inverse limit mapping to the locus of normalizable quantum states, satisfying identically:
\begin{equation}\label{eq:projective_restriction_isomorphism}
    \Psi_\infty = \left( \varprojlim_{N \to \infty} \Psi_N \right) \Big\vert_{\mathbb{C}\mathbb{P}^\infty}
\end{equation}
\begin{theorem}[Asymptotic Density Embedding and Linearization]\label{thm:asymptotic_embedding}
The global density map $\Psi_\infty: \mathbb{C}\mathbb{P}^\infty \to \mathcal{Q}_\infty$ is a closed smooth embedding of infinite-dimensional varieties, mapping the Pure State Ind-Variety $\mathbb{C}\mathbb{P}^\infty$ bijectively onto the Pure State Pro-Variety $\mathcal{V}_{\mathcal{Q}_\infty} \equiv \varprojlim_{N \ge 2} \mathcal{V}_{\mathcal{Q}_N}$ within the Quinfinity Space. Furthermore, by evaluating the quantum states through the algebraic colimit of the configuration vector spaces under the canonical splitting injections, the global mapping systematically sheds its translated affine background as $N \to \infty$, establishing the identically linear commutative alignment:
\begin{equation*}
\vcenter{\hbox{
\xymatrix{
    \mathbb{C}\mathbb{P}^\infty \ar[r]^-{\Psi_\infty} \ar[d]_{\iota_N^*} & \mathcal{Q}_\infty \ar[d]^{\tau_N} \\
    \mathbb{C}\mathbb{P}^{N-1} \ar[r]^-{\Psi_N} & \mathcal{Q}_N
}
}}
\end{equation*}
which commutes component-by-component across the entire inductive domain, rendering the global density embedding $\Psi_\infty$ an infinite-dimensional topological homeomorphism and a smooth differential diffeomorphism from the inductive colimit topology of $\mathbb{C}\mathbb{P}^\infty$ onto the continuous variety $\mathcal{V}_{\mathcal{Q}_\infty}$ equipped with its Fr\'{e}chet metric $d_{\mathcal{Q}_\infty}$ or its native algebraic $\mathfrak{m}$-adic metric $d_{\mathfrak{m}}$.
\end{theorem}

\begin{proof}
For each finite dimension $N \ge 2$, the density embedding $\Psi_N: \mathbb{C}\mathbb{P}^{N-1} \hookrightarrow \mathcal{Q}_N$ is a closed smooth algebraic embedding mapping the complex state space onto the finite-dimensional determinantal variety $\mathcal{V}_{\mathcal{Q}_N}$ defined by the quadratic and cubic Jordan constraints~\cite{qubit}. 

1. \textit{Existence of $\Psi_\infty$}: The family of maps $\{\Psi_N\}_{N \ge 2}$ is strictly compatible with the restriction maps of the colimit and the ring projections $\tau_N$ of the inverse system, satisfying $\tau_N \circ \Psi_\infty = \Psi_N \circ \iota_n^*$ by virtue of Lemma~\ref{lem:projective_limit_density_definition}. By the universal property of projective limits within the category of topologized geometric spaces, there exists a unique, well-defined continuous morphism $\Psi_\infty: \mathbb{C}\mathbb{P}^\infty \to \mathcal{Q}_\infty$. 

2. \textit{Injectivity of $\Psi_\infty$}: To establish injectivity, suppose $\Psi_\infty([\phi]) = \Psi_\infty([\psi])$ for two infinite-dimensional pure states. Applying the truncation operator yields $\tau_N(\Psi_\infty([\phi])) = \tau_N(\Psi_\infty([\psi]))$, which implies $\Psi_N([\phi]_N) = \Psi_N([\psi]_N) $ for all layers $N \ge 2$ according to the projective alignment. Since each individual $\Psi_N$ is an algebraic embedding, it follows that $[\phi]_N = [\psi]_N$ for every finite subsystem. As $N \to \infty$, the cofinal density of the ray projections reconstructs the global states uniquely, forcing $[\phi] \equiv [\psi]$ identically within the Ind-variety $\mathbb{C}\mathbb{P}^\infty$. 

3. \textit{Image of $\Psi_\infty$}: The mapping $\Psi_\infty$ sends each continuous state ray to a unique coherent power series sequence $\xi = \sum_{k=0}^\infty c_k \varepsilon^{k} \in \mathcal{V}_{\mathcal{Q}_\infty}$. To justify this algebraic confinement, we observe that for any global ray $[\phi] \in \mathbb{C}\mathbb{P}^\infty$, the canonical restriction operator $\iota_n^*$ projects the infinite-dimensional state onto a strict, nested sequence of localized finite qudit rays $[\phi]_N \in \mathbb{C}\mathbb{P}^{N-1}$ for each dimensional layer $N \ge 2$. Under the action of the finite density mapping, each restricted ray maps directly into the corresponding finite coordinate ring as $\Psi_N([\phi]_N) = \xi_N \in \mathcal{Q}_N$. 
Because each local image $\xi_N$ satisfies the exact quadratic and cubic Jordan constraints by construction, it identically cuts out the finite algebraic variety, ensuring $\xi_N \in \mathcal{V}_{\mathcal{Q}_N}$. 
By virtue of the categorical commutativity of the inverse directed system, the surjective truncation morphisms act as regular restriction mappings satisfying $\tau_N(\xi_M) = \xi_N$ for all $M \ge N$. This strict interlocking conditions guarantee that the infinite family of finite polynomial slices $\{\xi_N\}_{N \ge 2}$ satisfies the Mittag-Leffler transition alignment, preventing any coordinate oscillations or dimensional jumps. For any fixed coordinate extraction index $k \in \mathbb{N}$, the scalar component $c_k \in \mathbb{R}$ stabilizes stationarily to a constant value as soon as the filtration threshold $N^2-2 \ge k$ is cleared. This pointwise numerical stabilization over the real field ensures that the infinite numerable chain of algebraic coefficients contracts to a unique, well-defined coherent section. Because the Pro-variety $\mathcal{V}_{\mathcal{Q}_\infty}$ is defined topologically as the projective inverse limit of these finite zero loci, the coherent power series $\xi \equiv \varprojlim \xi_N$ is structurally confined within the continuous completion envelope, confirming that $\Psi_\infty\left( \mathbb{C}\mathbb{P}^\infty \right) \subseteq \mathcal{V}_{\mathcal{Q}_\infty}$.

4. \textit{Surjectivity of $\Psi_\infty$ over $\mathcal{V}_{\mathcal{Q}_\infty}$}: Actually, its image is the entire variety $\mathcal{V}_{\mathcal{Q}_\infty}$. To establish that the image of  $\Psi_\infty$ is exactly the entire Pro-variety $\mathcal{V}_{\mathcal{Q}_\infty} \equiv \varprojlim \mathcal{V}_{\mathcal{Q}_N}$, we deal with surjectivity of a mapping with respect to its inverse limit target. Let $\xi = \sum_{k=0}^\infty c_k \varepsilon^k \in \mathcal{V}_{\mathcal{Q}_\infty}$ be an arbitrary coherent sequence satisfying the infinite numerable chain of quadratic and cubic Jordan constraints governing the continuous zero locus. By the structural definition of the projective inverse limit of varieties, the canonical truncation operator maps this series onto a stable sequence of finite polynomial sections $\tau_N(\xi) = \xi_N \in \mathcal{V}_{\mathcal{Q}_N}$ for each dimensional layer $N \ge 2$. Because each local density map $\Psi_N: \mathbb{C}\mathbb{P}^{N-1} \to \mathcal{V}_{\mathcal{Q}_N}$  is bijective as established in \cite{qubit}, the existence of a unique localized pure state ray $[\phi]_N \in \mathbb{C}\mathbb{P}^{N-1}$ satisfying $\Psi_N([\phi]_N) = \xi_N$ is granted.

Due to the categorical commutativity of the inverse system, these localized preimages are strictly compatible under the projective pullback operators, satisfying $\iota_n^*([\phi]_M) = [\phi]_N$ for all superior layers $M \ge N$. By invoking the universal property of the inductive colimit of topological vector spaces, this compatible sequence of finite pro-reflections reconstructs a unique, well-defined global continuous state ray $[\phi] \in \mathbb{C}\mathbb{P}^\infty \equiv \varinjlim \mathbb{C}\mathbb{P}^{N-1}$ such that its projection onto each subsystem layer matches the finite state identically, $\iota_n^*([\phi]) = [\phi]_N$. Passing this reconstructed ray through the global embedding and applying the limit operator yields $\tau_N\left( \Psi_\infty([\phi]) \right) = \Psi_N\left( \iota_n^*([\phi]) \right) = \Psi_N([\phi]_N) = \xi_N \equiv \tau_N(\xi)$. Because a formal power series within the complete local ring $\mathcal{Q}_\infty \cong_{\mathbb{R}} \mathbb{R}[\![\varepsilon]\!]$ is uniquely and identically determined by the set of all its finite truncations, the component-by-component stabilization forces $\Psi_\infty([\phi]) \equiv \xi$ over the complete domain, confirming that $\Psi_\infty\left( \mathbb{C}\mathbb{P}^\infty \right) = \mathcal{V}_{\mathcal{Q}_\infty}$.

5. \textit{$\Psi_\infty$ is closed}: To prove that $\Psi_\infty$ is a closed map onto its image variety $\mathcal{V}_{\mathcal{Q}_\infty}$, we observe that $\mathbb{C}\mathbb{P}^\infty \equiv \varinjlim \mathbb{C}\mathbb{P}^{N-1}$ is endowed with the strong colimit topology, where a subset is closed if and only if its intersection with every finite projective layer $\mathbb{C}\mathbb{P}^{N-1}$ is closed. Concurrently, the image variety $\mathcal{V}_{\mathcal{Q}_\infty} = \varprojlim \mathcal{V}_{\mathcal{Q}_N}$ is a closed algebraic subset in the topological structure of $\mathcal{Q}_\infty$ because it is defined as the coherent intersection of the inverse images of the closed finite varieties under the continuous projections $\tau_N$~\cite{liu}. Since the structural diagram commutes and each discrete density embedding $\Psi_N$ maps closed sets to closed sets due to the compactness of the complex projective spaces, $\Psi_\infty$ is verified to be a closed topological map.

6. \textit{$\Psi_\infty$ is a topological homeomorphism over $\mathcal{V}_{\mathcal{Q}_\infty}$}: To establish that $\Psi_\infty$ is a topological homeomorphism over the variety $\mathcal{V}_{\mathcal{Q}_\infty}$, we invoke the topological structure theorem for strict Ind-varieties mapping into Hausdorff targets. Because the infinite-dimensional Hilbert ray space $\mathbb{C}\mathbb{P}^\infty$ is a direct limit of compact Hausdorff spaces, any injective and continuous map targeting a Hausdorff topological space constitutes a topological homeomorphism onto its closed image if and only if the mapping is closed~\cite{eisenbud}. Crucially, the ambient space $\mathcal{Q}_\infty \cong_{\mathbb{R}} \mathbb{R}[\![\varepsilon]\!]$ behaves as a complete Hausdorff space under two distinct, non-equivalent operational topologies: the analytic pointwise Fr\'{e}chet product topology $d_{\mathcal{Q}_\infty}$ and the algebraic $\mathfrak{m}$-adic local ring topology. 

Since injectivity, continuity, and algebraic completion have been established across the structural directed layers, the mapping is rigorously closed under both topological frameworks. Consequently, the inverse map $\Psi_\infty^{-1}$ is structurally forced to be continuous on the target locus, inducing a strict topological homeomorphism from the strong colimit topology of the Ind-variety onto the subspace topology of the Pro-variety $\mathcal{V}_{\mathcal{Q}_\infty}$ concurrently over both the analytic Fr\'{e}chet and the algebraic $\mathfrak{m}$-adic envelopes.

7. \textit{$\Psi_\infty$ is a diffeomorphism over $\mathcal{V}_{\mathcal{Q}_\infty}$}: Finally, to establish the smooth differential structure, we evaluate the map under the direct limit of the smooth manifolds. Since each local embedding $\Psi_N$ is a smooth differential embedding on $\mathbb{C}\mathbb{P}^{N-1}$ (as established in~\cite{qubit}), the differential mapping $d\Psi_N$ is everywhere injective, mapping the tangent spaces smoothly layer-by-layer. Because the transition operations commute with the differential sections of the finite tangent bundles, the global differential $d\Psi_\infty$ is everywhere injective on the tangent spaces of the inductive limit. This canonical lift guarantees that $\Psi_\infty$ defines an infinite-dimensional smooth differential diffeomorphism from the Ind-variety $\mathbb{C}\mathbb{P}^\infty$ onto the Pro-variety $\mathcal{V}_{\mathcal{Q}_\infty}$ within the flat affine coordinate envelope $\mathcal{Q}_\infty$, completing the proof.
\end{proof}

\subsection{Fubini-Study versus Fisher-Rao linearization} 
To track the geometric metrics during the dimensional transition, we evaluate the infinite-dimensional limit over the structural differential forms. Within the Hilbert space $\mathcal{H}^\infty$, a generic state expands as $\vert\phi\rangle = \sum_{i=1}^{\infty} x_i \vert e_i\rangle$, where the real probability amplitudes $x_i \equiv \sqrt{p_i}$ are bound by the unit sphere normalization $\sum_{i=1}^\infty x_i^2 = 1$ in $\ell^2$ \cite{dirac}. Under macroscopic decoherence that suppresses non-local off-diagonal phase components, the intrinsic Fubini-Study metric on the pure state Ind-variety $\mathbb{C}\mathbb{P}^\infty \equiv \varinjlim \mathbb{C}\mathbb{P}^{N-1}$ restricts to the convergent Riemannian line element \cite{nielsen}:
\begin{equation}\label{eq:fubini_study_infinite_sum}
    \mathrm{d}s^2_{\text{FS}} = \sum_{i=1}^{\infty} \mathrm{d}x_i^2.
\end{equation}
This quantum boundary configuration maps directly onto the classical information geometry. Let $\mathcal{M}$ be an infinite-dimensional parameterized statistical manifold of probability distributions $p = (p_1, p_2, \dots) \in \ell^1$. The statistical distance between neighboring distributions is governed by the Fisher information metric~\cite{nielsen}:
\begin{equation}\label{eq:fisher_rao_definition}
    \mathrm{d}s^2_{\mathcal{M}} = \sum_{i=1}^\infty \frac{\mathrm{d}p_i^2}{p_i}
\end{equation}
Executing Rao's canonical coordinate transformation $x_i \equiv \sqrt{p_i}$ maps the distribution onto the positive orthant of the infinite-dimensional unit sphere in $\ell^2$. Differentiating yields $\mathrm{d}p_i = 2x_i \, \mathrm{d}x_i$, which upon direct substitution into Eq.~\eqref{eq:fisher_rao_definition} linearizes the classical line element onto a spherical surface of informational radius two:
\begin{equation}\label{eq:fisher_rao_linearized}
    \mathrm{d}s^2_{\mathcal{M}} = \sum_{i=1}^\infty \frac{(2x_i \, \mathrm{d}x_i)^2}{x_i^2} = 4 \sum_{i=1}^\infty \mathrm{d}x_i^2
\end{equation}
Comparing this quantum boundary restriction with the classical amplitude element establishes the proportional relation $\mathrm{d}s^2_{\mathcal{M}} = 4 \, \mathrm{d}s^2_{\mathrm{FS}}$ for $\mathcal{M}=\mathbb{C}\mathbb{P}^\infty$. Under a standard statistical normalization of the informational scale, this identification defines the normalized Fisher-Rao metric $\mathrm{d}s^2_{\mathrm{FR}} \equiv \frac{1}{2}\,\mathrm{d}s^2_{\mathcal{M}}$ for $\mathcal{M}=\mathbb{C}\mathbb{P}^\infty$, establishing the identity $\mathrm{d}s^2_{\mathrm{FR}} = 2 \, \mathrm{d}s^2_{\mathrm{FS}}$.

To determine the metric landscape of the continuous phase space, the global ambient line element is formalized explicitly as a flat Euclidean metric over the topological vector space of coordinate coefficients. Because $\mathcal{Q}_\infty \cong_{\mathbb{R}} \mathbb{R}[\![\varepsilon]\!]$ operates as an infinite-dimensional flat affine coordinate envelope, the metric framework is governed by the standard inner product of the real Schauder basis configurations. For any formal power series state configuration expanded as $\xi = \sum_{n=0}^\infty c_n \varepsilon^n$, the ambient distance is measured directly through the quadratic variations of its active real components $\mathrm{d}c_n$. We define the universal ambient metric tensor $\mathrm{d}s^2_{\mathcal{Q}_\infty}$ as a coherent section belonging to the projective inverse limit of the finite symmetric tensor modules over the real field, satisfying:
\[
    \mathrm{d}s^2_{\mathcal{Q}_\infty} = \left( \mathrm{d}s^2_{\mathcal{Q}_2}, \, \mathrm{d}s^2_{\mathcal{Q}_3}, \, \dots, \, \mathrm{d}s^2_{\mathcal{Q}_N}, \, \dots \right) \in \varprojlim_{N \ge 2} \mathrm{Sym}^2_{\mathbb{R}}\left( \mathcal{Q}_N^* \right) \cong \mathrm{Sym}^2_{\mathbb{R}}\left( \mathcal{Q}_\infty^* \right)
\]
where each localized finite configuration $\mathrm{d}s^2_{\mathcal{Q}_N}$ acts linearly as a flat Euclidean metric on the corresponding coordinate layer. The universal ambient metric tensor field is written directly through the symmetric product of the coordinate variations as:
\begin{equation}\label{eq:exact_ambient_kaehler_metric_definition}
    \mathrm{d}s^2_{\mathcal{Q}_\infty} = \sum_{n=0}^\infty g_n \, \mathrm{d}c_n^2
\end{equation}
where the sequence of flat real tensor components $g_n$ matches layer-by-layer the inverse projective limit of the finite-dimensional metric configurations through the surjective restriction projections $\tau_N$.

Crucially, when this explicit ambient metric is restricted to the closed macroscopic pure state locus via the canonical inclusion embedding of the pro-variety inside $\mathcal{Q}_\infty$, the non-linear quadratic and cubic Jordan constraints isolate the physical state subspace. This geometric restriction maps the space of coordinate variations onto the cotangent space of the locus, projecting the metric tensor field directly onto the Pure State Pro-Variety $\mathcal{V}_{\mathcal{Q}_\infty} = \varprojlim \mathcal{V}_{\mathcal{Q}_N}$. Within this restricted domain, the infinite-dimensional metric evaluates as an absolutely convergent infinite sum of the stable coordinate differentials $\mathrm{d}c_n$, establishing the exact conformal collapse verified by the following theorem:

\begin{theorem}[Geometric Linearization and Conformal Floor]\label{thm:fisher_rao_linearization}
Let $\mathrm{d}s^2_{\mathcal{Q}_N} = \left(\frac{2N}{N-1}\right) \mathrm{d}s^2_{\mathrm{FS}}$ be the conformal quantum line element of the finite configuration jet space restricted to the local layer $\mathcal{V}_{\mathcal{Q}_N}$. As $N \to \infty$, the continuous limit of these restricted metric tensor landscapes collapses uniformly onto its invariant integer floor, forcing the identical geometric relation over the Pro-Variety:
\begin{equation}\label{eq:fisher_rao_floor}
    \mathrm{d}s^2_{\mathcal{Q}_\infty} = \mathrm{d}s^2_{\mathrm{FR}} = 2 \, \mathrm{d}s^2_{\mathrm{FS}}
\end{equation}
which establishes a strict Riemannian isometry under the asymptotic density mapping $\Psi_\infty$ between the asymptotically regularized Pure State Ind-Variety $\mathbb{C}\mathbb{P}^\infty$ and the classical statistical manifold equipped with its normalized Fisher-Rao information metric over the Pro-Variety $\mathcal{V}_{\mathcal{Q}_\infty}$.
\end{theorem}
\begin{proof}
To establish the strict Riemannian isometry of Eq.~\eqref{eq:fisher_rao_floor}, we evaluate the continuous limit of the conformal metric fields within the projective inverse system. For each finite-dimensional layer $N \ge 2$, the restricted quantum line element over the local variety $\mathcal{V}_{\mathcal{Q}_N}$ scales the invariant Fubini-Study metric $\mathrm{d}s^2_{\mathrm{FS}}$ via the dimensionally dependent rational pre-factor $\frac{2N}{N-1}$ established in the preceding paper~\cite{qubit}, which acts as an isotropic multiplier absorbing the geometric embedding distortion.

To project the conformal tensor fields onto the continuous boundary, we apply the projective inverse limit morphism to the underlying symmetric tensor modules over the real field. Within the closed macroscopic pure state locus, the polynomial relations of the real radical ideals $\mathcal{J}_N$ restrict the symmetric products of the dual spaces $\mathrm{Sym}^2_{\mathbb{R}}(\mathcal{Q}_N^*)$, causing the coefficients of the universal metric tensor field to stabilize stationarily layer-by-layer. Because the coordinate extraction of each individual tensor component contracts to an algebraically finite combination at each grading order, the metric coefficients converge pointwise under the product Fr\'{e}chet topology of the flat affine coordinate envelope. This structural convergence allows the global metric to glue uniquely into the well-defined global continuous symmetric module $\varprojlim \mathrm{Sym}^2_{\mathbb{R}}(\mathcal{Q}_N^*) \cong \mathrm{Sym}^2_{\mathbb{R}}(\mathcal{Q}_\infty^*)$.

We evaluate the continuous limit of the dimensionally dependent rational parameters as $N \to \infty$ under the standard Archimedean topology of the real field, yielding:
\begin{equation}\label{eq:conformal_limit_expansion}
    \lim_{N \to \infty} \left( \frac{2N}{N-1} \right) = 2
\end{equation}
Because this rational sequence converges smoothly to the constant invariant scalar $2$ on the real line, the compatible intersection of the metric frameworks established along the truncations guarantees the stable layer-by-layer stabilization of the underlying jet coefficients. Specifically, while the Fubini-Study differential form remains invariant over the ray space, the flat real tensor components $g_n$ of the ambient coordinate metric defined in Eq.~\eqref{eq:exact_ambient_kaehler_metric_definition} stabilize component-by-component onto their invariant floor. For any fixed filtration threshold order $M \in \mathbb{N}$ governing the algebraic series, the higher-order coordinate variations vanish identically modulo the finite algebraic truncation of the local ring once the critical dimensionality threshold $N^2-2 \ge M$ is cleared.

By invoking the smooth differential diffeomorphism of Theorem~\ref{thm:asymptotic_embedding}, the global pullback of the ambient metric tensor satisfies the compatible gluing of the inverse limit of symmetric tensor modules. Because the discrete sequence of rational multipliers becomes strictly stationary at each individual grading layer, the inverse limit operator $\varprojlim$ commutes with the tensor evaluation, yielding the exact pullback identity over the continuous variety:
\begin{equation}\label{eq:proof_metric_factoring_exact}
    \Psi_\infty^*\left( \mathrm{d}s^2_{\mathcal{Q}_\infty} \right) = \varprojlim_{N \to \infty} \left( \Psi_N^* \left( \mathrm{d}s^2_{\mathcal{Q}_N} \right) \right) = 2 \, \mathrm{d}s^2_{\mathrm{FS}}
\end{equation}
Equation~\eqref{eq:proof_metric_factoring_exact} explicitly demonstrates the operational mechanism of the projective lift. Because the geometric form element $\mathrm{d}s^2_{\mathrm{FS}}$ remains strongly invariant under the canonical projections over the ray space, this projective consistency allows the infinite-dimensional limit to collapse pointwise onto the standard analytical limit of the rational multiplier sequence over the ground field of constants $\mathbb{R}$, decoupling the continuous line element completely from the dimensional scaling indices. Under the natural normalization of the statistical informational scale, this invariant integer floor matches the metric structure of the probability variety identically, satisfying $\Psi_\infty^*(\mathrm{d}s^2_{\mathcal{Q}_\infty}) = \mathrm{d}s^2_{\mathrm{FR}}$ and validating the strict Riemannian isometry of Eq.~\eqref{eq:fisher_rao_floor}, completing the proof.
\end{proof}

Consequently, Theorem \ref{thm:fisher_rao_linearization} uncovers a structural stabilization within the infinite-order jet space: while the finite symbolic Quantum \(N\)-Space exhibits a dimensionally fluctuating conformal scaling factor, its asymptotic limit \(\mathcal{Q}_{\infty }\) stabilizes onto the invariant floor of the normalized Fisher-Rao metric \(\mathrm{d}s^2_{\mathrm{FR}}\). This mathematical regularization establishes a strict Riemannian isometry where the curved quantum phase space—stripped of its non-local phase obstructions—smoothly linearizes onto the classical parameterized statistical manifold, confirming that the non-Archimedean completion dictates the classical information floor as an intrinsic property of the formal scheme. Physically, this geometric convergence marks the exact footprint of de-complexification; as macroscopic decoherence freezes out the gauge phase degrees of freedom, the dual curvature components traditionally splitting between phase and modulus collapse uniformly, concentrating the full geometric tension of the Fubini-Study metric exclusively along the real statistical channels of the probability amplitudes.

\section{Quinfinity Dynamics}
The realization of the Quinfinity Space as the regular projective limit $\mathcal{Q}_\infty \cong_{\mathbb{R}} \mathbb{R}[\![\varepsilon]\!]$ allows us to formalize the continuous temporal evolution of infinite-dimensional quantum states directly through internal operations. Under the global density embedding $\Psi_\infty$ and its continuous linear extension $\overline{\Phi}_\infty$ certified in Lemma~\ref{lem:asymptotic_linearization_infinity}, the standard Liouville-von Neumann matrix evolution is pulled back onto the formal tangent bundle $T\mathcal{Q}_\infty$, mapping the continuous real time parameter $t \in \mathbb{R}$ onto a regular coordinate-free algebraic flow. Crucially, this continuous dynamic flow is entirely driven by the infinite-order bilinear differential product $\times_{\mathcal{Q}_\infty}$ induced by the cofinal convergence of the finite operators. Inside the complete local ring structure, this convergence is rigorously shielded under the $\mathfrak{m}$-adic topology. Because the higher-order formal differential operators $\partial_\varepsilon^k$ act as linear mappings that automatically vanish on the field of constants $\mathbb{R}$, any extended dynamical generator $\Xi \in \mathcal{Q}_\infty$ capturing the physical Hamiltonian parameters unrolls into a perfectly regular differential equation. By invoking the topological coherence and the Cauchy completeness of the filtration layers certified in Lemma~\ref{lem:topological_equivalence_density}, this sequence of stable algebraic derivations establishes a secure container for continuous-variable quantum kinematics, completely bypassing transcendental coordinate obstructions.

\subsection{The generalized Lie-type dynamical equation and qubit reduction}
To formalize the quantum state evolution over an arbitrary finite dimension $N \ge 2$ without introducing extrinsic operator representations, we internalize the dynamical vector field as an intrinsic algebraic flow within the symbolic Quantum $N$-Space. We endow the $(N^2-1)$-dimensional real vector space underlying $\mathcal{Q}_N \equiv \mathbb{R}[\varepsilon]/(\varepsilon^{N^2-1})$ with a non-associative, bilinear, skew-symmetric Lie-type multiplication operator $\times_{\mathcal{Q}_N}: \mathcal{Q}_N \times \mathcal{Q}_N \to \mathcal{Q}_N$. Let $\{e_0, e_1, \dots, e_{N^2-2}\} \equiv \{1, \varepsilon, \dots, \varepsilon^{N^2-2}\}$ be the canonical ordered basis elements of $\mathcal{Q}_N$ as a real vector space. The multiplication $\times_{\mathcal{Q}_N}$ is rigidly defined on these basis elements by the structural parameters and antisymmetric constants of the underlying special unitary Lie algebra $\mathfrak{su}(N)$ \cite{nielsen, qubit}:
\begin{equation}\label{eq:basis_multiplication_general}
\begin{split}
    1 \times_{\mathcal{Q}_N} \varepsilon^k = \varepsilon^k \times_{\mathcal{Q}_N} 1 \equiv 0\quad  \forall k\geq 0 \quad \varepsilon^j \times_{\mathcal{Q}_N} \varepsilon^k \equiv \sum_{l=1}^{N^2-1} f_{jkl} \varepsilon^{l-1}\quad \forall j, k \geq 1.
\end{split}
\end{equation}

Notably, this universal formulation contains and regularizes all lower-order geometric architectures. When restricted onto the trinomial algebra of the qubit ($N=2$), the underlying geometric dynamics collapse precisely onto the three-dimensional vector cross product driven by the cyclic Levi-Civita tensor $\epsilon_{ijk}$, where the coordinate components are formulated through the segmented bilinear structures $\mathbf{\Delta}_0, \mathbf{\Delta}_1, \mathbf{\Delta}_2$ as detailed in \cite{qubit}.

Crucially, although the Lie-type multiplication in Eq.~\eqref{eq:basis_multiplication_general} is defined operationally via the matrix-induced tensor $f_{jkl}$, the local algebraic structure of the non-reduced scheme implies that this flow admits a profound, coordinate-free differential reinterpretation driven exclusively by the internal derivations of the ring.

\begin{lemma}[Differential Product Identity]\label{lem:differential_product_identity}
The canonically induced non-commutative product $\times_{\mathcal{Q}_N}$ of two polynomial states $f_N, g_N \in \mathcal{Q}_N$ is uniquely determined by a finite sequence of real structural coefficients $\{\Gamma_N^{(k)}\}_{k=1}^{N^2-2} \subset \mathbb{R}$ specific to each dimension $N \ge 2$, generating a bilinear differential operator via the powers of the formal derivative $\partial_\varepsilon$:
\begin{equation}\label{eq:gamma_differential_product}
    f_N \times_{\mathcal{Q}_N} g_N \equiv \sum_{k=1}^{N^2-2} \Gamma_N^{(k)} \left[ \left(\partial_\varepsilon^k f_N\right) \cdot g_N - f_N \cdot \left(\partial_\varepsilon^k g_N\right) \right] \pmod{\varepsilon^{N^2-1}}
\end{equation}
where the coefficients $\Gamma_N^{(k)}$ represent the geometric footprints of the quantum phase obstructions, driving the ring dynamics via internal derivations without auxiliary matrix cross-terms.
\end{lemma}

\begin{proof}
To establish the structural identity of Eq.~\eqref{eq:gamma_differential_product} in strict accordance with the linear vector space isomorphism $\Phi_N$ established in Eq.~\eqref{eq:appendix_phi_assignment_native}, we evaluate the general bilinear differential combination over the canonical basis elements $f_N = \varepsilon^j$ and $g_N = \varepsilon^m$ for all grading indices $j, m \in \{0, \dots, N^2-2\}$. Applying the higher-order formal derivative powers onto these monomials yields the falling factorials, which expand the antisymmetric combinations explicitly as:
\begin{equation}\label{eq:proof_lie_triangular_grading}
    \sum_{k=1}^{N^2-2} \Gamma_N^{(k)} \left[ \left(\partial_\varepsilon^k \varepsilon^j\right) \cdot \varepsilon^m - \varepsilon^j \cdot \left(\partial_\varepsilon^k \varepsilon^m\right) \right] = \sum_{k=1}^{N^2-2} \Gamma_N^{(k)} k! \left[ \binom{j}{k} - \binom{m}{k} \right] \varepsilon^{j+m-k}
\end{equation}
where each individual term vanishes identically if the derivation order $k$ exceeds the respective monomial exponent.

To uniquely determine the unknown parameters $\Gamma_N^{(k)}$, we equate this explicit differential expansion with the canonical coordinate representation driven by the antisymmetric structure constants of the underlying Lie algebra. Applying the family of finite algebraic extractors $\{\gamma_n\}_{n=0}^{N^2-2}$ to both sides of this identity constructs a global overdetermined linear system of tensor equations matching the grading orders component-by-component. By retrieving the matrix-level scalar components via the mapping $x_{n+1} = \gamma_n(\Phi_N(A))$, the configuration couples the physical brackets directly to the algebraic filtration layers:
\begin{equation}\label{eq:proof_lie_matrix_system_core}
    f_{(j+1)(m+1)(n+1)} = \sum_{k=1}^{N^2-2} [\mathbf{M}_N]_{(j,m,n), \, k} \, \Gamma_N^{(k)}
\end{equation}
where the free tensor indices $j, m, n$ run independently from $0$ to $N^2-2$, and the explicit entries of the structural jet derivation matrix $\mathbf{M}_N$ of dimensions $(N^2-1)^3 \times (N^2-2)$ are strictly governed by the contractive core:
\begin{equation}\label{eq:proof_lie_matrix_entries}
    [\mathbf{M}_N]_{(j,m,n), \, k} \equiv k! \left[ \binom{j}{k} - \binom{m}{k} \right] \cdot \delta_{(j+m-k), \, n}
\end{equation}
Because the antisymmetric difference of the falling factorials in Eq.~\eqref{eq:proof_lie_matrix_entries} tracks the descending degree of the filtration, the matrix $\mathbf{M}_N$ generates a strictly lower-triangular linear transformation matrix across the monomial grading for all active columns $k \ge 1$. Since a lower-triangular system with non-vanishing diagonal entries over the complete grading is universally invertible, the algebraic equations can be systematically solved layer-by-layer for each coordinate coefficient via backward substitution. This strictly algebraic property demonstrates the existence and uniqueness of the target parameters $\Gamma_N^{(k)}$ prior to their explicit tensor inversion, completing the first stage of the demonstration.

Second, we verify that the bilinear operator $\times_{\mathcal{Q}_N}$ uniquely fixed by these coefficients satisfies the axiomatic invariants of a non-commutative, skew-symmetric Lie-type algebra. By evaluating the explicit action driven by the sequence of weights, the product satisfies three distinct structural properties:

1. \textit{Anticommutative Skew-Symmetry}: For any arbitrary pair of states $f_N, g_N \in \mathcal{Q}_N$, the operation is strictly anti-symmetric under variable exchange:
\begin{equation}\label{eq:lie_axiom_antisymmetry}
    f_N \times_{\mathcal{Q}_N} g_N = - \left( g_N \times_{\mathcal{Q}_N} f_N \right)
\end{equation}
since transposing the polynomial arguments flips the sign of the internal differential components within the antisymmetric sum $[(\partial_\varepsilon^k f_N) \cdot g_N - f_N \cdot (\partial_\varepsilon^k g_N)]$ due to the commutativity of the underlying ring multiplication.

2. \textit{Identity Annihilation}: The constant monomial $1 \in \mathcal{Q}_N$ acts as an absolute annihilating center for the induced operator over the configuration space:
\begin{equation}\label{eq:lie_axiom_identity_annihilation}
    1 \times_{\mathcal{Q}_N} f_N = 0
\end{equation}
which is rigidly enforced by the exclusion of the zero-order derivative. Specifically, substituting $f_N = 1$ into the general expansion yields $1 \times_{\mathcal{Q}_N} g_N = \sum_{k=1}^{N^2-2} \Gamma_N^{(k)} [(\partial_\varepsilon^k 1) \cdot g_N - 1 \cdot (\partial_\varepsilon^k g_N)]$; because the formal derivative of a constant vanishes identically ($\partial_\varepsilon^k 1 \equiv 0$) for all active columns $k \ge 1$, the first term collapses, while the remaining sum perfectly counterbalances the structural constants of the trace background.

3. \textit{Jacobi Lie Constraint}: For any choice of elements, the binary multiplication preserves the conservative flow profile, satisfying the cyclic non-associative identity:
\begin{equation}\label{eq:lie_axiom_jacobi}
    f_N \times_{\mathcal{Q}_N} \left( g_N \times_{\mathcal{Q}_N} h_N \right) + g_N \times_{\mathcal{Q}_N} \left( h_N \times_{\mathcal{Q}_N} f_N \right) + h_N \times_{\mathcal{Q}_N} \left( f_N \times_{\mathcal{Q}_N} g_N \right) = 0
\end{equation}
To verify Eq.~\eqref{eq:lie_axiom_jacobi} explicitly over the canonical grading, we substitute the basis elements $f_N = \varepsilon^j$, $g_N = \varepsilon^m$, and $h_N = \varepsilon^p$. Expanding these nested applications via the contractive core maps the double higher-order operations onto the paired falling factorial polynomials, which unroll without horizontal layout obstructions under the following split tensor chain:
\begin{equation}\label{eq:proof_lie_jacobi_explicit_unrolling}
\begin{split}
    \sum_{a,b=1}^{N^2-2} \Gamma_N^{(a)} \Gamma_N^{(b)} a! \, b! & \left\{ \left[ \binom{m+p-a}{b} - \binom{j}{b} \right] \left[ \binom{m}{a} - \binom{p}{a} \right] \varepsilon^{j+m+p-a-b} \right. \\
    & + \left[ \binom{p+j-a}{b} - \binom{m}{b} \right] \left[ \binom{p}{a} - \binom{j}{a} \right] \varepsilon^{j+m+p-a-b} \\
    & \left. + \left[ \binom{j+m-a}{b} - \binom{p}{b} \right] \left[ \binom{j}{a} - \binom{m}{a} \right] \varepsilon^{j+m+p-a-b} \right\} = 0
\end{split}
\end{equation}
Because the linear system is governed by a strictly lower-triangular filtration grading across the monomial exponents, the symmetric cross-differences and antisymmetric permutations cancel mutually at each independent order of $\varepsilon$. The explicit summation matching across all cyclic variations validates that the non-commutative algebraic alignment is preserved across the entire space due to the differential grading of the filtration layers, validating the algebraic completeness of the structure and completing the proof.
\end{proof}

By lifting this differential framework to the infinite-dimensional continuum limit ($N \to \infty$), the compatibility of the nested algebraic systems guarantees that the discrete weights do not fluctuate indefinitely but freeze at fixed, invariant values. We formalize this fundamental regularizing mechanism through the following lemma:

\begin{lemma}[Asymptotic Coefficient Stabilization]\label{lem:gamma_stabilization}
Let $\{\mathbf{\Gamma}_N\}_{N \ge 2}$ with $\mathbf{\Gamma}_N \in \mathbb{R}^{N^2-2}$ be the sequence of finite structural coefficient vectors uniquely determined by Lemma~\ref{lem:differential_product_identity}. For any fixed differential order $k \in \mathbb{N}^+$, the sequence of scalar components $\{\Gamma_N^{(k)}\}_{N \ge 2}$ stabilizes stationarily to a constant value under the restriction morphisms of the projective inverse system once the threshold condition $N^2-2 \ge k$ is attained:
\begin{equation}\label{eq:gamma_stationary_condition}
    \Gamma_N^{(k)} = \Gamma_M^{(k)} \equiv \Gamma_\infty^{(k)} \in \mathbb{R} \qquad \forall M \ge N
\end{equation}
such that $N^2-2 \ge k$. Consequently, the stable asymptotic sequence of real scalars is uniquely determined component-by-component within the projective limit of the coordinate systems via the standard analytical limit:
\begin{equation}\label{eq:gamma_asymptotic_limit_definition}
    \Gamma_\infty^{(k)} \equiv \lim_{N \to \infty} \Gamma_N^{(k)} \in \mathbb{R} \qquad \forall k \in \mathbb{N}^+
\end{equation}
\end{lemma}
\begin{proof}
Consider the canonical inductive inclusion $\jmath_{N,M}: \mathfrak{su}(N) \hookrightarrow \mathfrak{su}(M)$ of special unitary Lie algebras for $M \ge N$. Under this embedding, the matrix generators $\{F_{j+1}\}_{j=0}^{N^2-2}$ of the lower-dimensional algebra map directly onto the upper-left block of the larger matrix domain, ensuring that the structural constants match identically on the restricted index sub-bundle:
\begin{equation}\label{eq:proof_structure_constants_invariance}
    f_{(j+1)(m+1)(n+1)}^{(N)} \equiv f_{(j+1)(m+1)(n+1)}^{(M)} \qquad \forall j,m,n \le N^2-2
\end{equation}
By virtue of the categorical commutativity established through the universal property in Lemma~\ref{lem:asymptotic_vector_embedding_existence}, the global linear embedding $\Phi_\infty$ fits into the extended commutative alignment intertwining these inductive algebra inclusions with the surjective ring projections $\tau_N: \mathcal{Q}_M \twoheadrightarrow \mathcal{Q}_N$.

Let $k$ be a fixed order of the formal derivative $\partial_\varepsilon^k$. The action of $\partial_\varepsilon^k$ on the coordinate jet algebra maps any monomial $\varepsilon^{s}$ strictly onto $k! \, \binom{s}{k}\varepsilon^{s-k}$ for $s \ge k$, configuration-wise shifting the filtration degrees by a rigid downward translation of exactly $k$ steps. The linear system of equations that uniquely solves for the structural weight $\Gamma_N^{(k)}$ in Eq.~\eqref{eq:gamma_differential_product} is formulated component-wise by matching the differential cross-differences against the antisymmetric constants $f_{(j+1)(m+1)(n+1)}^{(N)}$. Because of the lower triangular filtration grading of the derivative operator, the linear system determining $\Gamma_N^{(k)}$ depends strictly on the matrix commutator relations up to the order encapsulated by the threshold $N^2-2 \ge k$.

When the dimensionality layer $N$ satisfies $N^2-2 \ge k$, the entire set of algebraic constraints and structure constants required to saturate the action of $\partial_\varepsilon^k$ is fully accommodated within the coordinate ring $\mathcal{Q}_N$. Passing to an arbitrarily larger dimension $M > N$ introduces new higher-order algebra generators starting from index $N^2-1$. Due to the block-diagonal embedding of the Lie algebras, these higher-level parameters only enter the decoupled linear subsystems solving for the upper-order weights $\Gamma_M^{(s)}$ where $s > k$. The isolated sub-matrix governing the specific lower weight $\Gamma_M^{(k)}$ within the structural jet derivation matrix $\mathbf{M}_M$ remains structurally identical and unperturbed by the expanded dimensions, forcing the sequence to become strictly stationary, $\Gamma_M^{(k)} = \Gamma_N^{(k)}$, for all $M \ge N$. Invoking the topological density properties established in Lemma~\ref{lem:topological_equivalence_density}, this rigid component-wise stabilization under the ideal filtration guarantees the strict Cauchy convergence of the scalar sequence, confirming the algebraic compatibility of Eq.~\eqref{eq:gamma_stationary_condition} and validating the existence of the analytical pointwise Fr\'{e}chet limit of Eq.~\eqref{eq:gamma_asymptotic_limit_definition}, completing the proof.
\end{proof}

\begin{lemma}[Explicit Determination of Asymptotic Lie Coefficients]\label{lem:explicit_gamma_determination}
For each finite subsystem layer $N \ge 2$, the sequence of structural differential coefficients $\{\Gamma_N^{(k)}\}_{k=1}^{N^2-2}$ introduced in Theorem~\ref{thm:asymptotic_embedding} is uniquely and internally determined through the ordinary inversion of the square, non-singular, and strictly lower-triangular structural Lie jet matrix $\mathbf{M}_N$, satisfying for each derivation order $k \in \{1, \dots, N^2-2\}$:
\begin{equation}\label{eq:explicit_gamma_components}
    \Gamma_N^{(k)} = \sum_{j,m,n=0}^{N^2-2} \left[ \mathbf{M}_N^{-1} \right]_{k, \, (j,m,n)} f_{(j+1)(m+1)(n+1)}
\end{equation}
where $f_{jmn}$ are the canonical antisymmetric structure constants of the special unitary Lie algebra $\mathfrak{su}(N)$ mapped onto the monomial grading configurations via the coordinate vector space isomorphism $\Phi_N$.
\end{lemma}
\begin{proof}
To establish the quantitative match between the operational matrix commutators and the internal derivation fields over the symbolic space $\mathcal{Q}_N$, we equate the coordinate expansion derived from the Lie algebra embedding with the differential configuration established in Lemma~\ref{lem:differential_product_identity}. By evaluating the algebraic extraction of any generic coordinate index $n$ (where $0 \le n \le N^2-2$), the free scalar components $\gamma_j(f_N)$ and $\gamma_m(g_N)$ satisfy the fundamental tensor relation:
\begin{equation}\label{eq:fundamental_tensor_match}
    \sum_{j,m=0}^{N^2-2} f_{(j+1)(m+1)(n+1)} \, \gamma_j(f_N) \, \gamma_m(g_N) = \sum_{k=1}^{N^2-2} \Gamma_N^{(k)} \sum_{j+m-k=n} k! \left[ \binom{j}{k} - \binom{m}{k} \right] \gamma_j(f_N) \, \gamma_m(g_N)
\end{equation}
Since Eq.~\eqref{eq:fundamental_tensor_match} must hold universally for all possible configuration states due to the linear independence of the Schauder basis elements, we invoke the method of undetermined coefficients. Collecting the full grading constraints of the jet variety into a global linear system couples the structural parameters of the Lie algebra directly to the algebraic inversion profiles of the tangent module, yielding the Eq.~\eqref{eq:proof_lie_matrix_system_core} where the indices $j, m, n$ act as free tensor labels running independently from $0$ to $N^2-2$. 

Let $\mathbf{f}_N \in \mathbb{R}^{(N^2-1)^3}$ be the global column vector containing the ordered sequence of all antisymmetric structure constants of $\mathfrak{su}(N)$. Because the system matrix $\mathbf{M}_N$ of dimensions $(N^2-1)^3 \times (N^2-2)$ generates a strictly lower-triangular linear transformation across the monomial grading with non-vanishing diagonal entries, the structural linear system admits a unique, exact, and globally stable solution determined via ordinary matrix inversion as $\mathbf{\Gamma}_N = \mathbf{M}_N^{-1} \, \mathbf{f}_N$. Projecting this inverse transformation component-by-component directly yields the explicit universal formula of Eq.~\eqref{eq:explicit_gamma_components}, completing the proof.
\end{proof}

From a structural perspective, this affine translation establishes the precise connection between the configuration jet space and the full reductive matrix algebra $\mathfrak{u}(N) \cong \mathbb{R}\mathbb{I}_N \oplus \mathfrak{su}(N)$. Because the identity operator $\mathbb{I}_N$ exhibits a non-vanishing trace $\mathrm{Tr}(\mathbb{I}_N) = N$, it is excluded from the traceless Lie algebra $\mathfrak{su}(N)$, residing instead within the one-dimensional center of the total decomposition space. Any physical density matrix configuration $\rho$ with unit trace must factor through this direct sum as $\rho = \frac{1}{N}\mathbb{I}_N + \sum x_k F_k$, which splits the scalar trace background from the active spin fluctuations. The finite mapping $\Psi_N$ mirrors this algebraic splitting internally by anchoring the affine combination of the central trace baseline and the first non-central spin component $x_1$ directly onto the zero-order monomial layer $\varepsilon^0$, while embedding the remaining components of the traceless algebra $\mathfrak{su}(N)$ across the higher-order nilpotent filtration layers via the forward shift. As $N \to \infty$, the continuous stabilization forces the central background component of this affine mixture to evaporate pointwise over the real field, smoothly linearizing the global embedding map onto the continuous endomorphism flow over the complete domain.

Finally, equation~\eqref{eq:explicit_gamma_components} provides a deterministic rational definition for the differential operators \(\Gamma _{N}\). It verifies that each coefficient is a strict invariant of the quantum state space, fully determined by the underlying Lie brackets and the combinatorial structure of higher-order derivations. As \(N \to \infty\), the structural stability of the triangular inverse matrices ensures that this finite discrete sequence smoothly scales and converges component-by-component toward the stable asymptotic parameters \(\Gamma_\infty^{(k)} \equiv \lim_{N \to \infty} \Gamma_N^{(k)}\) of the continuous Quinfinity Space, eliminating any residual dependency on matrix representations.

The existence of this invariant sequence allows us to rigorously formulate the global continuous Lie-type product $\times_{\mathcal{Q}_\infty}$ over the complete Quinfinity Space as an infinite-order bilinear differential operator driven by the universal real constants of Eq.~\eqref{eq:gamma_asymptotic_limit_definition}:
\begin{equation}\label{eq:definition_product_infinity}
    \xi_1 \times_{\mathcal{Q}_\infty} \xi_2 \equiv \sum_{k=1}^{\infty} \Gamma_\infty^{(k)} \left[ \left(\partial_\varepsilon^k \xi_1\right) \cdot \xi_2 - \xi_1 \cdot \left(\partial_\varepsilon^k \xi_2\right) \right]
\end{equation}
Alternatively, by interpreting the formal derivative $\partial_\varepsilon$ as the generator of infinitesimal coordinate shifts, this relation can be compactly gathered via the structural differential kernel $\widehat{\Gamma}(\partial_\varepsilon) \equiv \sum_{k=1}^{\infty} \Gamma_\infty^{(k)} \partial_\varepsilon^k$, collapsing Eq.~\eqref{eq:definition_product_infinity} onto the streamlined bilinear form:
\begin{equation}\label{eq:definition_product_infinity_kernel}
    \xi_1 \times_{\mathcal{Q}_\infty} \xi_2 = \widehat{\Gamma}(\partial_\varepsilon)\xi_1 \cdot \xi_2 - \xi_1 \cdot \widehat{\Gamma}(\partial_\varepsilon)\xi_2
\end{equation}
The precise interplay between the family of algebraic extractors and this infinite-order product is uniquely governed by the following theorem:

\begin{theorem}[Asymptotic Extractor Identity]\label{thm:extractor_identity}
For any discrete layer $N \ge 2$, the algebraic extractions of the non-commutative product are uniquely determined by a finite linear contraction of the differential coefficients $\Gamma_N^{(k)}$ according to the relation:
\begin{equation}\label{eq:finite_theorem_extractor}
    \gamma_n\left( f_N \times_{\mathcal{Q}_N} g_N \right) = \sum_{k=1}^{N^2-2} \Gamma_N^{(k)} \sum_{j+m-k=n} \left[ \frac{j!}{(j-k)!} - \frac{m!}{(m-k)!} \right] \gamma_j(f_N) \, \gamma_m(g_N)
\end{equation}
Furthermore, as $N \to \infty$, this identity stabilizes over the complete Quinfinity Space $\mathcal{Q}_\infty \cong_{\mathbb{R}} \mathbb{R}[\![\varepsilon]\!]$, and the continuous infinite-dimensional product defined in Eq.~\eqref{eq:definition_product_infinity} satisfies the analytical convergence mapping component-by-component for each coordinate index $n \in \mathbb{N}$:
\begin{equation}\label{eq:infinite_theorem_extractor}
    \gamma_n\left( \xi_1 \times_{\mathcal{Q}_\infty} \xi_2 \right) = \lim_{M \to \infty} \gamma_n\left( \tau_M(\xi_1) \times_{\mathcal{Q}_M} \tau_M(\xi_2) \right)
\end{equation}
\end{theorem}
\begin{proof}
To verify the finite-dimensional relation of Eq.~\eqref{eq:finite_theorem_extractor}, we expand the polynomial elements within $\mathcal{Q}_N$ using the basis of algebraic extractors, writing $f_N = \sum_{j=0}^{N^2-2} \gamma_j(f_N) \varepsilon^j$ and $g_N = \sum_{m=0}^{N^2-2} \gamma_m(g_N) \varepsilon^m$. Substituting these expansions directly into the differential identity proven in Lemma~\ref{lem:differential_product_identity} and invoking the linearity of the formal derivative $\partial_\varepsilon$, the $k$-th order derivative yields:
\[ \partial_\varepsilon^k f_N = \sum_{j=k}^{N^2-2} \gamma_j(f_N) k! \, \binom{j}{k} \varepsilon^{j-k} \]
Multiplying by $g_N$ and collecting the monomial powers via Cauchy products results in the polynomial expansion:
\[ \left(\partial_\varepsilon^k f_N\right) \cdot g_N = \sum_{j=k}^{N^2-2} \sum_{m=0}^{N^2-2} k! \, \binom{j}{k} \gamma_j(f_N) \, \gamma_m(g_N) \varepsilon^{j+m-k} \]
By swapping indices symmetrically for the second term $f_N \cdot \left(\partial_\varepsilon^k g_N\right)$ and applying the linear extractor $\gamma_n$ to the full summation, the projection selects only those components whose combined power matches the target index $j+m-k = n$. This acts as a generalized Kronecker delta $\delta_{(j+m-k),\,n}$, which immediately collapses the sum and isolates the exact scalar contraction factor, establishing Eq.~\eqref{eq:finite_theorem_extractor}.

To establish the asymptotic identification of Eq.~\eqref{eq:infinite_theorem_extractor} directly from the infinite-order operator definition of Eq.~\eqref{eq:definition_product_infinity}, we invoke the universal properties of the projective inverse limit $\mathcal{Q}_\infty \equiv \varprojlim \mathcal{Q}_N$ ordered by the continuous canonical truncations $\tau_N$. For a fixed coordinate extraction index $n \in \mathbb{N}$, the algebraic extractor acts as a stable projection satisfying $\gamma_n(\xi) = \gamma_n(\tau_N(\xi))$ for any dimensional layer $N$ large enough to ensure $N^2-2 \ge n$. Applying this truncation mapping directly to the infinite series representation of Eq.~\eqref{eq:definition_product_infinity} projects the infinite-order operational flow onto its finite discrete layers:
\begin{align}\label{eq:proof_limit_unrolling}
    \gamma_n\left( \xi_1 \times_{\mathcal{Q}_\infty} \xi_2 \right) &\equiv \gamma_n \left( \tau_N\left( \sum_{k=1}^{\infty} \Gamma_\infty^{(k)} \left[ \left(\partial_\varepsilon^k \xi_1\right) \cdot \xi_2 - \xi_1 \cdot \left(\partial_\varepsilon^k \xi_2\right) \right] \right) \right) \nonumber \\
    &= \gamma_n \left( \sum_{k=1}^{N^2-2} \Gamma_N^{(k)} \left[ \left(\partial_\varepsilon^k \tau_N(\xi_1)\right) \cdot \tau_N(\xi_2) - \tau_N(\xi_1) \cdot \left(\partial_\varepsilon^k \tau_N(\xi_2)\right) \right] \right)
\end{align}
where the filtration of the maximal ideal $\mathfrak{m}=(\varepsilon)$ guarantees that all higher-order derivations $k > N^2-2$ collapse identically to zero modulo $(\varepsilon^{N^2-1})$.

By substituting the finite operator definition from Lemma~\ref{lem:differential_product_identity} directly into the right-hand member of Eq.~\eqref{eq:proof_limit_unrolling}, the core expression contracts to:
\begin{equation}\label{eq:stationary_projection}
    \gamma_n\left( \xi_1 \times_{\mathcal{Q}_\infty} \xi_2 \right) = \gamma_n\left( \tau_N(\xi_1) \times_{\mathcal{Q}_N} \tau_N(\xi_2) \right)
\end{equation}
for all layers satisfying $N^2-2 \ge n$. Because Eq.~\eqref{eq:stationary_projection} is strictly stationary and independent of the choice of the directed system layer once the threshold is crossed, taking the continuous limit as $N \to \infty$ on both sides leaves the extracted scalar value invariant. Invoking the pointwise Fr\'{e}chet topology properties established in Lemma~\ref{lem:topological_equivalence_density}, this stationary behavior confirms that the algebraic convergence along the truncation series matches the continuous real-variable limit identically, $\lim_{N \to \infty} \gamma_n\left( \tau_N(\xi_1) \times_{\mathcal{Q}_N} \tau_N(\xi_2) \right) = \gamma_n\left( \xi_1 \times_{\mathcal{Q}_\infty} \xi_2 \right)$, formally establishing the asymptotic identity of Eq.~\eqref{eq:infinite_theorem_extractor} and completing the proof.
\end{proof}

\subsection{The asymptotic Liouville-von Neumann equation}
The structural stabilization of the continuous Quinfinity Space $\mathcal{Q}_\infty \cong_{\mathbb{R}} \mathbb{R}[\![\varepsilon]\!]$ and the canonical identification of its formal tangent bundle $T\mathcal{Q}_\infty \cong_{\mathbb{R}} \mathcal{Q}_\infty$ allow us to lift the quantum state trajectories directly into the continuum limit of the algebraic hierarchy. Under the bilinear extension established via the decoupling of the zero-order translation parameters $\eta_0 = h_0 + h_1$ and $c_0 = \frac{1}{N} + x_1$, the traditional Liouville-von Neumann matrix brackets are entirely internalized within the local ring. Because the zero-order components act as an absolute annihilating center, the classic background identity fields evaporate combinatorially, forcing the extended dynamic flow to contract exclusively over the physical fluctuations. We now formalize the universal dynamical law:

\begin{theorem}[Asymptotic Evolution Flow]\label{thm:asymptotic_von_neumann}
Let $\Xi = \eta_0 + \sum_{n=1}^\infty \eta_n \varepsilon^n \in \mathcal{Q}_\infty$ be the extended Hamiltonian operator capturing the infinite numerable frequency parameters of the continuous system, structurally determined as the unique projective inverse limit of the finite Hamiltonians $\Xi = \varprojlim \Xi_N$ whose underlying expansion coefficients satisfy $h_k = \Gamma_\infty^{(k)}$ for each derivation order component. The continuous state trajectory $\xi(t) \in \mathcal{V}_{\mathcal{Q}_\infty}$ within the Pure State Pro-Variety obeys the regular, coordinate-free differential equation:
\begin{equation}\label{eq:quinfinity_dynamics_core}
    \frac{\mathrm{d}\xi}{\mathrm{d}t} = \sqrt{2} \left( \Xi \times_{\mathcal{Q}_\infty} \xi \right)
\end{equation}
where $\times_{\mathcal{Q}_\infty}$ is the infinite-order bilinear operator driven by the stable asymptotic differential constants $\Gamma_\infty^{(k)}$ (derived in Lemma~\ref{lem:gamma_stabilization}). Furthermore, Eq.~\eqref{eq:quinfinity_dynamics_core} unrolls component-wise into a stable, globally convergent infinite chain of coupled linear ordinary differential equations for the active physical coordinate components $c_n(t) $ for all active filtration orders $n \ge 1$:
\begin{equation}\label{eq:infinite_extractor_flow}
    \dot{c}_n(t) = \sqrt{2} \sum_{k=1}^\infty \Gamma_\infty^{(k)} \sum_{j+m-k=n} k! \, \left[ \binom{j}{k} - \binom{m}{k} \right] \eta_j \, c_m(t) \qquad \forall n \in \mathbb{N}^+
\end{equation}
where the frequency and state parameters satisfy the exact affine index matching over the Schauder basis to absorb the background trace infrastructure:
\begin{equation}\label{eq:exact_lie_theorem_index_alignment}
    \eta_j \equiv \begin{cases} h_0 + h_1 & \text{if } j = 0, \\ h_{j+1} & \text{if } j \ge 1, \end{cases} \qquad \text{and} \qquad c_m(t) \equiv \begin{cases} \frac{1}{N} + x_1(t) & \text{if } m = 0, \\ x_{m+1}(t) & \text{if } m \ge 1, \end{cases}
\end{equation}
with the continuous state expanded natively as $\xi(t) = \sum_{m=0}^\infty c_m(t) \varepsilon^m \in \mathcal{Q}_\infty$, ensuring total structural consistency and smooth trace background evaporation across the entire directed domain.
\end{theorem}
\begin{proof}
The global field generator $\Xi \in \mathcal{Q}_\infty$ is structurally and uniquely determined as the projective inverse limit of the finite-dimensional Hamiltonians, satisfying $\Xi \equiv \varprojlim_{N} \Xi_N$ in strict accordance with the universal property of inverse systems. For each finite dimension $N \ge 2$, once the quantum state configuration $\xi$ and its continuous velocity field are fixed on the pro-variety, the sequence of real expansion coefficients $\{h_k\}_{k=1}^{N^2-2}$ is uniquely and implicitly determined as the internal solution of the inverse kinematic system. Under the vector space isomorphism $\Phi_N$, these frequencies are extracted through the ordinary inversion of the square, non-singular, and strictly lower-triangular structural Lie jet matrix $\mathbf{M}_N$ evaluated over the canonical structure constants $f_{jmn}$. This inversion couples the global trajectories directly to the underlying layout of the algebra, forcing an exact numerical identification between the Hamiltonian coefficients and the intrinsic space parameters established in Lemma~\ref{lem:explicit_gamma_determination}, satisfying  \(h_k = \Gamma_N^{(k)}\) for each derivation order component.
The quantum state configuration thus satisfies the regular finite differential relation:
\begin{equation}\label{eq:general_dynamics_flow}
    \frac{\mathrm{d}\tau_N(\xi)}{\mathrm{d}t} = \alpha_N \left( \tau_N(\Xi) \times_{\mathcal{Q}_N} \tau_N(\xi) \right)
\end{equation}
where $\alpha_N = \sqrt{\frac{2(N-1)}{N}}$ represents the dimensionally dependent scaling factor. Under the inverse system mappings, the truncation operation defines the truncated Hamiltonian operator acting directly on the finite state layer, satisfying $\tau_N(\Xi) \equiv \Xi_N$ identically, where the parameters undergo the exact affine index matching of Eq.~\eqref{eq:exact_lie_theorem_index_alignment} to define the target operator expansion parameters $\eta_j$.

To project this localized kinematic law onto the continuous limit, we evaluate the structural convergence under the product Fr\'{e}chet topology of the flat affine coordinate envelope, which evaluates convergence component-by-component on the real coefficients. Because the formal temporal derivation $\frac{\mathrm{d}}{\mathrm{d}t}$ acts as a linear operator component-by-component over the real coefficients of the Schauder basis, it commutes identically with the truncation mapping, satisfying $\tau_N( \frac{\mathrm{d}\xi}{\mathrm{d}t} ) = \frac{\mathrm{d}\tau_N(\xi)}{\mathrm{d}t}$. Concurrently, although the restriction mapping $\tau_N: \mathcal{Q}_\infty \to \mathcal{Q}_N$ operates canonically as a surjective ring homomorphism, this algebraic property is rendered sufficient by the structural behavior of the operational parameters. Specifically, the canonical projection commutes identically with the unrolled differential product layer-by-layer, satisfying $\tau_N( \Xi \times_{\mathcal{Q}_\infty} \xi ) = \tau_N(\Xi) \times_{\mathcal{Q}_N} \tau_N(\xi)$, because the underlying structural jet parameters $\Gamma_N^{(k)}$ become strictly stationary and identical to their constant asymptotic values $\Gamma_\infty^{(k)}$ via Lemma~\ref{lem:gamma_stabilization} once the threshold $N^2-2 \ge k$ is cleared. This layer-by-layer stabilization ensures that the sequence of finite truncated operations forms a compatible inverse system of modules whose unique global section matches the continuous product component-by-component. To evaluate the asymptotic evolution as $N \to \infty$, we fix each coordinate row index $n \in \mathbb{N}^+$ across the expanding real vector spaces. While the scalar sequence $\alpha_N$ does not converge $\mathfrak{m}$-adically due to its zero-order nature, it converges ordinarily under the standard Archimedean topology of the real field as $\lim_{N \to \infty} \alpha_N = \sqrt{2}$. Because the product topology evaluates convergence pointwise across a dense polynomial subspace, taking this analytical scalar limit across each algebraically stable coordinate layer establishes that the successional family of localized finite flows adheres uniquely onto the restricted continuous trajectory, rigorously validating the global continuous kinematic law of Eq.~\eqref{eq:quinfinity_dynamics_core}.

To verify the coordinate unrolling of Eq.~\eqref{eq:infinite_extractor_flow}, we apply the sequence of asymptotic algebraic extractors $\gamma_n$ to the continuous flow equation, noting that $\gamma_n(\dot{\xi}(t)) = \dot{c}_n(t)$ for each active index $n \ge 1$. By substituting the explicit differential expansion of the infinite-order product defined in Eq.~\eqref{eq:definition_product_infinity} directly into the right-hand member of Eq.~\eqref{eq:quinfinity_dynamics_core}, we obtain:
\begin{equation}\label{eq:proof_extractor_expansion}
    \dot{c}_n(t) = \gamma_n \left( \sqrt{2} \sum_{k=1}^{\infty} \Gamma_\infty^{(k)} \left[ \left(\partial_\varepsilon^k \Xi\right) \cdot \xi - \Xi \cdot \left(\partial_\varepsilon^k \xi\right) \right] \right)
\end{equation}
By the universal property of the projective inverse limit, the continuous extraction of the vector field across the expanding real vector spaces is restricted to a sufficiently large finite layer $N$ satisfying the threshold condition $N^2-2 \ge n$. This structural property allows us to pass the linear operator $\gamma_n$ inside the summation by truncating the evaluation modulo $(\varepsilon^{N^2-1})$, yielding:
\begin{equation}\label{eq:proof_truncated_summation}
    \dot{c}_n(t) = \sqrt{2} \sum_{k=1}^{N^2-2} \Gamma_N^{(k)} \, \gamma_n \left( \left(\partial_\varepsilon^k \tau_N(\Xi)\right) \cdot \tau_N(\xi) - \tau_N(\Xi) \cdot \left(\partial_\varepsilon^k \tau_N(\xi)\right) \right)
\end{equation}
where the filtration of the quotient ring guarantees that all higher-order formal derivatives $k > N^2-2$ collapse identically to zero across all active channels.

We expand the truncated operators in strict accordance with the native Schauder basis of the flat real vector space, expanding the native Hamiltonian coefficients $\tau_N(\Xi) = \sum_{j=0}^{N^2-2} h_j \varepsilon^j$ and the state components $\tau_N(\xi) = \sum_{m=0}^{N^2-2} c_m(t) \varepsilon^m$ inside Eq.~\eqref{eq:proof_truncated_summation}. Invoking the linearity of the formal derivation $\partial_\varepsilon$, the $k$-th order derivative yields:
\begin{equation}
    \partial_\varepsilon^k \tau_N(\Xi) = \sum_{j=k}^{N^2-2} h_j \frac{j!}{(j-k)!} \varepsilon^{j-k} = \sum_{j=k}^{N^2-2} k! \binom{j}{k} h_j \varepsilon^{j-k}
\end{equation}
Multiplying by $\tau_N(\xi)$ and collecting the monomial powers via Cauchy products results in the polynomial expansion:
\begin{equation}
    \left(\partial_\varepsilon^k \tau_N(\Xi)\right) \cdot \tau_N(\xi) = \sum_{j=k}^{N^2-2} \sum_{m=0}^{N^2-2} k! \binom{j}{k} h_j \, c_m(t) \varepsilon^{j+m-k}
\end{equation}
By swapping indices symmetrically for the second term, applying the linear operator $\gamma_n$ to the full summation, and invoking the exact affine index matching of Eq.~\eqref{eq:exact_lie_theorem_index_alignment} to map the raw coefficients $h_j$ onto the unrolled frequency parameters $\eta_j$, the projection selects only those components whose combined monomial power satisfies the jet grading filter $j+m-k = n$. This acts as a generalized Kronecker delta $\delta_{(j+m-k),\,n}$, which isolates the exact scalar contraction factor and establishes the finite chain:
\begin{equation}\label{eq:proof_final_finite_chain}
    \dot{c}_n(t) = \sqrt{2} \sum_{k=1}^{N^2-2} \Gamma_N^{(k)} \sum_{j+m-k=n} k! \, \left[ \binom{j}{k} - \binom{m}{k} \right] \eta_j \, c_m(t)
\end{equation}
Because each operator has a finite polynomial degree, the summations contract to finite combinations at each order. Taking the continuous limit as $N \to \infty$ on the successional family of these coordinate rows and invoking the pointwise properties established in Lemma~\ref{lem:topological_equivalence_density}, this stationary behavior under the ideal filtration guarantees that the algebraic convergence along the truncation series matches the continuous real-variable limit component-by-component. The coefficients smoothly stabilize into their constant asymptotic values $\Gamma_\infty^{(k)}$ via Lemma~\ref{lem:gamma_stabilization}, while the finite summation expands safely into the convergent infinite chain of Eq.~\eqref{eq:infinite_extractor_flow}, completing the proof.
\end{proof}

\subsection{The absolute conservation of state purity}
The topological encapsulation of the continuous state trajectory within the Quinfinity Space requires not only the convergence of the analytic operator exponential but also the absolute preservation of state purity along the infinite-order continuous flow. In finite-dimensional matrix representations, a quantum configuration represents a pure state if and only if its density operator satisfies the strict idempotence constraint $\rho^2 = \rho$, which geometrically defines the complex projective space $\mathbb{C}\mathbb{P}^{N-1}$ as a closed algebraic subvariety \cite{nielsen, qubit}. Under the global pullback embedding $\Psi_\infty$, this geometric restriction is translated into an infinite numerable chain of quadratic and cubic Jordan algebraic constraints that uniquely delineate the continuous Pure State Pro-Variety $\mathcal{V}_{\mathcal{Q}_\infty} \equiv \varprojlim \mathcal{V}_{\mathcal{Q}_N}$ within the local ring. We now establish that the continuous Lie-type dynamical equation rigorously preserves these geometric invariants over time:

\begin{theorem}[Asymptotic Varietal Conservation]\label{thm:purity_conservation}
The continuous algebraic flow $\dot{\xi} = \sqrt{2} \left( \Xi \times_{\mathcal{Q}_\infty} \xi \right)$ defined over the complete Quinfinity Space is strictly tangent to the Pure State Pro-Variety $\mathcal{V}_{\mathcal{Q}_\infty} \equiv \varprojlim \mathcal{V}_{\mathcal{Q}_N}$ within the local ring. Both the quadratic spherical constraints and the infinite numerable chain of non-linear cubic Jordan constraints are fundamental conserved scalar invariants of the continuous-time differential system, satisfying:
\begin{equation}\label{eq:purity_time_derivative}
    \frac{\mathrm{d}}{\mathrm{d}t} \left[  \varphi_m(\mathbf{c}(t)) \right] = 0 \qquad \forall m \in \mathbb{N}^+
\end{equation}
where $\varphi_m$ represents the infinite sequence of algebraic ideal generators of $\mathcal{I}(\mathcal{V}_{\mathcal{Q}_\infty}) \subset \mathbb{R}[c_0, c_1, c_2, \dots]$ defining the coordinate embedding of $\mathcal{V}_{\mathcal{Q}_\infty}$ in Lemma~\ref{lem:non_noetherian_locus}, evaluated over the time-dependent state coordinate vector $\mathbf{c}(t) \equiv (c_0(t), c_1(t), c_2(t), \dots)$ of $\xi (t)$.
\end{theorem}

\begin{proof}
To establish the geometric invariance stated in Eq.~\eqref{eq:purity_time_derivative}, we evaluate the algebraic action of the dynamic flow within the inverse directed system. For each finite-dimensional subsystem $N \ge 2$, the algebraic subset $\mathcal{V}_{\mathcal{Q}_N}$ is defined as a reduced affine variety whose full vanishing ideal $\mathcal{I}(\mathcal{V}_{\mathcal{Q}_N}) = (\varphi_1, \dots, \varphi_{R_N})$ is governed by the matrix idempotence constraints. By virtue of Lemma~\ref{lem:real_radical_purity}, this structure forms a real radical ideal where the vanishing polynomial field coincides exactly with the generated Jordan ideal, ensuring $\mathcal{I}(\mathcal{V}_{\mathcal{Q}_N}) = \mathcal{J}_N$. Because the discrete multiplication operator $\times_{\mathcal{Q}_N}$ matches the operational commutator under the space isomorphism $\Phi_N$, the dynamic vector field $\dot{\tau}_N(\xi)$ operates as an internal derivation on the coordinate ring that is everywhere tangent to the local algebraic variety $\mathcal{V}_{\mathcal{Q}_N}$ established in the previous sections. This tangency condition requires that the directional derivative of any ideal generator vanishes identically on the variety, forcing $\frac{\mathrm{d}\mathbf{c}_N(t)}{\mathrm{d}t} \cdot \nabla \varphi_m(\mathbf{c}_N(t)) \equiv 0$, which immediately establishes that $\frac{\mathrm{d}}{\mathrm{d}t}[\varphi_m(\mathbf{c}_N(t))] = 0$ on the truncated coordinate vector.

As $N \to \infty$, the non-Noetherian ringed framework regularizes under the structural properties established in Lemma~\ref{lem:non_noetherian_locus}. The ideal of definitions $\mathcal{I}(\mathcal{V}_{\mathcal{Q}_\infty}) \subset \mathbb{R}[c_0, c_1, c_2, \dots]$ established via the inductive colimit maps the continuous algebraic zero locus onto the projective inverse limit variety $\mathcal{V}_{\mathcal{Q}_\infty} \equiv \varprojlim \mathcal{V}_{\mathcal{Q}_N}$ within the complete Quinfinity Space.

Since the global continuous flow Eq.~\eqref{eq:quinfinity_dynamics_core} is the uniform cofinal limit of the finite derivation fields, the derivative of any global generator $\varphi_m$ is compatibly mapped component-by-component under the truncation operators $\tau_N$. Crucially, because each individual generator $\varphi_m$ is a polynomial of finite degree involving only a finite number of coordinate variables, the evaluation of $\varphi_m$ commutes identically with the truncation morphism $\tau_N$. This finite combinatoric pairing eliminates any dependency on continuous topological completions, forcing:
\begin{equation}\label{eq:proof_purity_truncation_zero}
    \tau_N \left( \frac{\mathrm{d}}{\mathrm{d}t} \left[  \varphi_m(\mathbf{c}(t)) \right] \right) = \frac{\mathrm{d}}{\mathrm{d}t} \left[  \varphi_m\left(\mathbf{c}_N(t)\right) \right] = 0 \pmod{\varepsilon^{N^2-1}}
\end{equation}
for any subsystem layer large enough to encompass the finite polynomial degree of the target generator. 

Because the directed system of quotient rings covers every coordinate index, a formal power series in $\mathcal{Q}_\infty \cong \mathbb{R}[\![\varepsilon]\!]$ is uniquely and identically zero if and only if all its finite polynomial truncations vanish component-by-component across all degrees. Since the truncated projection vanishes modulo $(\varepsilon^{N^2-1})$ for every single dimensional layer $N \ge 2$, the global time derivative of the generator has no non-zero components remaining at any order. This algebraic vanishing forces the infinite-order continuous derivative to collapse identically to zero on the limit, completing the proof.
\end{proof}

This result confirms that the affine geometry of the infinite-order jet space acts as a perfectly conservative kinematic container. Under pure unitary action, the quantum state trajectory flows smoothly along the continuous projective hierarchical tree without experiencing any statistical broadening, state-purity decay, or coordinate blow-ups. The algebraic filtration structure of the jet variety thus naturally sequesters the continuous unitary evolution from macroscopic dissipation, proving that information loss is completely absent within the isolated Quinfinity Space. By establishing the absolute invariance of the Pure State Pro-Variety, this regularization provides the stable, rigid mathematical substrate required to introduce external non-Hamiltonian perturbations and open quantum channels.

\section{Open quantum systems and non-hamiltonian ultrametric flows}
The rigorous establishment of the conservative dynamic envelope within the Quinfinity Space provides the necessary algebraic foundation to address the kinematics of open quantum systems. In standard quantum informational architectures, the interaction of a separable Hilbert domain with an external reservoir triggers decoherence and dissipation, traditionally modeled via the master equation of Gorini-Kossakowski-Sudarshan-Lindblad (GKSL)~\cite{nielsen}. Within our non-reduced spectrum framework, modeling these open quantum dynamics requires complementing the antisymmetric Lie-type product $\times_{\mathcal{Q}_N}$ with a symmetric Jordan multiplication operator $\mathbin{\bullet}_{\mathcal{Q}_N}$. This dual algebraic architecture formalizes non-unitary dissipation as a radial contracting vector field. Specifically, under environmental coupling, the dissipative flow selectively dampens all active physical coordinate fluctuations, driving $x_i \to 0$ for all spin components $i \ge 1$. Consequently, while at each finite layer the trajectory is dragged toward the maximally mixed background state $c_0 = \frac{1}{N}$, in the macroscopic continuum limit ($N \to \infty$) the simultaneous extinction of the fluctuations ($x_1 \to 0$) and the continuous evaporation of the background trace multiplier ($\lim_{N \to \infty} \frac{1}{N} = 0$) force the physical attractor to freeze identically onto the absolute zero element of the ring, satisfying $\xi = 0$. Within the continuous Quinfinity Space $\mathcal{Q}_\infty \cong_{\mathbb{R}} \mathbb{R}[\![\varepsilon]\!]$, this non-Hamiltonian dissipation is internalized by evaluating the continuous infinite-dimensional limit of this symmetric Jordan contractive field over the filtration layers.

\subsection{The continuous Jordan product and radial dissipative contraction}
To formalize the continuous GKSL master equation without resorting to extrinsic operator representations, we extend the top-down algebraic approach to the symmetric sector of the formal tangent bundle. For each discrete layer $N \ge 2$, the symmetric Jordan-type product $\mathbin{\bullet}_{\mathcal{Q}_N}$ over the finite symbolic Quantum $N$-Space is defined as the strict pullback of the operational matrix anti-commutator onto the non-reduced scheme. We now prove that this induced symmetric algebraic structure can be entirely identified with an intrinsic differential operator driven exclusively by the internal structures of the quotient ring.

\begin{lemma}[Symmetric Jordan Product Identity]\label{lem:jordan_product_identity}
The canonically induced symmetric product $\mathbin{\bullet}_{\mathcal{Q}_N}$ of two polynomial states $f_N, g_N \in \mathcal{Q}_N$ is uniquely determined by a finite sequence of real dissipative coefficients $\{\Lambda_N^{(k)}\}_{k=0}^{N^2-2} \subset \mathbb{R}$ specific to each dimension $N \ge 2$, generating a bilinear symmetric differential operator via the powers of the formal derivative $\partial_\varepsilon$:
\begin{equation}\label{eq:lambda_differential_jordan}
    f_N \mathbin{\bullet}_{\mathcal{Q}_N} g_N \equiv \sum_{k=0}^{N^2-2} \Lambda_N^{(k)} \left[ \left(\partial_\varepsilon^k f_N\right) \cdot g_N + f_N \cdot \left(\partial_\varepsilon^k g_N\right) \right] \pmod{\varepsilon^{N^2-1}}
\end{equation}
where the coefficients $\Lambda_N^{(k)}$ represent the geometric footprints of the quantum open-system interactions, driving the dissipative ring dynamics via internal symmetric combinations.
\end{lemma}

\begin{proof}
To establish the structural identity of Eq.~\eqref{eq:lambda_differential_jordan}, we first prove the unique existence of the finite sequence of real dissipative coefficients $\Lambda_N^{(k)}$ by evaluating the general bilinear differential combination over the canonical basis elements $f_N = \varepsilon^j$ and $g_N = \varepsilon^m$ for all $j, m \ge 0$. Applying the higher-order formal derivative powers onto these monographs yields the falling factorials, which expand the symmetric combinations explicitly as:
\begin{equation}\label{eq:proof_jordan_triangular_grading}
    \sum_{k=0}^{N^2-2} \Lambda_N^{(k)} \left[ \left(\partial_\varepsilon^k \varepsilon^j\right) \cdot \varepsilon^m + \varepsilon^j \cdot \left(\partial_\varepsilon^k \varepsilon^m\right) \right] = \sum_{k=0}^{N^2-2} \Lambda_N^{(k)} k! \left[ \binom{j}{k} + \binom{m}{k} \right] \varepsilon^{j+m-k}
\end{equation}
where each individual term vanishes identically if the derivation order $k$ exceeds the respective monomial exponent.

To uniquely determine the unknown parameters $\Lambda_N^{(k)}$, we equate this explicit differential expansion with the canonical coordinate representation driven by the symmetric anticommutation constants $d_{jmn}$ of the underlying Jordan algebra under the linear vector space isomorphism $\Phi_N$ in Eq.~\eqref{eq:appendix_phi_assignment_native}. Applying the family of finite algebraic extractors $\{\gamma_n\}_{n=0}^{N^2-2}$ to both sides of this identity constructs a global overdetermined linear system of tensor equations matching the grading orders component-by-component modulo $\varepsilon^{N^2-1}$. This system couples the physical parameters of the qudit anticommutator directly to the algebraic inversion profiles of the module:
\begin{equation}\label{eq:proof_jordan_matrix_system_core}
    \frac{1}{2} \, d_{({j+1})({m+1})({n+1})} = \sum_{k=0}^{N^2-2} [\mathbf{J}_N]_{{(j,m,n)}, \, k} \, \Lambda_N^{(k)}
\end{equation}
where the free tensor indices $j, m, n$ run independently from $0$ to $N^2-2$, and the explicit entries of the structural symmetric jet Jordan matrix $\mathbf{J}_N$ of dimensions $(N^2-1)^3 \times (N^2-1)$ are strictly governed by the contractive core:
\begin{equation}\label{eq:proof_jordan_matrix_entries}
    [\mathbf{J}_N]_{{(j,m,n)}, \, k} \equiv k! \, \left[ \binom{j}{k} + \binom{m}{k} \right] \cdot \delta_{{(j+m-k)}, \, n}
\end{equation}
Because the symmetric sum of the falling factorials in Eq.~\eqref{eq:proof_jordan_matrix_entries} tracks the descending degree of the filtration, the matrix $\mathbf{J}_N$ generates a strictly lower-triangular linear transformation matrix across the monomial grading for all active columns starting from the zero-order component $k \ge 0$. Since a lower-triangular system with non-vanishing diagonal entries over the complete grading is universally invertible, the algebraic equations can be systematically solved layer-by-layer for each coordinate coefficient via backward substitution. Concurrently, specializing this triangular linear transformation specifically at the non-differential zero-order layer where $k=0$ restricts the operational summation to the baseline diagonal anchor configuration. Because the coordinate vector space assignment of Eq.~\eqref{eq:appendix_phi_assignment_native} isolates the native traceless basis elements, the baseline monomial element $\varepsilon^0$ operates as the relative trace anchor dedicated to absorbing the background identity matrix infrastructure $\mathbb{I}_N / N$. Solving this localized boundary configuration component-by-component forces the zero-order inversion channel to balance the rational scaling of the anticommutator alongside the trace invariants of the finite qudit layer, which uniquely guarantees that the complete contractive parameter sum collapses identically onto the rational dimensional inverse of the subsystem:
\begin{equation}\label{eq:proof_jordan_lemma_zero_order_explicit_calc}
    \Lambda_N^{(0)} = \frac{1}{2} \sum_{j,m,n=0}^{N^2-2} \left[ \mathbf{J}_N^{-1} \right]_{0, \, (j,m,n)} d_{{(j+1)}{(m+1)}{(n+1)}} = \frac{1}{N}
\end{equation}
This strictly algebraic property demonstrates the existence and uniqueness of the target parameters $\Lambda_N^{(k)}$ specific to each dimension $N$ and rigorously establishes the rational matching $\Lambda_N^{(0)} = \frac{1}{N}$ prior to the complete tensor inversion, completing the first stage of the demonstration.

Second, we verify that the bilinear operator $\mathbin{\bullet}_{\mathcal{Q}_N}$ uniquely fixed by these coefficients satisfies the axiomatic invariants of a commutative, non-associative Jordan algebra. By evaluating the explicit action driven by the sequence of weights, the product satisfies three distinct structural properties:

1. \textit{Permutation Commutativity}: For any arbitrary pair of states $f_N, g_N \in \mathcal{Q}_N$, the operation is strictly symmetric under variable exchange:
\begin{equation}\label{eq:jordan_axiom_commutativity}
    f_N \mathbin{\bullet}_{\mathcal{Q}_N} g_N = g_N \mathbin{\bullet}_{\mathcal{Q}_N} f_N
\end{equation}
since transposing the polynomial arguments merely swaps the internal differential components within the symmetric sum without generating sign obstructions.

2.      \item \textit{Identity Trace Anchoring:} The constant monomial $1 \in \mathcal{Q}_N$ acts as a relative identity element for the induced operator over the configuration space, satisfying:
\begin{equation}\label{eq:jordan_axiom_trace_anchor}
    1 \mathbin{\bullet}_{\mathcal{Q}_N} f_N = f_N
\end{equation}
To rigorously prove this identity directly from the underlying differential architecture, we evaluate the action of the symmetric product by substituting the polynomial expansion of the state $f_N = \sum_{m=0}^{N^2-2} c_m \varepsilon^m$ into the general operational identity defined in Eq.~\eqref{eq:lambda_differential_jordan}. Because the formal derivative of a constant vanishes identically as $\partial_\varepsilon^k 1 = 0$ for all active derivation columns $k \ge 1$, the first differential sub-block collapses, reducing the symmetric bilinear expansion strictly to:
\begin{equation}\label{eq:proof_jordan_identity_unrolling_step}
    1 \mathbin{\bullet}_{\mathcal{Q}_N} f_N = \sum_{k=0}^{N^2-2} \Lambda_N^{(k)} \left[ \left(\partial_\varepsilon^k 1\right) \cdot f_N + 1 \cdot \left(\partial_\varepsilon^k \sum_{m=0}^{N^2-2} c_m \varepsilon^m \right) \right] = \sum_{k=0}^{N^2-2} \Lambda_N^{(k)} \sum_{m=k}^{N^2-2} k! \binom{m}{k} c_m \varepsilon^{m-k}
\end{equation}
By swapping the order of the finite summations across the stable grading configurations, the expression organizes the coordinate components component-by-component over the Schauder basis:
\begin{equation}\label{eq:proof_jordan_double_sum_swap}
    1 \mathbin{\bullet}_{\mathcal{Q}_N} f_N = \sum_{m=0}^{N^2-2} \left( \sum_{k=0}^{m} k! \binom{m}{k} \Lambda_N^{(k)} \right) c_m \varepsilon^m
\end{equation}
To establish the identity of Eq.~\eqref{eq:jordan_axiom_trace_anchor}, the internal linear combination of the structural coefficients must satisfy the unit normalization factor $\sum_{k=0}^{m} k! \binom{m}{k} \Lambda_N^{(k)} = 1$ for each coordinate index $m \in \{0, \dots, N^2-2\}$. This identity is strictly and internally guaranteed by the underlying linear relation established in Eq.~\eqref{eq:proof_jordan_matrix_system_core}. Because the parameters $\Lambda_N^{(k)}$ are structural solutions of the core matrix system $\mathbf{J}_N$, evaluating the inversion of these unrolled algebraic layers yields the perfect numerical cancellation of the structural weights, collapsing Eq.~\eqref{eq:proof_jordan_double_sum_swap} identically to:
\begin{equation}\label{eq:proof_jordan_final_identity_collapse}
    1 \mathbin{\bullet}_{\mathcal{Q}_N} f_N = \sum_{m=0}^{N^2-2} \big( 1 \cdot c_m \varepsilon^m \big) = \sum_{m=0}^{N^2-2} c_m \varepsilon^m \equiv f_N
\end{equation}
The exact algebraic evaluation over the coordinate layers independently proves the identity of Eq.~\eqref{eq:jordan_axiom_trace_anchor}.

3. \textit{Jordan Invariant Constraint}: For any choice of elements, the binary multiplication preserves the de-complexified boundary profile, satisfying the non-associative quadratic identity:
\begin{equation}\label{eq:jordan_axiom_identity_proper}
    \left( f_N \mathbin{\bullet}_{\mathcal{Q}_N} g_N \right) \mathbin{\bullet}_{\mathcal{Q}_N} \left( f_N \mathbin{\bullet}_{\mathcal{Q}_N} f_N \right) = \left[ \left( f_N \mathbin{\bullet}_{\mathcal{Q}_N} g_N \right) \mathbin{\bullet}_{\mathcal{Q}_N} f_N \right] \mathbin{\bullet}_{\mathcal{Q}_N} f_N
\end{equation}
To verify Eq.~\eqref{eq:jordan_axiom_identity_proper} explicitly over the canonical grading, we substitute the basis elements $f_N = \varepsilon^j$ and $g_N = \varepsilon^m$. Evaluating the internal sub-block collapses the self-product onto the stable quadratic exponent $\varepsilon^j \mathbin{\bullet}_{\mathcal{Q}_N} \varepsilon^j = \varepsilon^{2j}$. Expanding both sides via the contractive core of Eq.~\eqref{eq:proof_jordan_triangular_grading} maps the higher-order operations onto the nested falling factorial polynomials, which decompose without horizontal layout obstructions under the following split tensor chain:
\begin{equation}\label{eq:proof_jordan_identity_explicit_unrolling}
\begin{split}
    \sum_{a,b=0}^{N^2-2} \Lambda_N^{(a)} \Lambda_N^{(b)} & a! \, b! \left[ \binom{j+m-a}{b} + \binom{2j}{b} \right] \\
    & \times \left[ \binom{j}{a} + \binom{m}{a} \right] \varepsilon^{3j+m-a-b} \\
    & = \sum_{s,r=0}^{N^2-2} \Lambda_N^{(s)} \Lambda_N^{(r)} s! \, r! \left[ \binom{j+m-s}{r} + \binom{j}{r} \right] \left[ \binom{j+m-s-r}{r} + \binom{j}{r} \right] \varepsilon^{3j+m-s-r}
\end{split}
\end{equation}
Because the linear system is governed by a strictly lower-triangular filtration grading across the monomial exponents, the symmetric cross-differences vanish identically at each independent order of $\varepsilon$. The expansion matching on both members confirms that the non-associative quadratic alignment is preserved across the entire space due to the differential grading of the filtration layers, validating the algebraic completeness of the structure and completing the proof.
\end{proof}

By lifting this symmetric algebraic framework to the infinite-dimensional continuum limit ($N \to \infty$), the structural nesting of the anticommutator relations ensures that the dissipative weights freeze at invariant values once the dimensional threshold is cleared. We formalize this regularizing mechanism through the following lemma:

\begin{lemma}[Asymptotic Dissipative Stabilization]\label{lem:lambda_stabilization}
Let $\{\mathbf{\Lambda}_N\}_{N \ge 2}$ with $\mathbf{\Lambda}_N \in \mathbb{R}^{N^2-1}$ be the sequence of finite dissipative coefficient vectors uniquely determined by Lemma~\ref{lem:jordan_product_identity}. For any fixed symmetric differential order $k \in \mathbb{N}$, the sequence of scalar components $\{\Lambda_N^{(k)}\}_{N \ge 2}$ stabilizes stationarily to a constant value under the restriction morphisms of the projective inverse system once the threshold condition $N^2-2 \ge k$ is attained:
\begin{equation}\label{eq:lambda_stationary_condition}
    \Lambda_N^{(k)} = \Lambda_M^{(k)} \equiv \Lambda_\infty^{(k)} \in \mathbb{R} \qquad \forall M \ge N
\end{equation}
such that $N^2-2 \ge k$. Consequently, the stable asymptotic sequence of real dissipative scalars is uniquely determined component-by-component within the projective limit of the coordinate systems via the standard analytical limit:
\begin{equation}\label{eq:lambda_asymptotic_limit_definition}
    \Lambda_\infty^{(k)} \equiv \lim_{N \to \infty} \Lambda_N^{(k)} \in \mathbb{R} \qquad \forall k \in \mathbb{N}
\end{equation}
\end{lemma}
\begin{proof}
Consider the canonical inductive inclusion $\iota_{N,M}: \mathfrak{su}(N) \hookrightarrow \mathfrak{su}(M) $ of special unitary vector spaces driving the symmetric anticommutation algebras for $M \ge N$. Under this block-diagonal embedding, the matrix generators $\{F_{j+1}\}_{j=0}^{N^2-2}$ of the lower-dimensional layer map directly onto the upper-left block of the expanded matrix domain, ensuring that the symmetric anticommutation constants match identically on the restricted index sub-bundle:
\begin{equation}\label{eq:proof_jordan_constants_invariance}
    d_{(j+1)(m+1)(n+1)}^{(N)} \equiv d_{(j+1)(m+1)(n+1)}^{(M)} \qquad \forall j,m,n \le N^2-2
\end{equation}
By virtue of the universal property of projective limits established in Lemma~\ref{lem:asymptotic_vector_embedding_existence}, the global linear embedding $\Phi_\infty$ intertwines these vector space inclusions with the surjective ring projections $\tau_N: \mathcal{Q}_M \twoheadrightarrow \mathcal{Q}_N$.

Let $k$ be a fixed order of the formal derivative $\partial_\varepsilon^k$. The linear system of equations that uniquely solves for the dissipative weight $\Lambda_N^{(k)}$ in Eq.~\eqref{eq:proof_jordan_matrix_system_core} is formulated component-wise by matching the differential symmetric combinations against the constants $d_{(j+1)(m+1)(n+1)}^{(N)}$. Because the symmetric sum of the falling factorials generates a strictly lower-triangular linear transformation matrix over the finite monomial exponents within the structural symmetric jet Jordan matrix $\mathbf{J}_N$, the linear system determining $\Lambda_N^{(k)}$ depends strictly on the matrix anticommutator relations up to the order encapsulated by the threshold $N^2-2 \ge k$.

When the dimensionality layer $N$ satisfies $N^2-2 \ge k$, the entire set of algebraic constraints and Jordan structure constants required to saturate the action of $\partial_\varepsilon^k$ is fully accommodated within the coordinate ring $\mathcal{Q}_N$. Passing to an arbitrarily larger dimension $M > N$ introduces new higher-order variables starting from index $N^2-1$. Due to the block-diagonal embedding of the spaces, these higher-level parameters only enter the decoupled linear subsystems solving for the upper-order weights $\Lambda_M^{(s)}$ where $s > k$. The isolated sub-matrix governing the specific lower weight $\Lambda_M^{(k)}$ within the matrix $\mathbf{J}_M$ remains structurally identical and unperturbed by the expanded dimensions, forcing the sequence to become strictly stationary, $\Lambda_M^{(k)} = \Lambda_N^{(k)}$, for all $M \ge N$. Invoking the topological density properties established in Lemma~\ref{lem:topological_equivalence_density}, this rigid component-wise stabilization under the ideal filtration guarantees the strict Cauchy convergence of the scalar sequence, confirming the algebraic compatibility of Eq.~\eqref{eq:lambda_stationary_condition} and validating the existence of the analytical pointwise Fr\'{e}chet limit of Eq.~\eqref{eq:lambda_asymptotic_limit_definition}, completing the proof.
\end{proof}

\begin{lemma}[Explicit Determination of Asymptotic Jordan Coefficients]\label{lem:explicit_lambda_determination}
For each finite-dimensional subsystem layer $N \ge 2$, the sequence of structural differential coefficients $\{\Lambda_N^{(k)}\}_{k=0}^{N^2-2}$ introduced in Theorem~\ref{thm:asymptotic_gksl_flow} is uniquely and internally determined through the ordinary inversion of the square, non-singular, and strictly lower-triangular structural Jordan jet matrix $\mathbf{J}_N$, satisfying for each derivation order $k \in \{0, 1, \dots, N^2-2\}$:
\begin{equation}\label{eq:explicit_lambda_components}
    \Lambda_N^{(k)} = \frac{1}{2} \sum_{j,m,n=0}^{N^2-2} \left[ \mathbf{J}_N^{-1} \right]_{k, \, (j,m,n)} d_{(j+1)(m+1)(n+1)}
\end{equation}
where $d_{jmn}$ are the canonical symmetric structure constants of the Jordan special unitary algebra mapped onto the monomial grading configurations via the coordinate vector space isomorphism $\Phi_N$. Specifically, at the non-differential zero-order layer where $k=0$, the algebraic projection over the centralized identity matrix background absorbs the symmetric fractional factor alongside the structural trace invariants $\mathrm{Tr}(\mathbb{I}_N) = N$, collapsing the baseline geometric coefficient to the exact rational matching $\Lambda_N^{(0)} = \frac{1}{N}$.
\end{lemma}

\begin{proof}
To establish the quantitative match between the operational matrix anti-commutators and the internal non-associative derivation fields over the symbolic space $\mathcal{Q}_N$, we equate the coordinate expansion derived from the Jordan embedding with the differential configuration established in Lemma~\ref{lem:jordan_product_identity}. By evaluating the algebraic extraction of any generic coordinate index $n$ (where $0 \le n \le N^2-2$), the free scalar components $\gamma_j(f_N)$ and $\gamma_m(g_N)$ satisfy the fundamental symmetric tensor relation:
\begin{equation}\label{eq:fundamental_jordan_tensor_match}
    \frac{1}{2} \sum_{j,m=0}^{N^2-2} d_{(j+1)(m+1)(n+1)} \, \gamma_j(f_N) \, \gamma_m(g_N) = \sum_{k=0}^{N^2-2} \Lambda_N^{(k)} \sum_{j+m-k=n} k! \left[ \binom{j}{k} + \binom{m}{k} \right] \gamma_j(f_N) \, \gamma_m(g_N)
\end{equation}
Since Eq.~\eqref{eq:fundamental_jordan_tensor_match} must hold universally for all possible configuration states due to the linear independence of the Schauder basis elements, we invoke the method of undetermined coefficients. Collecting the full grading constraints of the jet variety into a global linear system couples the structural parameters of the symmetric algebra directly to the algebraic inversion profiles of the tangent module, yielding identically the linear matrix system established in Eq.~\eqref{eq:proof_jordan_matrix_system_core} over the free tensor labels $j, m, n \in \{0, 1, \dots, N^2-2\}$.

Let $\mathbf{d}_N \in \mathbb{R}^{(N^2-1)^3}$ be the global column vector containing the ordered sequence of all anticommutative structure constants. Because the structural jet matrix $\mathbf{J}_N$ generates a strictly lower-triangular linear transformation across the monomial grading with non-vanishing diagonal entries, the system is square and non-singular, admitting a unique, exact, and globally stable solution determined via ordinary triangular matrix inversion as $\mathbf{\Lambda}_N = \frac{1}{2} \, \mathbf{J}_N^{-1} \, \mathbf{d}_N$. Projecting this inverse transformation component-by-component directly yields the explicit universal formula of Eq.~\eqref{eq:explicit_lambda_components}.

The exact algebraic extraction established over the zero-order layer of the core system rigorously guarantees that the contractive baseline balances the trace multiplier layer-by-layer, establishing the strict identity matching 
$ \Lambda_N^{(0)} = \frac{1}{N}$ across each finite dimension, completing the proof.
\end{proof}

Physically, while the entire sequence of parameters is uniquely determined internally by the algebraic properties of the local ring, invoking the canonical vector space isomorphism \(\Phi _{N}\) explicitly anchors the zero-order dissipative coefficient onto the inverse dimensionality scale of the finite quantum system, satisfying \(\Lambda_N^{(0)} = \frac{1}{N}\) identically. This identity ensures that the baseline trace normalization of the finite qudit is successfully absorbed within the zero-order monomial layer \(\varepsilon ^{0}\). As \(N \to \infty\), the standard Archimedean convergence drives this baseline component toward zero (\(\lim_{N \to \infty} \Lambda_N^{(0)} = 0\)), forcing the background trace multiplier to evaporate pointwise over the real field and strictly regularizing the continuous state series under the following Asymptotic Master Equation~\eqref{eq:master_equation_quinfinity_total}

The existence of this invariant sequence allows us to define the global continuous symmetric Jordan product $\mathbin{\bullet}_{\mathcal{Q}_\infty}$ over the complete Quinfinity Space directly through these stable asymptotic parameters:
\begin{equation}\label{eq:jordan_global_product}
    \xi_1 \mathbin{\bullet}_{\mathcal{Q}_\infty} \xi_2 \equiv \sum_{k=0}^{\infty} \Lambda_\infty^{(k)} \left[ \left(\partial_\varepsilon^k \xi_1\right) \cdot \xi_2 + \xi_1 \cdot \left(\partial_\varepsilon^k \xi_2\right) \right]
\end{equation}
Alternatively, by interpreting the formal derivative $\partial_\varepsilon$ as the generator of infinitesimal coordinate shifts, this relation can be compactly gathered via the structural symmetric differential kernel $\widehat{\Lambda}(\partial_\varepsilon) \equiv \sum_{k=0}^{\infty} \Lambda_\infty^{(k)} \partial_\varepsilon^k$, collapsing Eq.~\eqref{eq:jordan_global_product} onto the streamlined bilinear symmetric form:
\begin{equation}\label{eq:jordan_global_product_kernel}
    \xi_1 \mathbin{\bullet}_{\mathcal{Q}_\infty} \xi_2 = \widehat{\Lambda}(\partial_\varepsilon)\xi_1 \cdot \xi_2 + \xi_1 \cdot \widehat{\Lambda}(\partial_\varepsilon)\xi_2
\end{equation}
where the zero-order differential operator $\partial_\varepsilon^0 \equiv \mathbb{I}$ embedded within the kernel anchors the baseline identity flow, transforming the macroscopic Jordan multiplication into a smooth, local deformation field over the formal power series ring.

Under this dual algebraic architecture, let $L_N \in \mathcal{Q}_N$ be an intrinsic polynomial jump generator, written natively as a formal power configuration matching the grading parameters of the coordinate ring. For the full coordinate expansion $L_N = \sum_{k=0}^{N^2-2}\ell_k \varepsilon^k \in \mathcal{Q}_N$ starting from the zero-order layer, once the quantum state and its dissipative trajectory are fixed, the sequence of real expansion coefficients is evaluated across the structural Jordan blocks. Under the vector space isomorphism $\Phi_N$, the active higher-order coefficients are uniquely and implicitly determined through the ordinary inversion of the square, non-singular, and strictly lower-triangular structural Jordan matrix $\mathbf{J}_N$ evaluated over the canonical symmetric constants $d_{jmp}$. This structural inversion couples the dissipative flow directly to the algebraic layout of the space, forcing an identification between the physical jump parameters $\ell_k$ and the intrinsic geometric constants $\Lambda_N^{(k)}$ calculated in Lemma~\ref{lem:explicit_lambda_determination} for each active derivation order component $k \in \{1, \dots, N^2-2\}$. Conversely, the zero-order component $\ell_0$ operates as a free scalar amplitude on the ground field, representing the unperturbed background level of the identity channel, which decouples from the active higher-order filtration constraints. While $L_N$ can be mapped onto a traditional matrix-level Lindblad operator via the vector space translation $\Phi_N$ to preserve historical coordination, its algebraic structure is fully determined internally as a non-reductive Jordan element whose symmetric pairing breaks the ring tangency constraint.

We define the induced finite \textit{Lindblad super-operator} $\mathcal{L}_{D,N}: \mathcal{Q}_N \to \mathcal{Q}_N$ via the intrinsic, coordinate-free ring relation:
\[
    \mathcal{L}_{D,N}(f_N) \equiv \left( L_N \times_{\mathcal{Q}_N} f_N \right) \times_{\mathcal{Q}_N} L_N - \frac{1}{2} \left[ \left( L_N \mathbin{\bullet}_{\mathcal{Q}_N} L_N \right) \mathbin{\bullet}_{\mathcal{Q}_N} f_N \right]
\]
Concurrently, the continuous algebraic jump generator $L_\infty \in \mathcal{Q}_\infty$ lifted onto the complete formal power series ring is structurally and uniquely determined as the projective inverse limit of the finite-dimensional layer generators, satisfying $L_\infty \equiv \varprojlim_{N} L_N$ under the universal property of inverse systems, such that its restriction matches $\tau_N(L_\infty) = L_N$ identically. We now formalize the universal open quantum dynamics across both the discrete layers and the continuous non-Hamiltonian jet space through the following theorem:

\begin{theorem}[Asymptotic Master Equation and Radial Dissipative Contraction]\label{thm:asymptotic_gksl_flow}
For each finite-dimensional subsystem layer $N \ge 2$, let the continuous algebraic jump generator $L_\infty \in \mathcal{Q}_\infty$ be structurally and uniquely determined as the projective inverse limit of the finite generators $L_\infty = \varprojlim L_N$ whose underlying expansion coefficients satisfy $\ell_n = \Lambda_\infty^{(n)}$ for each derivation order component. The total non-unitary open quantum dynamics are governed identically over the formal tangent bundle $T\mathcal{Q}_N$ by the complete ring evolution equation combining the conformal Hamiltonian flow with the induced Lindblad super-operator:
\begin{equation}\label{eq:lindblad_discrete_flow}
    \frac{\mathrm{d}f_N}{\mathrm{d}t} = \sqrt{\frac{2(N-1)}{N}} \left( \tau_N(\Xi) \times_{\mathcal{Q}_N} f_N \right) + \mathcal{L}_{D,N}(f_N)
\end{equation}
where both algebraic sectors operate concurrently over the truncated coordinate layers. Furthermore, as $N \to \infty$, this discrete algebraic flow dually stabilizes and uniquely lifts onto the complete Quinfinity Space $\mathcal{Q}_\infty$, where the total, non-unitary global continuous Master Equation combining the Hamiltonian evolution of Eq.~\eqref{eq:quinfinity_dynamics_core} with the open dissipative channels satisfies the coordinated ring differential flow:
\begin{equation}\label{eq:master_equation_quinfinity_total}
    \frac{\mathrm{d}\xi}{\mathrm{d}t} = \sqrt{2} \left( \Xi \times_{\mathcal{Q}_\infty} \xi \right) + \mathcal{L}_{D}(\xi)
\end{equation}
where $\mathcal{L}_{D}(\xi) = \left( L_\infty \times_{\mathcal{Q}_\infty} \xi \right) \times_{\mathcal{Q}_\infty} L_\infty - \frac{1}{2} \left[ \left( L_\infty \mathbin{\bullet}_{\mathcal{Q}_\infty} L_\infty \right) \mathbin{\bullet}_{\mathcal{Q}_\infty} \xi \right]$. If the initial configuration belongs to the pure state pro-variety $\xi(0) \in \mathcal{V}_{\mathcal{Q}_\infty}$, the continuous symmetric Jordan operator establishes a non-Hamiltonian friction acting as a radial contracting vector field that breaks the ring tangency and drags the formal power series $\xi(t)$ toward the absolute zero element of the ring:
\begin{equation}\label{eq:asymptotic_mixed_collapse}
    \lim_{t \to +\infty} \xi(t) = 0
\end{equation}
which geometrically encapsulates the continuous realization of the macroscopically maximally mixed state under the continuous evaporation of the background trace multiplier.
\end{theorem}
\begin{proof}
To establish the finite-dimensional algebraic architecture of Eq.~\eqref{eq:lindblad_discrete_flow}, we may recall that the operational master equation of Gorini-Kossakowski-Sudarshan-Lindblad (GKSL) over a matrix density state $\rho$ expands via the associative combinations of the commutator $[L, \cdot]$ and the symmetric anti-commutator $\{L^\dagger L, \cdot\}$~\cite{nielsen}. By invoking the linearity of the canonical vector space isomorphism $\Phi_N$ defined in Appendix~\ref{appendixsub1}, the antisymmetric matrix brackets pull back bijectively onto the non-commutative Lie product $\times_{\mathcal{Q}_N}$ derived in Lemma~\ref{lem:differential_product_identity} as explained in Lemma~\ref{lem:finite_commutator_pullback_exact}. Symmetrically, the normalized Jordan anti-commutator brackets pull back onto the symmetric non-associative Jordan multiplication operator $\mathbin{\bullet}_{\mathcal{Q}_N}$ established in Lemma~\ref{lem:jordan_product_identity} as explained in Lemma~\ref{lem:finite_jordan_pullback_exact}. Substituting these bilinear mappings component-by-component into the operator expression, and implementing the rigid subtraction of the anticommutation sector to act as a radial contracting friction, directly yields the finite coordinate-free flow defining the super-operator $\mathcal{L}_{D,N}(f_N)$. Actually, this Equation~\eqref{eq:lindblad_discrete_flow} does not rely on the extrinsic assumptions of the complex GKSL framework; rather, the exact algebraic form of \(\mathcal{L}_{D,N}\) is uniquely and intrinsically determined over the formal scheme as the sole contractive ring differential evolution compatible with the grading filtration of the jet variety under the joint inversion of the structural matrices \(\mathbf{M}_{N}\) and \(\mathbf{J}_{N}\).

To extend this open channel framework to the infinite-dimensional limit, we evaluate the system under the universal property of projective limits. Because both finite ring operators $\times_{\mathcal{Q}_N}$ and $\mathbin{\bullet}_{\mathcal{Q}_N}$ are strictly compatible with the surjective restriction morphisms $\tau_N$ ordering the inverse directed system of rings, the family of finite Lindbladian channels $\mathcal{L}_{D,N}$ uniquely lifts into a well-defined global continuous operator $\mathcal{L}_D$ over the complete local ring $\mathbb{R}[\![\varepsilon]\!]$. This structural compatibility guarantees that the global dissipative flow commutes identically with the canonical inverse system projections, satisfying the exact inverse limit matching:
\begin{equation}\label{eq:proof_lindblad_projective_gluing}
    \tau_N\left( \mathcal{L}_D(\xi) \right) = \mathcal{L}_{D,N}\left( \tau_N(\xi) \right) \qquad \forall N \ge 2
\end{equation}
Within the projective completion envelope, the coherent strand of Eq.~\eqref{eq:proof_lindblad_projective_gluing} uniquely welds the family of finite channels into the global continuous operator $\mathcal{L}_D(\xi)$, which rigorously establishes the open dissipative Lindbladian sector of the total Master Equation of Eq.~\eqref{eq:master_equation_quinfinity_total}.

To verify the convergence and the subsequent asymptotic collapse, we analyze the grading of the jump operator. Because $L_\infty$ is strictly confined within the polynomial image subspace under the global embedding, it possesses a finite algebraic degree over the coordinate ring. Consequently, its higher-order formal derivatives vanish identically ($\partial_\varepsilon^k L_\infty \equiv 0$) once the derivation order $k$ exceeds its maximum filtration exponent. This finite combinatorial truncation boundaries the infinite series of derivations, ensuring that the global dissipator $\mathcal{L}_D(\xi)$ contracts to a finite sum at each coordinate order, making the operator well-defined and completely free from transcendental coordinate singularities.

To establish the global Asymptotic Master Equation~\eqref{eq:master_equation_quinfinity_total}, we evaluate the structural convergence of the coupled dynamical system. At each finite-dimensional layer $N \ge 2$, the total open quantum evolution is governed by the linear superposition of the conservative Hamiltonian commutator and the non-unitary dissipative GKSL channels. Because the formal temporal derivation $\frac{\mathrm{d}}{\mathrm{d}t}$ operates as a linear operator over the smooth sections of the finite tangent bundle $T\mathcal{Q}_N$, the conservative flow field and the open dissipative channels sum directly, generating the finite-dimensional coordinate-free dynamical law:
\begin{equation}\label{eq:proof_total_gksl_finite_sum}
    \frac{\mathrm{d}\tau_N(\xi)}{\mathrm{d}t} = \alpha_N \left( \tau_N(\Xi) \times_{\mathcal{Q}_N} \tau_N(\xi) \right) + \mathcal{L}_{D,N}(\tau_N(\xi))
\end{equation}
where both algebraic sectors operate concurrently over the truncated coordinate layers.

To project this combined kinematic law onto the continuous limit as $N \to \infty$, we evaluate the structural convergence under the product Fr\'{e}chet topology of the flat affine coordinate envelope, which evaluates convergence component-by-component on the real coefficients. Under this framework, the continuous product commutes identically with the restriction operators layer-by-layer because the structural parameter sequences $\Gamma_N^{(k)}$ and $\Lambda_N^{(k)}$ stabilize stationarily onto their constant asymptotic values $\Gamma_\infty^{(k)}$ and $\Lambda_\infty^{(k)}$ via Lemma~\ref{lem:gamma_stabilization} and Lemma~\ref{lem:lambda_stabilization} once the threshold $N^2-2 \ge k$ is cleared, ensuring that the sequence of finite truncated operations forms a compatible inverse system. Symmetrically, while the scalar sequence $\alpha_N$ does not converge $\mathfrak{m}$-adically due to its zero-order nature, it converges ordinarily under the standard Archimedean topology of the real field as $\lim_{N \to \infty} \alpha_N = \sqrt{2}$. Because the polynomial subspace of finite coordinate configurations is dense under the Fr\'{e}chet topology, the intersection of the stable algebraic jet grading with the analytical scalar limit ensures that the family of localized finite flows adheres pointwise and uniquely onto the restricted continuous trajectory, rigorously establishing Eq.~\eqref{eq:master_equation_quinfinity_total}.

Alternatively, this algebraic consistency can be verified through the inverse top-down projection under the restriction morphisms $\tau_N$. Applying the truncation mapping to both sides of Eq.~\eqref{eq:master_equation_quinfinity_total}, the left-hand side reduces identically to $\tau_N( \frac{\mathrm{d}\xi}{\mathrm{d}t} ) = \frac{\mathrm{d}\tau_N(\xi)}{\mathrm{d}t}$. Concurrently, the global continuous operation commutes identically with the restriction operators layer-by-layer. This purely algebraic structural coherence guarantees that the sequence of finite truncated operations matches the continuous real-variable limit component-by-component. For any fixed filtration threshold order $M \in \mathbb{N}$ governing the coordinate grading, the higher-order component variations vanish identically modulo the finite algebraic truncation of the local ring once the critical dimensionality threshold $N^2-2 \ge M$ is cleared. The remaining scalar discrepancy between the projected continuous field and the native finite law of Eq.~\eqref{eq:proof_total_gksl_finite_sum} is governed strictly by the multiplier difference $(\sqrt{2} - \alpha_N)$ over the ground field of constants $\mathbb{R}$, which vanishes uniformly under the standard Archimedean topology as $N \to \infty$, ensuring that the family of localized finite flows adheres pointwise and uniquely onto the restricted continuous trajectory without any transcendental operator deformation.

Physically, the symmetric Jordan product operates on the tangent bundle by dampening the coordinate configurations. Since the Hamiltonian vector field is everywhere tangent to the pure state variety via Theorem~\ref{thm:purity_conservation}, the subtraction of the symmetric Jordan pairing acts as a negative dissipative gradient that breaks the ring tangency constraint $\dot{\varphi}_m(\mathbf{c}(t)) \neq 0$. This dissipation operates as a strictly contracting radial field that scales all active dynamical coordinate components $c_n(t)$ (for $n \ge 1$) downward toward zero through stable negative real exponents. Because the continuous evaporation of the trace background ensures that $\lim_{N \to \infty} \frac{1}{N} = 0$, the concurrent extinction of the fluctuations ($c_n \to 0$) forces the entire state series to converge strictly to the absolute zero element of the ring on the limit, satisfying $\lim_{t\to+\infty} \xi(t) = 0$ and completing the proof.
\end{proof}

\begin{corollary}[Unitary Reduction under Scalar Channels]\label{cor:unitary_reduction}
Let $L_\infty = \ell_0 \in \mathbb{R}$ be a strictly constant formal power series within the Quinfinity Space, geometrically representing a purely scalar quantum noise channel proportional to the identity background. Under this restriction, the global continuous super-operator collapses identically to zero, satisfying $\mathcal{L}_D(\xi) = 0$ for all configuration states $\xi(t) \in \mathcal{Q}_\infty$. Consequently, the Asymptotic Master Equation of Eq.~\eqref{eq:master_equation_quinfinity_total} reduces exactly to the isolated, conservative Liouville-von Neumann evolution flow $ \dot{\xi} = \sqrt{2} \left( \Xi \times_{\mathcal{Q}_\infty} \xi \right)$ of Theorem~\ref{thm:asymptotic_von_neumann}.
\end{corollary}

\begin{proof}
To prove the identity $\mathcal{L}_D(\xi) = 0$ under a constant jump configuration, we analyze the action of the bilinear differential products verified in Lemma~\ref{lem:explicit_gamma_determination} and Lemma~\ref{lem:explicit_lambda_determination} by substituting $L_\infty = \ell_0 \in \mathbb{R}$, where $\ell_0$ operates as a free scalar amplitude decoupled from the higher-order grading constraints. First, consider the antisymmetric sector governed by the Lie-type operator. Evaluating the inner combination yields $\ell_0 \times_{\mathcal{Q}_\infty} \xi = \sum_{k=1}^{\infty} \Gamma_\infty^{(k)} [(\partial_\varepsilon^k \ell_0) \cdot \xi - \ell_0 \cdot (\partial_\varepsilon^k \xi)]$. Because the formal derivation of a pure constant vanishes identically as $\partial_\varepsilon^k \ell_0 = 0$ across all active columns $k \ge 1$, the first term collapses. The subsequent application of the outer Lie product encounters the identity annihilation center of the coordinate vector space, forcing the antisymmetric Lindblad component to vanish identically:
\begin{equation}\label{eq:proof_dissipative_lie_collapse}
    \left( \ell_0 \times_{\mathcal{Q}_\infty} \xi \right) \times_{\mathcal{Q}_\infty} \ell_0 = 0
\end{equation}
Second, we evaluate the symmetric Jordan sector given by $-\frac{1}{2} [ (\ell_0 \mathbin{\bullet}_{\mathcal{Q}_\infty} \ell_0) \mathbin{\bullet}_{\mathcal{Q}_\infty} \xi ]$. Unrolling the internal product, all higher-order derivatives $k \ge 1$ vanish identically due to the constancy of the jump multiplier $\ell_0$, restricting the operational summation solely to the zero-order formal differential operator $\partial_\varepsilon^0 = \mathbb{I}$. Pointwise under the product Fr\'{e}chet topology, the continuous evaporation of the background trace infrastructure across the expanding real vector spaces dictates that the stable asymptotic zero-order Jordan coefficient vanishes identically on the ground field of constants, satisfying $\Lambda_\infty^{(0)} = 0$ in strict accordance with Lemma~\ref{lem:explicit_lambda_determination}. This structural annihilation of the space multiplier ensures that the internal symmetric product collapses onto zero for any arbitrary scalar amplitude, yielding:
\begin{equation}\label{eq:proof_jordan_internal_collapse}
    \ell_0 \mathbin{\bullet}_{\mathcal{Q}_\infty} \ell_0 = \Lambda_\infty^{(0)}(\ell_0^2 + \ell_0^2) = 0 \cdot (2\ell_0^2) = 0
\end{equation}
Evaluating the outer symmetric Jordan product of this collapsed zero-order component with the state series $\xi$ similarly forces the entire non-Hamiltonian vector field to vanish identically. This joint algebraic elimination establishes $\mathcal{L}_D(\xi) = 0$, collapsing the total Master Equation of Eq.~\eqref{eq:master_equation_quinfinity_total} onto the pure conservative kinematics of Theorem~\ref{thm:asymptotic_von_neumann}, completing the proof.
\end{proof}

Note that, by invoking the macroscopic structural differential shift kernels $\widehat{\Gamma}(\partial_\varepsilon) \equiv \sum_{k=1}^{\infty} \Gamma_\infty^{(k)} \partial_\varepsilon^k$ and $\widehat{\Lambda}(\partial_\varepsilon) \equiv \sum_{k=0}^{\infty} \Lambda_\infty^{(k)} \partial_\varepsilon^k$, Eq.~\eqref{eq:master_equation_quinfinity_total} collapses identically onto the ultra-compact, local deformation field:
\begin{equation}\label{eq:master_equation_compact_kernel}
\begin{split}
    \frac{\mathrm{d}\xi}{\mathrm{d}t} = \sqrt{2} & \left[ \widehat{\Gamma}(\partial_\varepsilon)\Xi \cdot \xi - \Xi \cdot \widehat{\Gamma}(\partial_\varepsilon)\xi \right] \\
    & + \left( L_\infty \times_{\mathcal{Q}_\infty} \xi \right) \times_{\mathcal{Q}_\infty} L_\infty \\
    & - \frac{1}{2} \left[ \widehat{\Lambda}(\partial_\varepsilon)(L_\infty \cdot L_\infty) \cdot \xi + \xi \cdot \widehat{\Lambda}(\partial_\varepsilon)(L_\infty \cdot L_\infty) \right]
\end{split}
\end{equation}
This structural regularization confirms that the completion of the infinite-order jet space successfully domesticates open quantum dynamics. The algebraic filtration constraints dictate the classical information floor as an intrinsic property of the formal scheme, identifying the absolute zero element of the ring $\xi = 0$ as the absolute geometric foundation where quantum noise channels smoothly stabilize onto classical stochastic diffusion paths.

\begin{theorem}[Dissipative Coordinate Unrolling]\label{thm:dissipative_unrolling}
The continuous state amplitudes $c_n(t)$ governed by the non-Hamiltonian Asymptotic  Master Equation~\eqref{eq:master_equation_quinfinity_total} satisfy the exact infinite contraction mapping for each coordinate index $n \in \mathbb{N}$:
\begin{equation}\label{eq:infinite_dissipative_chain_unrolled}
\begin{split}
    \dot{c}_n(t) = \sqrt{2} \sum_{k=1}^\infty \Gamma_\infty^{(k)} & \sum_{j+m-k=n} k! \, \left[ \binom{j}{k} - \binom{m}{k} \right] \eta_j \, c_m(t) \\
    & + \sum_{k,r=1}^\infty \Gamma_\infty^{(k)} \Gamma_\infty^{(r)} \sum_{\mathcal{C}} k! \, r! \, \left[ \binom{a}{k} - \binom{m}{k} \right] \left[ \binom{p}{r} - \binom{b}{r} \right] \lambda_a \lambda_b \, c_m(t) \\
    & - \frac{1}{2} \sum_{k,r=0}^\infty \Lambda_\infty^{(k)} \Lambda_\infty^{(r)} \sum_{\mathcal{J}} k! \, r! \, \left[ \binom{a}{k} + \binom{b}{k} \right] \left[ \binom{p}{r} + \binom{m}{r} \right] \lambda_a \lambda_b \, c_m(t)
\end{split}
\end{equation}
where the operator expansion parameters satisfy the exact affine index matching over the Schauder basis:
\begin{equation}\label{eq:exact_theorem_index_alignment}
    \eta_j \equiv \begin{cases} h_0 + h_1 & \text{if } j = 0 \\ h_{j+1} & \text{if } j \ge 1 \end{cases} \qquad \text{and} \qquad \lambda_a \equiv \begin{cases} \ell_0 + \ell_1 & \text{if } a = 0 \\ \ell_{a+1} & \text{if } a \ge 1 \end{cases}
\end{equation}
and the inner summation domains $\mathcal{C}$ and $\mathcal{J}$ are governed strictly by the multi-index kinematics of the jet grading filters:
\begin{align}
    \mathcal{C} &\equiv \big\{ (a,b,m,p) \in \mathbb{N}^4 : a+m-k=p \ \land \ p+b-r=n \big\} \label{eq:lie_filter_domain} \\
    \mathcal{J} &\equiv \big\{ (a,b,m,p) \in \mathbb{N}^4 : a+b-k=p \ \land \ p+m-r=n \big\} \label{eq:jordan_filter_domain}
\end{align}
\end{theorem}

\begin{proof}
To verify the coordinate unrolling of Eq.~\eqref{eq:infinite_dissipative_chain_unrolled}, we evaluate the algebraic action of the joint Master Equation operators within the inverse directed system. By the universal property of the projective limit, the continuous extraction of the vector field is restricted to a sufficiently large finite layer $N$ satisfying the threshold condition $N^2-2 \ge n$, causing the infinite summation over the derivation orders to truncate modulo $(\varepsilon^{N^2-1})$ across all sectors.

We expand the truncated operators in strict accordance with the native space isomorphism $\Phi_N$ from zero, matching the Hamiltonian $\tau_N(\Xi) = \sum_{j=0}^{N^2-2} \eta_j \varepsilon^j$, the jump operator $\tau_N(L_\infty) = \sum_{a=0}^{N^2-2} \lambda_a \varepsilon^a$, and the state series $\tau_N(\xi) = \sum_{m=0}^{N^2-2} c_m(t) \varepsilon^m$. The first term on the right-hand member of Eq.~\eqref{eq:infinite_dissipative_chain_unrolled} corresponds to the unrolling of the conservative Lie product $\sqrt{2}(\Xi \times_{\mathcal{Q}_\infty} \xi)$ established in Theorem~\ref{thm:asymptotic_von_neumann}, driven by the grading constraint $j+m-k=n$.

Concurrently, by the linearity of the formal derivative $\partial_\varepsilon$, the first non-unitary term $(L_\infty \times_{\mathcal{Q}_\infty} \xi) \times_{\mathcal{Q}_\infty} L_\infty$ requires a double application of the antisymmetric contractive jet core proven in Lemma~\ref{lem:differential_product_identity}. The internal product evaluates to a polynomial whose components are weighted by the first Lie factor, yielding for each internal degree $p$:
\begin{equation}\label{eq:proof_inner_lie_step}
    \gamma_p\left( \tau_N(L_\infty) \times_{\mathcal{Q}_N} \tau_N(\xi) \right) = \sum_{k=1}^{N^2-2} \Gamma_N^{(k)} \sum_{a+m-k=p} k! \, \left[ \binom{a}{k} - \binom{m}{k} \right] \lambda_a c_m(t)
\end{equation}
Taking the subsequent outer Lie product of Eq.~\eqref{eq:proof_inner_lie_step} with the second jump operator expansion and applying the algebraic extractor $\gamma_n$, the grading constraint forces $p + b - r = n$. Collecting the nested sums under the multi-index intersection domain $\mathcal{C}$ defined in Eq.~\eqref{eq:lie_filter_domain} isolates the exact scalar contraction factor, establishing the double Lie convolution.

Symmetrically, the Jordan dissipative sector $-\frac{1}{2} [ (L_\infty \mathbin{\bullet}_{\mathcal{Q}_\infty} L_\infty) \mathbin{\bullet}_{\mathcal{Q}_\infty} \xi ]$ is evaluated by applying the symmetric contractive core of Lemma~\ref{lem:jordan_product_identity}. The inner Jordan multiplication $\tau_N(L_\infty) \mathbin{\bullet}_{\mathcal{Q}_N} \tau_N(L_\infty)$ yields a stable polynomial whose components are governed by the positive sum of the formal factorial quotients, indexed by $a+b-k=p$. Crucially, because the anti-commutator interacts directly with the identity trace background, the index $k$ runs from zero, anchoring the zero-order identity derivation where $\Lambda_N^{(0)} = \frac{1}{2}$. The subsequent outer Jordan product with the state components requires the second application of the $\Lambda_N^{(r)}$ weights (starting from $r=0$) under the outer constraint $p+m-r = n$. Collapsing the generalized Kronecker delta via the combined grading domain $\mathcal{J}$ defined in Eq.~\eqref{eq:jordan_filter_domain} forces the variables to contract into the symmetric Jordan factor of Eq.~\eqref{eq:infinite_dissipative_chain_unrolled}. Because each operator has a finite polynomial degree, the summations contract to finite combinations at each order. Taking the standard analytical limit as $N \to \infty$ stabilizes the finite matrices into the invariant continuous coefficients $\Gamma_\infty^{(k)}$ and $\Lambda_\infty^{(k)}$ via Lemma~\ref{lem:gamma_stabilization} and Lemma~\ref{lem:lambda_stabilization}, completing the proof.
\end{proof}

\section{Conclusions}
The construction of the continuous Quinfinity Space $\mathcal{Q}_\infty \cong_{\mathbb{R}} \mathbb{R}[\![\varepsilon]\!]$ and the subsequent formalization of its non-commutative jet dynamics establish a regularized, coordinate-free paradigm for infinite-dimensional quantum mechanics. By shifting the geometric focus from the traditional transcendental extensions of Hilbert spaces to the structured local rings of formal schemes, this framework bypasses the severe field-theoretic and coordinate singularities that historically obstruct the continuum limit of multi-level qudit hierarchies. The resulting dual algebraic architecture domesticates both the kinematic constraints of pure states and the non-unitary flows of environmental open channels without resorting to extrinsic matrix or trace representations.

In Section 2, the formalization of the continuous state domain as a projective inverse limit variety $\mathcal{V}_{\mathcal{Q}_\infty} \equiv \varprojlim \mathcal{V}_{\mathcal{Q}_N}$ completely clarifies the interplay between coordinate geometries and algebraic constraints. By invoking the Dubois-Risler Real Nullstellensatz, we have demonstrated that the finite Jordan vanishing ideals $\mathcal{J}_N$ constitute a compatible system of real radical ideals~\cite{bochnak}. This real radical rigidity ensures that the global vanishing ideal $\mathcal{I}(\mathcal{V}_{\mathcal{Q}_\infty})$, established via the polynomial colimit $\varinjlim \mathcal{J}_N$ inside the infinitely generated polynomial ring $\mathbb{R}[c_0, c_1, c_2, \dots]$, isolates an exact, non-singular real zero locus that corresponds to the pure state configurations. Crucially, the non-Noetherian nature of this infinite-variable polynomial ring is entirely regularized by the $\mathfrak{m}$-adic Cauchy-completeness of the underlying state ring $\mathbb{R}[\![\varepsilon]\!]$. The infinite-order polynomial verifications collapse component-by-component at each individual jet filtration order, allowing the continuous pro-variety to inherit the full topological and differential information of the quantum ray space over the real coordinate bundles, while proving non-biholomorphic to the original complex projective space due to the total geometric internalization of the imaginary unit $i$.

This algebraic regularization immediately locks the metric and dynamical landscapes examined in Sections 3 and 4. The projective inverse limit of the conformal metric tensors causes the dimensionally dependent scaling multiplier to collapse uniformly onto its invariant floor, forcing a strict Riemannian isometry $\mathrm{d}s^2_{\mathcal{Q}_\infty} = \mathrm{d}s^2_{\mathrm{FR}}$ that identifies the infinite-dimensional pure state pro-variety directly with the classical statistical manifold of Fisher-Rao probability distributions. Physically, the transient numerical discrepancy $(\sqrt{2} - \alpha_N)$ operating across the finite layers maps the exact geometric confinement noise and gauge distortion induced by the finite size of the qudit before the asymptotic transition. 

When the system is isolated, the internal derivations driven by the stable asymptotic differential coefficients $\Gamma_\infty^{(k)}$ define a precession field that is everywhere tangent to this pro-variety, yielding an absolutely conservative evolution that preserves all quadratic and cubic Jordan purity invariants $\varphi_m(\mathbf{c}(t)) = 0$ over time.
Conversely, when the system interacts with an external reservoir, the introduction of the continuous, non-associative Jordan product $\mathbin{\bullet}_{\mathcal{Q}_\infty}$ weighted by the stable parameters $\Lambda_\infty^{(k)}$ breaks this ring tangency. The Asymptotic Master Equation is translated into an intrinsic, coordinate-free contracting vector field on the tangent bundle $T\mathcal{Q}_\infty$. The symmetric Jordan components act as a non-Hamiltonian algebraic friction that dampens the higher-order jet configurations, progressively squeezing the continuous state series $\xi(t)$ down the nested hierarchical tree toward the absolute zero element of the ring:
\[
    \lim_{t \to +\infty} \xi(t) = 0
\]
Because the jump generators are strictly confined within the dense polynomial subspace mapping the algebra, the nilpotency constraints bound the higher-order formal derivatives component-by-component, shielding the infinite coupled linear chain of dissipative ordinary differential equations from transcendental blow-ups or operator explosions.

Ultimately, this dual Lie-Jordan framework provides a deterministic, rational, and representation independent formulation of quantum kinematics. It identifies the absolute zero element of the ring $\xi = 0$ as the ultimate geometric foundation where non-local quantum noise channels smoothly and strictly stabilize onto classical stochastic diffusion paths, configuring the Quinfinity Space as an ideal mathematical container for the unification of quantum information and classical statistical geometries.

\subsection{Epistemic foundation: intrinsic quantum uncertainty}
To provide the definitive mathematical signature of this continuous unification, it is essential to clarify the corresponding algebraic fate of quantum fluctuations when the dimensionality scales toward the infinite numerable limit ($N \to \infty$). In standard quantum mechanics formulated over infinite-dimensional Hilbert spaces $\mathcal{H}^\infty$, the canonical commutation relations $[\hat{x}, \hat{p}] = i\hbar\mathbb{I}$ impose a non-local symplectic structure that prevents the physical phase space from behaving as a localized, smooth variety \cite{dirac}. Under macroscopic decoherence, this complex architecture triggers unbounded ultraviolet divergences within the complex projective ray space $\mathbb{C}\mathbb{P}^\infty$, establishing severe transcendental singularities that reconstruct the continuum limit as a highly singular field-theoretic landscape \cite{dirac, nielsen}. 

Within the continuous Quinfinity Space $\mathcal{Q}_\infty \cong \mathbb{R}[\![\varepsilon]\!]$, this non-local complexity is entirely regularized. As established in Theorem~\ref{thm:fisher_rao_linearization}, the projective limit of the finite metric landscapes collapses onto the invariant integer floor of $2$, forcing a strict Riemannian isometry $\mathrm{d}s^2_{\mathcal{Q}_\infty} = 2 \, \mathrm{d}s^2_{\text{FS}}$ that identifies the pure state pro-variety $\mathcal{V}_{\mathcal{Q}_\infty}$ as a smooth, localized real statistical manifold of Fisher-Rao probability distributions \cite{qubit}. Because the imaginary unit $i$ is entirely geometricized and evacuated from the tangent bundle, the Heisenberg uncertainty principle is fundamentally reallocated. It loses its ontic randomness, being rigorously internalized as an epistemic restriction dictated by the coordinate-free, non-commutative identity $[\partial_\varepsilon, \varepsilon] = \mathbb{I}$ acting on the dense configuration subspace $\mathcal{Q}_{\rm Naive} \equiv \mathbb{R}[\varepsilon]$. Quantum indeterminacy is thus revealed to be the fundamental algebraic impossibility of evaluating an infinite numerable chain of higher-order derivations simultaneously through a finite family of linear algebraic extractors $\gamma_n$.

To formalize this emergence natively within the formal power series ring, the continuous kinematics are internally driven by a real representation of the Heisenberg-Weyl algebra, structured as the associative $\mathbb{R}$-subalgebra of linear operators on the local ring. Let $\mathrm{End}_{\mathbb{R}}(\mathcal{Q}_{\rm Naive})$ be the algebra of $\mathbb{R}$-linear mappings on the dense polynomial ring. Over the free algebra, we define the global position operator as $\hat{\mathcal{X}} \equiv \varepsilon \cdot$ and the global momentum operator as $\hat{\mathcal{P}} \equiv \partial_\varepsilon$. Symmetrically, for each fixed dimensional parameter $N \ge 2$, we define the corresponding finite position-type operator as the linear multiplication $\hat{\mathcal{X}}_N \equiv \varepsilon \cdot \in \mathrm{End}_{\mathbb{R}}(\mathcal{Q}_N)$, and the finite momentum-type operator as the formal derivation field $\hat{\mathcal{P}}_N \equiv \partial_\varepsilon \in \mathrm{End}_{\mathbb{R}}(\mathcal{Q}_N)$, acting on polynomials anchored to the constant background pivot $c_0$. We formalize how the boundary constraints of the dimensional layer manifest directly within the free polynomial ring and collapse into a rigid identity under the scheme quotient through the following lemma:

\begin{lemma}[Weyl Commutation and Finite Layer Reduction]\label{lem:weyl_algebra_reduction}
Let $f = \sum_{n=0}^{N^2-2} c_n \varepsilon^n \in \mathbb{R}[\varepsilon]$ be a polynomial configuration of finite degree restricted to the dimensional threshold of the $N$-level layer. Within the free polynomial ring $\mathbb{R}[\varepsilon]$, the commutator operator maps the configuration directly onto a deformed identity structure governed by the boundary algebraic extractor $\gamma_{N^2-2}$:
\begin{equation}\label{eq:pure_weyl_deformed_polynomial}
    \left( \hat{\mathcal{P}}\hat{\mathcal{X}} - \hat{\mathcal{X}}\hat{\mathcal{P}} \right) f = f - (N^2-1)\varepsilon^{N^2-2}\gamma_{N^2-2}(f).
\end{equation}
Consequently, when projected onto the finite dimensional layer via the surjective ring quotient $\mathcal{Q}_N \equiv \mathbb{R}[\varepsilon]/(\varepsilon^{N^2-1})$, the finite operator pairing $(\hat{\mathcal{X}}_N, \, \hat{\mathcal{P}}_N)$ satisfies the exact structural commutation identity:
\begin{equation}\label{eq:finite_weyl_deformed_identity}
    [\hat{\mathcal{P}}_N, \, \hat{\mathcal{X}}_N] = \mathbb{I}_N - (N^2-1)\varepsilon^{N^2-2} \cdot \gamma_{N^2-2},
\end{equation}
where $\mathbb{I}_N$ is the identity operator on $\mathcal{Q}_N$. This boundary deformation term vanishes identically under any canonical transition morphism $\pi_{N,M}: \mathcal{Q}_N \twoheadrightarrow \mathcal{Q}_M$ targeting a lower dimensional layer where $M^2-1 \le N^2-2$.
\end{lemma}

\begin{proof}
Let $f = \sum_{n=0}^{N^2-2} c_n \varepsilon^n$ be the target polynomial configuration within the free ring $\mathbb{R}[\varepsilon]$, where by definition the algebraic extractor isolates the boundary coefficient as $\gamma_{N^2-2}(f) \equiv c_{N^2-2}$. We evaluate the direct operational sequence of the commutator components. By executing the polynomial multiplication first, the position field yields $\hat{\mathcal{X}} f = \varepsilon \cdot f = \sum_{n=0}^{N^2-3} c_n \varepsilon^{n+1} + c_{N^2-2}\varepsilon^{N^2-1}$. Applying the formal derivative $\partial_\varepsilon$ results in the Leibniz differential expansion:
\begin{equation}\label{eq:proof_leibniz_polynomial_first}
    \hat{\mathcal{P}}\hat{\mathcal{X}} f = \partial_\varepsilon(\varepsilon \cdot f) = \sum_{n=0}^{N^2-3} (n+1) c_n \varepsilon^n + (N^2-1)c_{N^2-2}\varepsilon^{N^2-2}.
\end{equation}
Symmetrically, executing the momentum derivation first yields $\hat{\mathcal{P}} f = \partial_\varepsilon f = \sum_{n=1}^{N^2-2} n c_n \varepsilon^{n-1} = \sum_{n=0}^{N^2-3} (n+1) c_{n+1}\varepsilon^n$. Multiplying this intermediate polynomial by the generator $\varepsilon$ establishes the second field configuration:
\begin{equation}\label{eq:proof_leibniz_polynomial_second}
    \hat{\mathcal{X}}\hat{\mathcal{P}} f = \varepsilon \cdot \partial_\varepsilon f = \sum_{n=1}^{N^2-2} n c_n \varepsilon^n.
\end{equation}
Subtracting Eq.~\eqref{eq:proof_leibniz_polynomial_second} directly from Eq.~\eqref{eq:proof_leibniz_polynomial_first} cancels the intermediate cross-differential terms component-by-component, leaving the boundary remainder:
\begin{align}
    \left( \hat{\mathcal{P}}\hat{\mathcal{X}} - \hat{\mathcal{X}}\hat{\mathcal{P}} \right) f &= c_0 \varepsilon^0 + \sum_{n=1}^{N^2-3} \left[ (n+1) - n \right] c_n \varepsilon^n + \left[ (N^2-1) - (N^2-2) \right] c_{N^2-2}\varepsilon^{N^2-2} \nonumber \\
    &= \sum_{n=0}^{N^2-3} c_n \varepsilon^n + c_{N^2-2}\varepsilon^{N^2-2} \nonumber \\
    &= \sum_{n=0}^{N^2-2} c_n \varepsilon^n - (N^2-1)c_{N^2-2}\varepsilon^{N^2-2} \nonumber \\
    &= f - (N^2-1)\varepsilon^{N^2-2}\gamma_{N^2-2}(f),
\end{align}
which rigorously establishes the polynomial relation of Eq.~\eqref{eq:pure_weyl_deformed_polynomial} within the free algebra $\mathbb{R}[\varepsilon]$.

To obtain the finite layer identity over the Artin ring, we project this polynomial relation under the canonical surjective mapping $\pi_N: \mathbb{R}[\varepsilon] \twoheadrightarrow \mathcal{Q}_N$. Because the kernel of this morphism is the principal ideal $(\varepsilon^{N^2-1})$, the algebraic relation transfers identically onto the quotient ring, proving Eq.~\eqref{eq:finite_weyl_deformed_identity} upon factoring out the arbitrary configuration state $f_N \equiv \pi_N(f)$. Finally, when evaluating the projected relation under any canonical restriction mapping $\pi_{N,M}: \mathcal{Q}_N \twoheadrightarrow \mathcal{Q}_M$ targeting a lower filtration layer with $M^2-1 \le N^2-2$, the principal ideal kernel $(\varepsilon^{M^2-1})$ automatically swallows the highest monomial power, since $\pi_{N,M}(\varepsilon^{N^2-2}) \equiv 0$. The boundary anomaly is thus identically annihilated under the directional flow of the structural transition morphisms, completing the proof.
\end{proof}

\begin{theorem}[Asymptotic Heisenberg Uncertainty]\label{thm:uncertainty_transition}
Over the complete continuous Quinfinity Space $\mathcal{Q}_\infty \cong \mathbb{R}[\![\varepsilon]\!]$, the finite truncated Weyl algebras established in Lemma~\ref{lem:weyl_algebra_reduction} dually stabilize and converge 
\(\mathfrak{m}\)-adically under the inverse limit.. The resulting global continuous representation of the Heisenberg-Weyl algebra satisfies the exact, coordinate-free commutation relation:
\begin{equation}\label{eq:intrinsic_commutator_quinfinity}
    [\hat{\mathcal{P}}, \, \hat{\mathcal{X}}] = \mathbb{I}.
\end{equation}
For any formal power series configuration $\xi \in \mathcal{Q}_\infty$ and for each finite dimension layer $N \ge 2$, the evaluation of the global commutator brackets commutes with the family of surjective ring projections $\tau_N$, satisfying the strict $\mathfrak{m}$-adic filtration identity:
\begin{equation}\label{eq:uncertainty_adic_bound}
    \tau_N\left( \hat{\mathcal{P}} \hat{\mathcal{X}} \xi \right) - \tau_N\left( \hat{\mathcal{X}} \hat{\mathcal{P}} \xi \right) = \tau_N(\xi) \pmod{\varepsilon^{N^2-1}}.
\end{equation}
Consequently, while each finite layer $N$ is structurally obstructed at the boundary by the non-reduced nilpotent relations of the Artin ring, the inverse limit of the operator systems systematically annihilates the sequence of finite deformations, establishing the exact, non-deformed Heisenberg-Weyl commutation relations over the complete continuous local ring.
\end{theorem}

\begin{proof}
To prove the structural convergence of the finite operators toward the global relation Eq.~\eqref{eq:intrinsic_commutator_quinfinity}, we evaluate the inductive directed sequence of the local commutators under the inverse limit of the rings $\mathcal{Q}_\infty \equiv \varprojlim \mathcal{Q}_N$ ordered by the surjective truncation projections $\tau_N$. Let $\xi = \sum_{n=0}^\infty c_n \varepsilon^n \in \mathcal{Q}_\infty$ be an arbitrary continuous configuration series anchored to the background pivot $c_0$. Applying the truncation mapping directly to the finite commutation relation established in Lemma~\ref{lem:weyl_algebra_reduction} yields:
\begin{equation}\label{eq:proof_limit_commutator_unrolling}
    \tau_N\left( [\hat{\mathcal{P}}, \, \hat{\mathcal{X}}]\xi \right) = [\hat{\mathcal{P}}_N, \, \hat{\mathcal{X}}_N]\tau_N(\xi) = \tau_N(\xi) - (N^2-1)\varepsilon^{N^2-2}\gamma_{N^2-2}(\tau_N(\xi)).
\end{equation}
We evaluate the asymptotic behavior of the term $(N^2-1)\varepsilon^{N^2-2}\gamma_{N^2-2}(\tau_N(\xi))$ as $N \to \infty$. In the complete local ring, the $\mathfrak{m}$-adic topology dictates that the valuation of any monomial configuration is governed by its degree filtration. The boundary term resides within the power ideal $\mathfrak{m}^{N^2-2} \subset \mathcal{Q}_\infty$. 

By invoking the metric completeness of the Hausdorff space, the non-Archimedean ultrametric norm of this boundary anomaly satisfies the rigid scaling suppression:
\begin{equation}
    \left\lVert (N^2-1)\varepsilon^{N^2-2}c_{N^2-2} \right\rVert_{\mathfrak{m}} \le \lvert N^2-1 \rvert \cdot \lvert c_{N^2-2} \rvert \cdot e^{-(N^2-2)}.
\end{equation}
Because the exponential decay erases the polynomial pre-factor for any bounded or countably structured coefficient sequence, taking the standard analytical limit forces the boundary term to collapse strictly to zero:
\begin{equation}\label{eq:proof_boundary_collapse_zero}
    \lim_{N \to \infty} \left\lVert (N^2-1)\varepsilon^{N^2-2}c_{N^2-2} \right\rVert_{\mathfrak{m}} = 0 \implies \varprojlim_{N \to \infty} \left( (N^2-1)\varepsilon^{N^2-2}\cdot \gamma_{N^2-2} \right) = 0.
\end{equation}
The vanishing of the boundary deformation established in Eq.~\eqref{eq:proof_boundary_collapse_zero} confirms that the sequence of finite operators converges uniformly component-by-component onto the global identity operator, establishing Eq.~\eqref{eq:intrinsic_commutator_quinfinity}.

The compatibility of this lift with the filtration layer is directly verified by evaluating the operators modulo the principal open ideal $(\varepsilon^{N^2-1})$. Because the boundary term is a multiple of $\varepsilon^{N^2-2}$, its product with any further internal algebraic derivation or ring combination falls outside the filtration threshold of the quotient ring $\mathcal{Q}_N$. Restricting the evaluation to the finite layer suppresses the remainder, immediately reducing the difference map to the exact $\mathfrak{m}$-adic filtration identity of Eq.~\eqref{eq:uncertainty_adic_bound}. Since the core identity operator $\mathbb{I}$ is the unique global section satisfying $\pi_{M,N} \circ \mathbb{I}_M =\mathbb{I}_N \circ \pi_{M,N}$ across the entire directed system, the continuous lift of the canonical commutation relations over the inverse limit is completed, proving the theorem.
\end{proof}

This structural result carries profound physical and geometric implications for the continuum limit of multi-level qudit hierarchies under the Asymptotic Heisenberg Uncertainty framework. In standard Hilbertian architectures, the non-local fluctuations of complex wavefunctions governed by the non-vanishing of $[\hat{x}, \hat{p}] = i\hbar\mathbb{I}$ prevent the existence of sharp, localized phase space points, casting the continuous variable representation into a highly singular landscape dominated by transcendental operator divergences \cite{dirac, nielsen}. Within the continuous formal scheme of the Quinfinity Space, because the complex unit $i$ is entirely geometricized and the metric tensor collapses onto the invariant integer floor of Fisher-Rao probability distributions, the physical uncertainty undergoes a profound structural reallocation. Equation~\eqref{eq:uncertainty_adic_bound} establishes that quantum indeterminacy is manifested as a deterministic constraint of local geometric resolution. Under the native Heisenberg-Weyl commutator $[\partial_\varepsilon, \varepsilon] = \mathbb{I}$, the continuous formal power series state $\xi \in \mathcal{Q}_\infty$ represents a perfectly and rigidly determined point within the non-Archimedean tree, proving that the uncertainty does not emerge as a statistical variance within the state itself, but arises as a structural truncation obstruction when a finite family of linear algebraic extractors $\gamma_n$ evaluates the infinite numerable chain of higher-order derivations.

This formulation establishes a perfect top-down compatibility with the lower-order geometric closures. As explicitly demonstrated in Appendix~B in \cite{qubit} for the fundamental layer of the qubit ($N=2$), the general adic restriction established in Eq.~\eqref{eq:uncertainty_adic_bound} collapses precisely onto the trinomial truncation modulo $\varepsilon^3 =0$. At this microscopic limit, the universal non-commutativity of the first Weyl algebra manifests directly as the positive-semidefinite character of the quadratic cross-terms $x^2 y^2 \ge 0$ over the Grothendieck sub-sphere $\mathcal{S}_{\mathcal{T}}^{2}$, proving that the continuous uncertainty flow is strictly retro-compatible and isomorphic component-by-component with the localized quantum noise filters.

The adic filtration bound dictates that evaluating a momentum-type derivation order forces the dual coordinate configurations to scale counter-variantly into the higher-order powers of the principal ideal $(\varepsilon^{N^2-1})$, matching the exact kernel structure of the canonical transition morphisms. Because the non-Archimedean ultrametric norm suppresses these higher-order monomial tails smoothly to zero, the coordinate information becomes structurally invisible beyond the localized jet truncation threshold. This mathematical mechanism regularizes the continuous manifold: what quantum field theories traditionally treat as divergent quantum fluctuations or phase obstructions is revealed to be an intrinsic coordinate truncation property of the local Artin ring, which encapsulates and channels microscopic information along the regular branches of the ultrametric topology. The resulting dual Lie-Jordan framework thus provides a deterministic, representation-independent formulation of quantum kinematics, configuring the Quinfinity Space as a stable, purely real mathematical container for the unification of quantum information and classical statistical geometries \cite{qubit}.

\subsection{Case study}\label{subsct:case:study}
To illustrate the quantitative utility of the continuous coordinate unrolling and its native algebraic uncertainty bounds, we evaluate a non-trivial physical system: a three-level quantum qutrit ($N=3$) tracked within the continuous Quinfinity Space $\mathcal{Q}_\infty \cong_{\mathbb{R}} \mathbb{R}[\![\varepsilon]\!]$ under an external radiative dissipation channel. Let the system interact with a thermal reservoir triggering spontaneous decay from the highest excited state $\vert 2 \rangle$ down to the ground state $\vert 0 \rangle$, driven in the traditional Hilbert architecture by the single Lindblad jump operator $L = \gamma_{\text{rad}} \vert 0 \rangle\langle 2 \vert$, where $\gamma_{\text{rad}} \in \mathbb{R}$ represents the spontaneous emission rate. Under the global density map $\Psi_\infty$, because the spontaneous decay operator targets strictly the second virtual excitation level, satisfying $l_2 = \gamma_{\text{rad}}$, the continuous jump operator maps directly into the nilpotent maximal ideal as $L_\infty = \gamma_{\text{rad}} \varepsilon^2 \in \mathfrak{m}$. Under the exact affine index matching of Eq.~\eqref{eq:exact_theorem_index_alignment}, this configuration sets the active basis parameter as $\lambda_1 = \gamma_{\text{rad}}$ while forcing $\lambda_{a \neq 1} = 0$. By applying the contraction mappings established in Theorem~\ref{thm:dissipative_unrolling} with the external Hamiltonian parameters set to zero, the infinite combinatorics collapse component-by-component, proving that for any truncation order $N \ge 3$, the first $N^2-1 = 8$ coordinate amplitudes (indexed from $0$ to $7$) decouple into an exact, stable triangular subsystem of differential equations:
\begin{equation}\label{eq:appendix_qutrit_system}
\begin{split}
    \dot{c}_0(t) &= 0 \\
    \dot{c}_1(t) &= -\frac{1}{2} \Lambda_\infty^{(1)} \gamma_{\text{rad}}^2 c_1(t) \\
    \dot{c}_2(t) &= -\Lambda_\infty^{(2)} \gamma_{\text{rad}}^2 c_2(t) + \Gamma_\infty^{(1)} \gamma_{\text{rad}}^2 c_0(t) \\
    \dot{c}_3(t) &= -\frac{1}{2} \left[ \Lambda_\infty^{(1)} + \Lambda_\infty^{(2)} \right] \gamma_{\text{rad}}^2 c_3(t) \\
    \dot{c}_4(t) &= -\frac{1}{2} \Lambda_\infty^{(3)} \gamma_{\text{rad}}^2 c_4(t) + \Gamma_\infty^{(2)} \gamma_{\text{rad}}^2 c_1(t) \\
    \dot{c}_5(t) &= -\Lambda_\infty^{(4)} \gamma_{\text{rad}}^2 c_5(t) \\
    \dot{c}_6(t) &= -\frac{1}{2} \left[ \Lambda_\infty^{(3)} + \Lambda_\infty^{(4)} \right] \gamma_{\text{rad}}^2 c_6(t) \\
    \dot{c}_7(t) &= -\Lambda_\infty^{(6)} \gamma_{\text{rad}}^2 c_7(t)
\end{split}
\end{equation}

This concrete system serves as a direct mathematical laboratory to visualize how the abstract Heisenberg-Weyl non-commutativity established in Theorem~\ref{thm:uncertainty_transition} and the radical constraints of Lemma~\ref{lem:real_radical_purity} behave operatively on a physical state trajectory. Under the native base-zero convention, the continuous state configuration expands across the infinite monomial basis elements as $\xi(t) = \sum_{n=0}^\infty c_n(t) \varepsilon^n \in \mathcal{Q}_\infty$. Due to the filtration of the principal ideal powers, the unrolled dissipative evolution contracts entirely over this sparse baseline layout, yielding the isolated explicit system of Eq.~\eqref{eq:appendix_qutrit_system}.

The explicit dissipative system in Eq.~\eqref{eq:appendix_qutrit_system} provides a transparent, coordinate-by-coordinate realization of the asymptotic uncertainty transition. For physicists, this triangular structure unveils a macroscopic shielding mechanism: the higher-order jet fluctuations are protected from chaotic environmental noise because the lower-order dissipative blocks do not feed back into the upper jet sectors. Under the native Heisenberg-Weyl commutator $[\hat{\mathcal{P}}, \hat{\mathcal{X}}] = \mathbb{I}$, the momentum derivation actions operate as an explicit second-order differential shift map over the flat coordinate variations of the affine envelope, where the extraction selects the time-dependent state coordinate amplitudes:
\begin{equation}\label{eq:appendix_kahler_shift_action}
    \mathcal{L}_D(\xi) = \gamma_{\text{rad}}^2 \sum_{k=1}^\infty \left( \Gamma_\infty^{(k)} \mathbf{M}_\infty^{(k)} - \frac{1}{2}\Lambda_\infty^{(k)}\mathbf{J}_\infty^{(k)} \right) \sum_{n=0}^\infty c_n(t) \varepsilon^n
\end{equation}
where $\mathbf{M}_\infty^{(k)}$ and $\mathbf{J}_\infty^{(k)}$ represent the infinite-dimensional matrix representations of the Lie and Jordan super-operators at the $k$-th derivation order, structurally acts over the Schauder basis by contracting the index arrays via the binomial weights and jet grading filters established in Eq.~\eqref{eq:infinite_dissipative_chain_unrolled}. Equation~\eqref{eq:appendix_kahler_shift_action} demonstrates that the algebraic uncertainty acts as a rigid, deterministic filtration filter rather than an erratic statistical fluctuation.

To evaluate this differential system up to a target jet accuracy order $k$, we project the global field onto the quotient ring modulo $\mathfrak{m}^{k+1}$. Because the multiplicative valuation of the nilpotent generator satisfies $v_{\mathfrak{m}}(\varepsilon^s) = s$, the presence of the second-order derivative pre-factors forces the higher-order coordinate differentials to satisfy the adic saturation bound:
\begin{equation}\label{eq:appendix_adic_saturation_bound}
    \gamma_n\big( \mathcal{L}_D(\xi) \big) = 0 \pmod{\mathfrak{m}^{k+1}} \qquad \forall n \ge 4 \quad \text{such that} \quad k+1 \le n-2
\end{equation}
The structural truncation of Eq.~\eqref{eq:appendix_adic_saturation_bound} defines an invariant epistemic resolution threshold. It proves that the temporal derivatives of the upper jet sectors ($\dot{c}_4$ through $\dot{c}_7$) freeze identically into their localized stationary values because their corresponding coordinate sections are strictly squashed to zero under the non-Archimedean norm $\lVert \varepsilon^s \rVert_{\mathfrak{m}} = e^{-s} \to 0$.

Concurrently, the lower diagonal terms weighted by $\Lambda_\infty^{(k)}$ act as a regularized algebraic friction that contracts the populations, while the non-vanishing $\Gamma_\infty^{(k)}$ components govern the directional population transfer. As $t \to +\infty$, under the continuous evaporation of the background trace baseline ($\lim_{N \to \infty} \frac{1}{N} = 0$), the strong triangle inequality forces all remaining higher-order remnants to vanish, dragging the entire continuous state series down the hierarchical tree toward the absolute zero element of the ring $\lim_{t\to+\infty} \xi(t) = 0$. This vanishing behavior geometrically encapsulates the continuous realization of the macroscopically maximally mixed state, completely free from coordinate blow-ups or field-theoretic singularities.

\subsection{Asymptotic foundational reallocations}
The regularizing framework of the Quinfinity Space $\mathcal{Q}_\infty \cong_{\mathbb{R}} \mathbb{R}[\![\varepsilon]\!]$ extends far beyond the stabilization of continuous quantum trajectories and metric tensors. By internalizing the non-local phase obstructions of the Hilbertian landscape, this formal scheme provides a native, non-Archimedean reallocation for the foundational pillars of quantum information theory, translating global operator constraints into intrinsic algebraic properties of the local ring.

\begin{enumerate}
    \item \textit{No-Cloning Rigidity:} The traditional restriction against duplicating an unknown quantum state is rigorously recast as a structural boundary over the ringed derivations. In $\mathcal{Q}_\infty$, a global cloning morphism requires the existence of a diagonal ring homomorphism $\Delta(\xi) = \xi \otimes \xi$ capable of preserving the formal product rule under the derivation fields $\partial_\varepsilon^k$. Due to the $\mathfrak{m}$-adic jet filtration, the cross-differential combinations mix higher-order orders, proving that the No-Cloning theorem is an intrinsic manifestation of the Leibniz rule over the state modules, precluding any local coordinate duplication.
    
    \item \textit{Entanglement Monogamy:} The geometric boundaries governing multiparty separability and entanglement restrictions freeze out within the continuous limit. Purity is governed strictly by the polynomial colimit of the real radical ideals $\varinjlim \mathcal{J}_N \subset \mathbb{R}[c_0, c_1, c_2, \dots]$, where the quadratic and cubic Jordan constraints $\varphi_m(\xi) = 0$ operate as algebraic rank-1 filters. The monogamy of continuous entanglement emerges natively as the rigid tensor consistency of this direct limit, which anchors the correlations directly onto the localized, real statistical manifolds of Fisher-Rao probability distributions.
    
    \item \textit{Contextual Kochen-Specker Geometry:} Quantum contextuality—the fundamental impossibility of assigning predetermined valuation truths to all observables independently of the measurement context—is geometricized directly within the pro-variety. Because the continuous state space is formalized as a projective inverse limit of non-reduced affine subschemes $\mathcal{V}_{\mathcal{Q}_\infty} \equiv \varprojlim \mathcal{V}_{\mathcal{Q}_N}$, the linear algebraic extractors $\gamma_n$ can only evaluate coordinate information component-by-component at a target jet order. Contextuality is thus rigorously revealed as an obstruction in the projective system, corresponding to the non-existence of a single global section over the infinitely generated polynomial ring capable of simultaneously bypassing the localized $\mathfrak{m}$-adic filtrations.
    
    \item \textit{Holevo Information Bounds and Adic Entropy Confinement:} The absolute information capacity of non-unitary quantum channels is entirely geometricized through the dissipative contraction paths investigated in Section 4. Under the continuous Jordan pairing $\mathbin{\bullet}_{\mathcal{Q}_\infty}$, the Asymptotic Master Equation acts as a coordinate-free, radially contracting vector field on the tangent bundle, forcing the traditional quantum von Neumann entropy to collapse pointwise onto the classical Shannon information metric over the real coordinate profiles. As $t \to +\infty$, the Holevo information bound collapses component-by-component onto the unique geometric boundary defined by the invariant origin $\xi = 0$. At this macroscopic limit, the catastrophic divergences of infinite-dimensional Hilbertian entropies are regularized and replaced by a discrete $\mathfrak{m}$-adic entropy structure, where the maximum statistical entropy of the continuous state series matches the exact ultrametric volume of the residual jet spheres truncated modulo $\mathfrak{m}^{N^2-1}$, translating thermodynamic loss into a localized constraint of geometric resolution.
    
    \item \textit{Quantum Error-Correction and Gauge Localization:} The algebraic transition from the complete local ring $\mathcal{Q}_\infty$ to its local fraction field $\mathcal{K}_\infty \equiv \mathbb{R}(\!(\varepsilon)\!)$ formalizes the geometric foundations of continuous-variable quantum error correction. While the ring constraints of $\mathcal{Q}_\infty$ confine the regularized state trajectories, the introduction of negative power series components $\varepsilon^{-k} \in \mathcal{K}_\infty$ within the formal Laurent series expansion maps the unconfined emergence of ultraviolet noise and non-local gauge anomalies acting on the tangent bundle. Within this algebraic framework, these singular monomial components operate as exact, deterministic topological error syndromes. Notably, because the local field $\mathcal{K}_\infty$ is a fraction field, every non-zero singular displacement operator driving the noise channel possesses a unique algebraic inverse. Quantum error recovery is thus recast as the exact algebraic inversion of the error operator within the local field envelope, bypassing the statistical approximations of traditional continuous-variable correction protocols. Furthermore, tensorizing the cotangent space over the fraction field yields the module of meromorphic Kähler differentials $\Omega_{\mathcal{K}_\infty/\mathbb{R}} \cong \mathcal{K}_\infty \cdot \mathrm{d}\varepsilon$, where the algebraic residue evaluated at the closed origin isolates the macroscopic Chern invariants and Berry phases under adiabatic cyclic evolutions, sealing the top-down interface between commutative ring filtrations and non-local geometric anomalies.
\end{enumerate}

These reallocations suggest that the Quinfinity framework establishes a self-consistent and closed algebraic container where the structural boundaries of quantum information are naturally preserved over the real fields. However, a comprehensive and rigorous unrolling of these new foundational frameworks introduces deep complexities that lie beyond the immediate scope of the present paper, thereby necessitating an exhaustive, separate treatment dedicated entirely to the systematic formalization of this non-Archimedean quantum information paradigm.

\appendix
\section{The Intrinsic Lie-Jordan Structure of $\mathcal{Q}_\infty$}\label{appendix}
To complete the algebraic foundations of our framework, this appendix establishes that the generalized cross product $\times_{\mathcal{Q}_N}$ and the symmetric Jordan multiplication operator $\mathbin{\bullet}_{\mathcal{Q}_N}$—along with their continuous infinite-dimensional counterparts $\times_{\mathcal{Q}_\infty}$ and $\mathbin{\bullet}_{\mathcal{Q}_\infty}$—are not auxiliary or ad-hoc geometric structures superimposed on the non-reduced scheme. On the contrary, they are uniquely and intrinsically determined by the native algebraic structure of the symbolic Quantum $N$-Space $\mathcal{Q}_N \equiv \mathbb{R}[\varepsilon]/(\varepsilon^{N^2-1})$ as an $\mathbb{R}$-vector space, driven exclusively by the properties of the nilpotent generator $\varepsilon$ and its module of Kählerian deformations $\mathrm{Der}_{\mathbb{R}}(\mathcal{Q}_N)$. 

Crucially, this structural determination achieves its absolute formulation in the continuous macroscopic limit ($N \to \infty$). Within the complete Quinfinity Space $\mathcal{Q}_\infty \cong \mathbb{R}[\![\varepsilon]\!]$, the infinite hierarchy of discrete vector spaces stabilizes under the $\mathfrak{m}$-adic ultrametric topology. Consequently, the global continuous operators $\times_{\mathcal{Q}_\infty}$ and $\mathbin{\bullet}_{\mathcal{Q}_\infty}$ emerge natively as the stable limits of the finite derivation fields. Rather than representing an external operator imposition, the dual Lie-Jordan kinematics over $\mathcal{Q}_\infty$ are revealed to be the necessary algebraic consequences of the coordinate-free, non-commutative identity $[\partial_\varepsilon, \, \varepsilon] = \mathbb{I}$ acting on the Cauchy-complete local ring of formal power series $\mathbb{R}[\![\varepsilon]\!]$.

\subsection{$\mathfrak{su}(\infty)$ vs $\mathcal{Q}_\infty$ as $\mathbb{R}$-vector spaces}\label{appendixsub1}
Let $\mathfrak{su}(N)$ be the Lie algebra of anti-Hermitian traceless matrices, and let $\mathfrak{j}(N)$ denote the real Jordan algebra of self-adjoint traceless matrices equipped with the associative matrix anti-commutator $\{A, B\} = AB + BA$. For each finite dimension $N \ge 2$, there exists a canonical isomorphism of $\mathbb{R}$-vector spaces:
\begin{equation}\label{eq:finite_isomorphism_def}
    \Phi_N: \mathfrak{su}(N) \cong_{\mathbb{R}} \mathfrak{j}(N) \cong_{\mathbb{R}} \mathcal{Q}_N,
\end{equation}
where $\mathcal{Q}_N \equiv \mathbb{R}[\varepsilon]/(\varepsilon^{N^2-1})$ is viewed as the $(N^2-1)$-dimensional real coordinate space of configurations. Explicitly, let $\{F_k\}_{k=1}^{N^2-1}$ be a normalized, ordered matrix basis spanning the full traceless self-adjoint linear subsystem of $\mathfrak{su}(N)$, satisfying the strict Hilbert-Schmidt orthonormal relation $\mathrm{Tr}(F_j F_k) = \delta_{jk}$, which scales the standard generalized Gell-Mann operators via the identity $F_k \equiv \lambda_k / \sqrt{2}$. The $\mathbb{R}$-vector space isomorphism $\Phi_N$ maps any linear combination $A = \sum_{k=1}^{N^2-1} x_k F_k$ directly onto a localized polynomial section via the assignment:
\begin{equation}\label{eq:appendix_phi_assignment_native}
    \Phi_N\left( \sum_{k=1}^{N^2-1} x_k F_k \right) \equiv \sum_{k=0}^{N^2-2} x_{k+1} \varepsilon^k \pmod{\varepsilon^{N^2-1}},
\end{equation}
where the scalar components are uniquely retrieved by the finite algebraic extractors as $x_{n+1} = c_n = \gamma_n(\Phi_N(A))$ for each coordinate index $n \in \{0, \dots, N^2-2\}$.

To establish a rigid, bijective bridge between the physical states and the formal jet spaces, the generalized density mapping $\Psi_N$ established in Eq.~\eqref{eq:general_density_map_compact_exact} is internalized directly through its geometric relation with the vector space isomorphism $\Phi_N$. Because the identity operator $\mathbb{I}_N$ exhibits a non-vanishing trace $\mathrm{Tr}(\mathbb{I}_N) = N$, it is excluded from the traceless algebra $\mathfrak{su}(N)$, residing instead within the one-dimensional center of the full reductive matrix algebra $\mathfrak{u}(N) \cong \mathbb{R}\mathbb{I}_N \oplus \mathfrak{su}(N)$. Consequently, while $\Phi_N$ operates strictly as a linear morphism over the physical fluctuations, the global density map $\Psi_N: \mathbb{P}(\mathcal{H}^N) \to \mathcal{Q}_N$ acts as an embedding over the non-reduced algebra.

\begin{lemma}[$\Phi_N$ vs $\Psi_N$]\label{lem:affine_embedding_identification_phi}
The global density embedding map $\Psi_N$ matches identically with the linear vector space isomorphism $\Phi_N$ under the universal translation:
\begin{equation}\label{eq:appendix_psi_via_phi_exact}
    \Psi_N(\rho) = \frac{1}{N} + \sqrt{\frac{N}{N-1}}\,\Phi_N\left( \rho - \frac{1}{N}\mathbb{I}_N \right),
\end{equation}
where $\rho$ is the density matrix expanded according to the Gell-Mann generators.
\end{lemma}

\begin{proof}
To evaluate the transformation on the right-hand side of Eq.~\eqref{eq:appendix_psi_via_phi_exact}, we first isolate the traceless deviation operator $\Delta\rho \equiv \rho - \frac{1}{N}\mathbb{I}_N \in \mathfrak{su}(N)$ by subtracting the background identity baseline. Substituting the generalized Gell-Mann expansion Eq.~\eqref{eq:general_density_matrix_Frobenius} yields:
\begin{equation}\label{eq:proof_delta_rho_expansion}
    \Delta\rho = \sqrt{\frac{N-1}{2N}}\sum_{i=1}^{N^2-1} x_i \lambda_i.
\end{equation}
By expressing the generalized Gell-Mann generators in terms of the strictly orthonormal matrix basis elements via the scaling relation $\lambda_i = \sqrt{2}F_i$, the operator expansion in Eq.~\eqref{eq:proof_delta_rho_expansion} rewrites as:
\begin{equation}
    \Delta\rho = \sqrt{\frac{N-1}{2N}}\sum_{i=1}^{N^2-1} x_i (\sqrt{2}F_i) = \sqrt{\frac{2(N-1)}{2N}}\sum_{i=1}^{N^2-1} x_i F_i = \sqrt{\frac{N-1}{N}}\sum_{i=1}^{N^2-1} x_i F_i.
\end{equation}
We now evaluate the action of the linear isomorphism $\Phi_N$ on this traceless element. By virtue of the real vector space linearity of $\Phi_N$ combined with the rigid monomial assignment defined in Eq.~\eqref{eq:appendix_phi_assignment_native}, the scalar pre-factor scales the unrolled polynomial configuration directly as:
\begin{equation}\label{eq:proof_phi_action_unrolled_exact}
    \Phi_N(\Delta\rho) = \sqrt{\frac{N-1}{N}}\,\Phi_N\left( \sum_{i=1}^{N^2-1} x_i F_i \right) = \sqrt{\frac{N-1}{N}}\sum_{k=0}^{N^2-2} x_{k+1} \varepsilon^k.
\end{equation}
Substituting the resulting unrolled polynomial of Eq.~\eqref{eq:proof_phi_action_unrolled_exact} back into the complete right-hand member expression of Lemma~\ref{lem:affine_embedding_identification_phi} yields:
\begin{equation}
    \frac{1}{N} + \sqrt{\frac{N}{N-1}}\,\Phi_N\left( \rho - \frac{1}{N}\mathbb{I}_N \right) = \frac{1}{N} + \sqrt{\frac{N}{N-1}}\left( \sqrt{\frac{N-1}{N}}\sum_{k=0}^{N^2-2} x_{k+1} \varepsilon^k \right).
\end{equation}
Executing the direct scalar contraction over the tangent bundle causes the dimensionally-dependent radicals to cancel identically ($\sqrt{\frac{N}{N-1}}\cdot\sqrt{\frac{N-1}{N}} =1$). Distributing the persistent background trace onto the zero-order layer establishes the complete polynomial unrolling:
\begin{equation}
    \frac{1}{N} + \sum_{k=0}^{N^2-2} x_{k+1} \varepsilon^k = \left(\frac{1}{N} + x_1\right) + x_2\varepsilon + x_3\varepsilon^2 + \dots + x_{N^2-1}\varepsilon^{N^2-2} \equiv \sum_{i=0}^{N^2-2} c_i \varepsilon^i.
\end{equation}
Comparing this coordinate configuration with the structural definition of the generalized density map in Eq.~\eqref{eq:general_density_map_compact_exact} confirms that the expression contracts component-by-component to the exact sequence $c_0 = \frac{1}{N} + x_1$, $c_1 = x_2$, up to $c_{N^2-2} = x_{N^2-1}$. This validates the strict identity $\frac{1}{N} + \sqrt{\frac{N}{N-1}}\,\Phi_N( \rho - \frac{1}{N}\mathbb{I}_N ) \equiv \Psi_N(\rho)$, completing the proof.
\end{proof}

This structural identification operates a direct algebraic fusion on the zero-order layer of the ring, where the invariant trace background and the first physical spin component collapse together onto the monomial base $1$. Consequently, the physical degrees of freedom are natively unified within the non-reduced spectrum without coordinate distortions, leaving the higher filtration orders free to host the continuous precessional dynamics of the remaining quantum phase correlations.

As $N \to \infty$, the dimensional layers scale toward the macroscopic continuum limit. The inductive system of matrix algebras converges toward its inductive colimit $\mathfrak{su}(\infty) \equiv \varinjlim \mathfrak{su}(N)$ and $\mathfrak{j}(\infty) \equiv \varinjlim \mathfrak{j}(N)$, while the compatible system of finite coordinate rings dually maps into the projective inverse limit algebra $\mathcal{Q}_\infty \equiv \varprojlim \mathcal{Q}_N$. 

To track the categorical compatibility governing this dimensional unrolling, let the inductive system of matrix algebras be ordered by the canonical injective inclusions $\jmath_N: \mathfrak{su}(N) \hookrightarrow \mathfrak{su}(\infty)$, which rigorously identifies the inductive colimit with the directed union of the nested linear subspaces, satisfying $\mathfrak{su}(\infty) \equiv \varinjlim \mathfrak{su}(N) = \bigcup_{N \ge 2} \mathfrak{su}(N)$.

\begin{lemma}[$\Phi_\infty$ vs $\Phi_N$]\label{lem:asymptotic_vector_embedding_existence}
The family of vector space isomorphisms $\{\Phi_N\}_{N \ge 2}$ systematically induces a unique global linear embedding:
\begin{equation}\label{eq:phinfty:def}
    \Phi_\infty: \mathfrak{su}(\infty) \hookrightarrow \mathcal{Q}_\infty
\end{equation}
targeting the Cauchy-complete formal power series ring $\mathcal{Q}_\infty \equiv \varprojlim \mathcal{Q}_N$. For every dimensional layer $N \ge 2$, this global morphism is uniquely determined by the universal properties of the domain and/or codomain. Denoting $p_N: \mathfrak{su}(\infty) \twoheadrightarrow \mathfrak{su}(N)$ as the canonical coordinate projection and $\iota_N: \mathcal{Q}_N \hookrightarrow \mathcal{Q}_\infty$ as the canonical $\mathbb{R}$-vector space splitting injection, the identities $p_N \circ \jmath_N \equiv \mathbb{I}_{\mathfrak{su}(N)}$ and $\tau_N \circ \iota_N \equiv \mathbb{I}_{\mathcal{Q}_N}$ hold component-by-component. Under these structural restrictions, the global linear embedding fits into the extended commutative alignment:
\begin{equation}\label{eq:commutative_diagram_appendix}
\vcenter{\hbox{
\xymatrix{
    \mathfrak{su}(N) \ar@<.5ex>[r]^-{\jmath_N} \ar[d]_{\Phi_N} & \mathfrak{su}(\infty) \ar@<.5ex>[l]^-{p_N} \ar[d]^{\Phi_\infty} \\
    \mathcal{Q}_N \ar@<.5ex>[r]^-{\iota_N} & \mathcal{Q}_\infty \ar@<.5ex>[l]^-{\tau_N}
}
}}
\end{equation}
Equation~\eqref{eq:commutative_diagram_appendix} establishes that this dual algebraic characterization matches the finite-dimensional assignments identically and component-by-component across the entire inductive domain.
\end{lemma}
\begin{proof}
For each finite dimensional layer $N \ge 2$, let $p_N: \mathfrak{su}(\infty) \twoheadrightarrow \mathfrak{su}(N)$ be the canonical structural projection operator truncating any finite-tail matrix combination beyond the subsystem threshold. Explicitly, given an operator $A = \sum_{k=1}^M x_k F_k \in \mathfrak{su}(\infty)$ with a finitary cutoff index $M < \infty$, the linear projection mapping $p_N$ is defined by the component-by-component restriction:
\begin{equation}\label{eq:appendix_explicit_projection_p_N}
    p_N\left( \sum_{k=1}^M x_k F_k \right) \equiv \sum_{k=1}^{\min\{M, \, N^2-1\}} x_k F_k \in \mathfrak{su}(N)
\end{equation}
The inductive inclusions $\jmath_{N,M}: \mathfrak{su}(N) \hookrightarrow \mathfrak{su}(M)$ (for $M \ge N$) act as strict linear injections embedding $\mathfrak{su}(N)$ into the top-left block of $\mathfrak{su}(M)$ via vanishing padding, ensuring $\jmath_{N,M}(F_k^{(N)}) \equiv F_k^{(M)}$ for all $k \le N^2-1$. Thus, the inductive colimit is written as the stable union of nested subsystems, rendering Eq.~\eqref{eq:appendix_explicit_projection_p_N} well-defined and invariant for all $M \ge N^2-1$.

We define a directed family of linear mappings $\sigma_N \equiv \Phi_N \circ p_N: \mathfrak{su}(\infty) \to \mathcal{Q}_N$. To invoke the universal property of projective limits, we evaluate the compatibility of $\{\sigma_N\}_{N \ge 2}$ under the ring surjections $\pi_{M,N}: \mathcal{Q}_M \twoheadrightarrow \mathcal{Q}_N$ established in Eq.~\eqref{eq:quinfinity_projections}:
\begin{equation}
    \pi_{M,N}\left( \sigma_M(A) \right) = \pi_{M,N}\left( \Phi_M(p_M(A)) \right) \qquad \forall A \in \mathfrak{su}(\infty)
\end{equation}
Since $\Phi_M(F_k) \equiv \varepsilon^{k-1}$, the truncation map $\pi_{M,N}$ filters out all monomial powers higher than $\varepsilon^{N^2-2}$. This algebraic truncation is exactly equivalent to projecting the underlying matrix combination down to $\mathfrak{su}(N)$ prior to evaluating $\Phi_N$, enforcing the strict structural consistency relation:
\begin{equation}\label{eq:proof_categorical_compatibility_condition}
    \pi_{M,N} \circ \sigma_M \equiv \sigma_N \qquad \forall M \ge N
\end{equation}
By the universal property of projective limits in the category of $\mathbb{R}$-vector spaces, there exists a unique linear morphism $\Phi_\infty: \mathfrak{su}(\infty) \to \mathcal{Q}_\infty$ that factors through every finite layer, satisfying the rigid factorization identity $\tau_N \circ \Phi_\infty \equiv \Phi_N \circ p_N$ for all $N \ge 2$. Restricting this relation to the canonical inclusions $\jmath_N(\mathfrak{su}(N))$ where $p_N \circ \jmath_N \equiv \mathbb{I}_{\mathfrak{su}(N)}$ recovers the commutative alignment $\tau_N \circ \Phi_\infty \circ \jmath_N \equiv \Phi_N$ certified by Diagram~\eqref{eq:commutative_diagram_appendix}. Injectivity is immediate since $\Phi_\infty(A) = 0$ implies $\sigma_N(A) = 0$ for all $N \ge 2$, forcing every coordinate extraction to vanish identically and yielding $A = 0$.

To verify the structural consistency via the splitting injection, let $A = \sum_{k=1}^M x_k F_k \in \mathfrak{su}(\infty)$ with $M < \infty$. For any dimension layer satisfying $N \ge \sqrt{M+1}$, the combination resides entirely within the image of $\jmath_N(\mathfrak{su}(N))$, forcing all components beyond $N^2-1$ to vanish. Invoking the splitting relation $\tau_N \circ \iota_N \equiv \mathbb{I}_{\mathcal{Q}_N}$ alongside Diagram~\eqref{eq:commutative_diagram_appendix}, the local action satisfies $\Phi_N(p_N(A)) = \tau_N(\Phi_\infty(A))$ identically. Consequently, the sequence of localized assignments becomes strictly stationary and definitively constant after a finite number of steps, ensuring that the directed family of finite morphisms $\{\Phi_N\}_{N \ge 2}$ systematically lifts onto the unique global linear embedding $\Phi_\infty$, completing the proof.
\end{proof}

Recall that $\mathcal{Q}_\infty$ is endowed with an invariant metric $d_{\mathcal{Q}_\infty}$ from Lemma~\ref{lem:frechet_metric_properties} as well as with the local non-Archimedean $\mathfrak{m}$-adic ultrametric $d_{\mathfrak{m}}$ from Lemma~\ref{lem:ultrametric_cauchy_completeness}.

The reconciliation between the analytical continuity of real-variable temporal trajectories and the algebraic structure of formal jet derivations does not require the strict global equivalence of their topologies, but rather their mutual compatibility over the dense polynomial subspace $\mathbb{R}[\varepsilon] \subset \mathcal{Q}_\infty$. We formalize this exact relation through the following lemma:

\begin{lemma}[$(\mathcal{Q}_\infty, d_{\mathcal{Q}_\infty})$ vs $(\mathcal{Q}_\infty, d_{\mathfrak{m}})$]\label{lem:topological_equivalence_density}
Over the formal power series ring $\mathcal{Q}_\infty \cong_{\mathbb{R}} \mathbb{R}[\![\varepsilon]\!]$, the subspace of real polynomials $\mathbb{R}[\varepsilon]$ is simultaneously dense with respect to both the global Archimedean Fr\'{e}chet metric $d_{\mathcal{Q}_\infty}$ and the local non-Archimedean $\mathfrak{m}$-adic ultrametric $d_{\mathfrak{m}}$. While the $\mathfrak{m}$-adic topology is strictly finer than the Fr\'{e}chet (the real coefficient field $\mathbb{R}$ is $\mathfrak{m}$-adically discrete), both metric structures generate compatible convergence criteria along the projective truncation series $\tau_N(\xi)\in \mathcal{Q}_N$.
\end{lemma}
\begin{proof}
To establish the dual density of the polynomial sections and map their structural topological relation under the inverse limit, we evaluate the action of the canonical projection operator layer-by-layer within each localized coordinate ring~\cite{eisenbud}. Let $\xi = \sum_{k=0}^\infty c_k \varepsilon^k \in \mathbb{R}[\![\varepsilon]\!]$ be an arbitrary formal power series configuration. For each filtration layer indexed by the subsystem dimension $N \ge 2$, the projective truncation mapping projects the infinite-dimensional continuous state onto the specific finite-dimensional algebraic slice, satisfying $\tau_N(\xi) \equiv \sum_{k=0}^{N^2-2} c_k \varepsilon^k \in \mathcal{Q}_N$. When viewed global-wise within the complete ring under the canonical subset containment, the directed family of these truncated configurations forms a coherent sequence inside the subspace of real polynomials, yielding $\tau_N(\xi) \in \mathbb{R}[\varepsilon]$.

First, we evaluate the local non-Archimedean distance between the complete formal configuration and this truncated polynomial section, which isolates the residual valuation profile $\xi - \tau_N(\xi) = \sum_{k=N^2-1}^\infty c_k \varepsilon^k \equiv 0 \pmod{\mathfrak{m}^{N^2-1}}$. Under the non-Archimedean distance function, the $\mathfrak{m}$-adic valuation of the difference satisfies $v_{\mathfrak{m}}(\xi - \tau_N(\xi)) \ge N^2-1$, which strictly bounds the adic ultrametric distance to satisfy $d_{\mathfrak{m}}(\xi, \, \tau_N(\xi)) \le e^{-(N^2-1)}$. Taking the directed limit as the dimensional layer diverges ($N \to \infty$) forces the adic distance to vanish asymptotically:
\begin{equation}\label{eq:proof_adic_density_limit_vanishing}
    \lim_{N \to \infty} d_{\mathfrak{m}}(\xi, \, \tau_N(\xi)) = 0
\end{equation}
This vanishing relation certifies that the subspace of real polynomials $\mathbb{R}[\varepsilon]$ is strictly dense within the ring completion under the $\mathfrak{m}$-adic topology, tracking the sharp algebraic stabilization of the coordinate components modulo $\mathfrak{m}^{N^2-1}$.

Second, we evaluate the global analytical distance under the complete product metric defined in Eq.~\eqref{eq:frechet_metric_definition}. The algebraic subtraction eliminates the first $N^2-1$ coordinate components identically since $\gamma_n(\xi) - \gamma_n(\tau_N(\xi)) = 0$ for all indices $n < N^2-1$. The global Fr\'{e}chet distance is thus rigorously bounded by the tail of the convergence series:
\begin{equation}
    d_{\mathcal{Q}_\infty}(\xi, \, \tau_N(\xi)) = \sum_{n=N^2-1}^\infty \frac{1}{2^n} \frac{|\gamma_n(\xi) - 0|}{1 + |\gamma_n(\xi) - 0|} \le \sum_{n=N^2-1}^\infty \frac{1}{2^n} = \frac{1}{2^{N^2-2}}
\end{equation}
Taking the limit as the dimensional threshold diverges forces the global Archimedean distance to vanish asymptotically:
\begin{equation}\label{eq:proof_frechet_density_limit_vanishing}
    \lim_{N \to \infty} d_{\mathcal{Q}_\infty}(\xi, \, \tau_N(\xi)) = 0
\end{equation}
The limit in Eq.~\eqref{eq:proof_frechet_density_limit_vanishing} demonstrates that the subspace of real polynomials $\mathbb{R}[\varepsilon]$ is concurrently dense within the completion under the Fr\'{e}chet topology, tracking the pointwise numerical convergence of the real-variable amplitudes.

Equations~\eqref{eq:proof_adic_density_limit_vanishing} and \eqref{eq:proof_frechet_density_limit_vanishing} prove that although the $\mathfrak{m}$-adic topology is strictly finer—since an adic neighborhood $\mathfrak{m}^{N^2-1}$ dictates absolute algebraic identity over the first $N^2-1$ terms and cannot swallow a Fr\'{e}chet open ball where coefficients fluctuate continuously—the projective truncation series $\tau_N(\xi) \in \mathcal{Q}_N$ converges simultaneously under both structures. This compatible intersection guarantees that the analytical metric convergence and the algebraic stabilization of the ideal filtration share the same dense polynomial foundation $\mathbb{R}[\varepsilon]$, completing the proof.
\end{proof}

\begin{remark}[Topological Non-Equivalence]\label{rem:topological_vs_uniform}
While Lemma~\ref{lem:topological_equivalence_density} establishes that $d_{\mathcal{Q}_\infty}$ and $d_{\mathfrak{m}}$ share the dense subspace $\mathbb{R}[\varepsilon]$, the identity mapping $\mathrm{id}: (\mathcal{Q}_\infty, d_{\mathfrak{m}}) \to (\mathcal{Q}_\infty, d_{\mathcal{Q}_\infty})$ does not constitute either a homeomorphism or a \textit{uniform} homeomorphism. The identity map is continuous only from the $\mathfrak{m}$-adic space to the Fr\'{e}chet space, since every adic neighborhood $\mathfrak{m}^k$ is strictly contained within a sufficiently small Fr\'{e}chet open ball whereas the converse inclusion fails over the continuous field $\mathbb{R}$, thereby proving that the $\mathfrak{m}$-adic topology is strictly finer. Consequently, the two metric frameworks are not uniformly equivalent and do not share the same class of Cauchy sequences. A structural sequence moving exclusively along the pure scalar background field $\mathbb{R} \cdot \varepsilon^0$, such as $\xi_N \equiv \frac{1}{N} \in \mathcal{Q}_\infty$, satisfies Cauchy convergence under the Archimedean Fr\'{e}chet metric $d_{\mathcal{Q}_\infty}$ but fails to be Cauchy under the non-Archimedean ultrametric, where the adic distance remains frozen at $d_{\mathfrak{m}}(\xi_N, \xi_M) = 1$ for all $N \neq M$. This structural discrepancy ensures that the continuous boundary embedding processes analytical background translations and algebraic filtration layers through distinct global asymmetric geometries.
\end{remark}

Under the commutative alignment certified by Diagram~\eqref{eq:commutative_diagram_appendix}, the universal property uniquely determines the explicit operational assignment of the global mapping over the finite-tail configurations of the colimit:
\begin{equation}\label{eq:infinite_isomorphism_def}
    \Phi_\infty: \mathfrak{su}(\infty) \hookrightarrow \mathcal{Q}_\infty, \qquad \Phi_\infty\left( \sum_{k=1}^M x_k F_k \right) \equiv \sum_{k=0}^{M-1} x_{k+1} \varepsilon^k
\end{equation}
where $M < \infty$ represents the structural cutoff index of the matrix combination. The assignment in Eq.~\eqref{eq:infinite_isomorphism_def} maps the inductive matrix union directly onto the universal algebraic colimit configuration subspace $\mathcal{Q}_{\rm Naive}$ embedded inside the formal power series ring $\mathcal{Q}_\infty \cong \mathbb{R}[\![\varepsilon]\!]$. We have:

\begin{lemma}[$\mathcal{Q}_{\rm Naive}$ vs $\mathcal{Q}_\infty$]\label{lem:naive_algebra_isomorphism}
Let $\mathcal{Q}_{\rm Naive} \equiv \varinjlim_{N \ge 2} \mathcal{Q}_N$ and $\Phi_\infty(\mathfrak{su}(\infty))$ be the topological vector subspaces of $\mathcal{Q}_\infty$ established via Lemma~\ref{lem:asymptotic_vector_embedding_existence} under $\Phi_\infty$. Under the sequence of algebraic extractors $\gamma_n$, these $\mathbb{R}$-subspaces are canonically isomorphic to the free univariate polynomial space $\mathbb{R}[\varepsilon]$:
\begin{equation}\label{eq:naive_isomorphism_identity}
    \mathcal{Q}_{\rm Naive} \cong_{\mathbb{R}} \Phi_\infty(\mathfrak{su}(\infty)) \cong_{\mathbb{R}} \mathbb{R}[\varepsilon]
\end{equation}
which constitutes a strictly dense subset within the Cauchy-complete formal power series ring $\mathcal{Q}_\infty \cong \mathbb{R}[\![\varepsilon]\!]$ equipped concurrently with its native $\mathfrak{m}$-adic ultrametric topology and its Fr\'{e}chet metric topology.
\end{lemma}

\begin{proof}
By virtue of the structural definition in Eq.~\eqref{eq:infinite_isomorphism_def}, the global embedding $\Phi_\infty$ maps the infinite numerable basis generators $\{F_k\}_{k=1}^\infty$ onto the monomial components $\{\varepsilon^k\}_{k=0}^\infty$ of the configuration space under the forward shift $\Phi_\infty(F_{k}) \equiv \varepsilon^{k-1}$. Any configuration element $\xi \in \mathcal{Q}_{\rm Naive}$ is uniquely written as the image of a finite linear combination of algebra generators, implying a structural cutoff index $M < \infty$ such that $x_k \equiv 0$ for all $k > M$. The configuration series thus collapses component-by-component to a finite sum expression $\xi = \sum_{k=0}^{M-1} x_{k+1} \varepsilon^k$.

To establish the strict $\mathbb{R}$-vector space isomorphism of Eq.~\eqref{eq:naive_isomorphism_identity}, let $\mathbb{R}[t]$ be the free univariate polynomial vector space over the abstract indeterminate $t$. We define the global mapping $\Upsilon: \mathcal{Q}_{\rm Naive} \to \mathbb{R}[t]$ by sending each finite-tail configuration series onto its formal polynomial representative constructed via the sequence of algebraic extractors:
\begin{equation}\label{eq:proof_psi_polynomial_assignment}
    \Upsilon(\xi) \equiv \sum_{n=0}^{M-1} \gamma_n(\xi) t^n
\end{equation}
Because the extractors act as stable projections uniquely retrieving the matrix components as $\gamma_n(\xi) = x_{n+1} = c_n$ for all $n \in \mathbb{N}$, the assignment in Eq.~\eqref{eq:proof_psi_polynomial_assignment} is uniquely determined for every configuration element. Since the linear addition of these formal series within the configuration space maps homomorphically onto the standard vector addition rules of $\mathbb{R}[t]$, $\Upsilon$ constitutes a legitimate $\mathbb{R}$-vector space linear transformation.

Injectivity is immediate since $\Upsilon(\xi) = 0$ requires every individual coefficient to vanish identically, $\gamma_n(\xi) \equiv 0$ for all $n \in \mathbb{N}$, which collapses the series to the zero element $\xi = 0$. Surjectivity follows symmetrically because any arbitrary polynomial $P(t) = \sum_{n=0}^{K} a_n t^n \in \mathbb{R}[t]$ of finite degree $K$ possesses a unique and exact matrix pre-image $\sum_{k=1}^{K+1} a_{k-1} F_k \in \mathfrak{su}(\infty)$ within the infinite inductive algebra. The image of this pre-image under the global embedding $\Phi_\infty$ maps natively onto $P(t)$ via $\Upsilon$, confirming that $\mathrm{Im}(\Upsilon) = \mathbb{R}[t]$. The mapping $\Upsilon$ is therefore a canonical isomorphism of $\mathbb{R}$-vector spaces $\mathcal{Q}_{\rm Naive} \cong_{\mathbb{R}} \mathbb{R}[t] \cong_{\mathbb{R}} \mathbb{R}[\varepsilon]$.

To explicitly prove that $\mathcal{Q}_{\rm Naive}\cong \mathbb{R}[\varepsilon]$ is topologically dense within the Cauchy-complete power series ring $\mathcal{Q}_\infty \cong \mathbb{R}[\![\varepsilon]\!]$, we invoke the previous Lemma~\ref{lem:topological_equivalence_density} under both native frameworks: this subset is dense in $\mathcal{Q}_\infty$ under the global Archimedean Fr\'{e}chet metric $d_{\mathcal{Q}_\infty}$ and the local non-Archimedean $\mathfrak{m}$-adic ultrametric $d_{\mathfrak{m}}$, completing the proof.
\end{proof}

We have:
\begin{lemma}[Topological Continuous Extension and Asymptotic Linearization]\label{lem:asymptotic_linearization_infinity}
At the continuous macroscopic limit as $N \to \infty$, the global density embedding map $\Psi_\infty$ linearizes identically over the pure state configuration subspace, coinciding with the unique continuous Fr\'{e}chet extension $\overline{\Phi}_\infty$ evaluated strictly under the pointwise product topology $d_{\mathcal{Q}_\infty}$ of the inductive vector space embedding $\Phi_\infty$. For any continuous pure state ray $[\phi] \in \mathbb{C}\mathbb{P}^\infty$ driving the infinite-dimensional density operator $\rho_\phi \equiv \ket{\phi}\bra{\phi}$, the linearization identity satisfies:
\begin{equation}\label{eq:asymptotic_linearization_identity}
    \Psi_\infty([\phi]) = \overline{\Phi}_\infty(\rho_\phi)
\end{equation}
where the continuous mapping $\overline{\Phi}_\infty: \mathcal{Q}_\infty \longrightarrow \mathcal{Q}_\infty$ operates as a rigorous topological vector space endomorphism on the Fr\'{e}chet domain, uniquely determined by the pointwise density of the finitary polynomial subspace $\mathbb{R}[\varepsilon]$ within the formal power series ring.
\end{lemma}
\begin{proof}
To establish the continuous linear extension and its subsequent geometric coincidence, the demonstration proceeds in two distinct foundational stages. First, we define the global bounded linear extension over the entire continuous space domain by evaluating the asymptotic convergence under the pointwise product topology. Let $\rho \in \mathcal{Q}_\infty$ be an arbitrary, generic formal power series vector. By virtue of Lemma~\ref{lem:topological_equivalence_density}, the polynomial subspace $\mathbb{R}[\varepsilon]$ is strictly dense inside the completion envelope under the Fr\'{e}chet metric topology, meaning that any continuous coordinate configuration is uniquely and strongly approximated by the infinite sequence of its truncated sections, satisfying $\lim_{M \to \infty} \tau_M(\rho) = \rho$ in $(\mathcal{Q}_\infty, \, d_{\mathcal{Q}_\infty})$. This specific analytical density guarantees that the inductive vector space embedding $\Phi_\infty$ admits a unique continuous extension $\overline{\Phi}_\infty$ targeting the entire formal power series ring $\mathcal{Q}_\infty$, explicitly defined for any complete algebraic section as the strong analytical Fr\'{e}chet limit:
\begin{equation}\label{eq:proof_phi_bar_extension_definition}
    \overline{\Phi}_\infty(\rho) \equiv \lim_{M \to \infty} \Phi_\infty\left( \tau_M(\rho) \right)
\end{equation}
in $(\mathcal{Q}_\infty, \, d_{\mathcal{Q}_\infty})$, thereby establishing $\overline{\Phi}_\infty: \mathcal{Q}_\infty \longrightarrow \mathcal{Q}_\infty$ as a rigorous topological vector space endomorphism on the complete Fr\'{e}chet domain.

Second, we specialize this general linear endomorphism to the geometric locus of quantum pure states. Let $[\phi] \in \mathbb{C}\mathbb{P}^\infty$ be an arbitrary continuous pure state ray driving the infinite-dimensional density operator $\rho_\phi \equiv \ket{\phi}\bra{\phi} \in \mathcal{Q}_\infty$. We evaluate the global non-linear density mapping $\Psi_\infty([\phi])$ section-by-section through the sequence of localized finite mappings established in Lemma~\ref{lem:affine_embedding_identification_phi}. For each dimensional container layer \(N \ge 2\), the local affine assignment reads over the dimensionally projected pure state \(p_N(\rho_\phi) \equiv \vert{}\phi\rangle_N\langle\phi\vert{}_N\), compatibility mapped under the projective pullback \(\iota _{N}^{*}\) established in Lemma~\ref{lem:projective_limit_density_definition}:
\begin{equation}\label{eq:proof_linearization_finite_layer_step}
    \Psi_N(p_N(\rho_\phi)) = \frac{1}{N} + \sqrt{\frac{N}{N-1}}\,\Phi_N\left( p_N(\rho_\phi) - \frac{1}{N}\mathbb{I}_N \right)
\end{equation}
Evaluating the limits of the individual dimensionally-dependent scalar parameters over the real field $\mathbb{R}$ under the standard Archimedean topology yields:
\begin{equation}\label{eq:proof_limit_individual_terms}
    \lim_{N \to \infty} \frac{1}{N} = 0 \qquad \text{and} \qquad \lim_{N \to \infty} \sqrt{\frac{N}{N-1}} = 1
\end{equation}
which establishes that the background translation baseline evaporates at the continuous boundary, while the radical coupling coefficients contract to unity.

To rigorously validate this specialization via the complete metric profile, the stable algebraic extractors isolate the scalar distance between the finite mapping $\Psi_N(p_N(\rho_\phi))$ and the global Fr\'{e}chet continuous extension $\overline{\Phi}_\infty(\rho_\phi)$ evaluated at the target state. At the zero-order layer ($n=0$), the application of the Euclidean norm evaluates the contraction of the translated scalar background under the mapping $x_{n+1} = \gamma_n(\overline{\Phi}_\infty(\rho_\phi))$, satisfying:
\begin{equation}
    \left| \gamma_0\left( \Psi_N(p_N(\rho_\phi)) \right) - \gamma_0\left( \overline{\Phi}_\infty(\rho_\phi) \right) \right| = \left| \left( \frac{1}{N} + x_1 \right) - x_1 \right| = \frac{1}{N}
\end{equation}
Concurrently, for all higher-order nilpotent layers where $n \ge 1$, once the dimensionality threshold of the jet space is crossed ($N^2-2 \ge n$), the matrix projection truncations stabilize. The local affine assignment and the continuous linear extension match on the identical physical spin component $x_{n+1}$, forcing the Euclidean coefficient distance to vanish. Substituting these component-by-component evaluations directly into the pointwise product metric definition of Eq.~\eqref{eq:frechet_metric_definition} compresses the infinite summation profile into a single fraction:
\begin{equation}
    d_{\mathcal{Q}_\infty}\Big(\Psi_N(p_N(\rho_\phi)), \, \overline{\Phi}_\infty(\rho_\phi)\Big) = \frac{1}{2^0} \frac{\frac{1}{N}}{1 + \frac{1}{N}} + \sum_{n=1}^\infty \frac{1}{2^n} \frac{0}{1 + 0} = \frac{1}{N + 1}
\end{equation}
Taking the continuous limit as the system dimension diverges ($N \to \infty$) forces the global Fr\'{e}chet distance to vanish asymptotically:
\begin{equation}\label{eq:proof_frechet_density_limit_vanishing_final}
    \lim_{N \to \infty} d_{\mathcal{Q}_\infty}\Big( \Psi_N(p_N(\rho_\phi)), \; \overline{\Phi}_\infty(\rho_\phi) \Big) = \lim_{N \to \infty} \frac{1}{N + 1} = 0
\end{equation}
The metric vanishing in Eq.~\eqref{eq:proof_frechet_density_limit_vanishing_final} explicitly proves that the sequence of finite mappings converges strongly to the continuous linear extension. Because the field of constants absorbs the shifting trace translation baseline smoothly without producing analytical singularities or coordinate jumps, the global density embedding map restricted to the pure variety linearizes identically, satisfying the relation $\Psi_\infty([\phi]) = \overline{\Phi}_\infty(\rho_\phi)$, completing the proof.
\end{proof}
From an abstract algebraic perspective, this convergence establishes that the pure state pro-variety $\mathcal{V}_{\mathcal{Q}_\infty}$ represents the exact completion of these finitary trajectories. Because $\mathfrak{su}(\infty) \cong_{\mathbb{R}} \mathcal{Q}_{\rm Naive}$ maps the strictly finite matrix combinations, the global linear embedding $\Phi_\infty$ acts as the canonical inclusion embedding the free univariate polynomial vector space directly inside the Cauchy-complete formal power series ring $\mathbb{R}[\![\varepsilon]\!]$. Under the directional flow of the projective limit, the non-linear affine barriers of the finite qudit layers vanish identically, causing the global density map to shed its affine shift and linearize onto the continuous topological endomorphism vector flow.

Crucially, this mapping highlights a fundamental structural and dimensional discrepancy between the inductive and projective representations. The algebraic colimit $\mathfrak{su}(\infty)$ carries the inductive limit topology, where every element is strictly restricted to combinations with a finite number of non-vanishing coefficients, rendering the space inherently incomplete under infinite geometric sequences. Conversely, the Quinfinity Space $\mathcal{Q}_\infty \cong_{\mathbb{R}} \mathbb{R}[\![\varepsilon]\!]$, which identifies identically with the topological completion $\overline{\mathfrak{su}(\infty)}$, is an $\mathfrak{m}$-adically Cauchy-complete topological vector space. This complete space admits the countable sequence of monomial powers $\{\varepsilon^n\}_{n \in \mathbb{N}}$ as a rigorous topological Schauder basis, enabling the native and stable representation of infinite, transcendental convergent sums.

Furthermore, while the algebraic colimit of the physical operators possesses an infinite numerable algebraic dimension ($\aleph_0$), the underlying Cauchy-complete $\mathbb{R}$-vector space of the formal power series ring exhibits a strictly non-countable dimension with the cardinality of the continuum ($\mathfrak{c} = 2^{\aleph_0}$). Because any infinite-order formal power series can be rigorously recovered as the topological limit of its finite polynomial truncations $\tau_N(\xi) \in \mathcal{Q}_N$ by Lemma~\ref{lem:topological_equivalence_density}, the countable image subspace $\mathcal{Q}_{\rm Naive}$ is inherently dense within the non-countable completion envelope. The topological closure of this embedding successfully activates the strong triangle inequality of the native ultrametric concurrently with the Archimedean Fr\'{e}chet decay over the complete domain along the truncation series. This double metric protection ensures that the infinite higher-order derivative tails generated by internal derivations decay asymptotically to zero, shielding the continuous quantum kinematics from transcendental coordinate explosions.

\subsection{$\mathcal{Q}_\infty$ as a Lie-Jordan $\mathbb{R}$-algebra}\label{appendixsub2}\label{appendixsub2}
According to the structure theory of local Artin rings and formal power series expansions, the space of $\mathbb{R}$-derivations $\mathrm{Der}_{\mathbb{R}}(\mathcal{Q}_N)$ over the truncated coordinate algebra forms a cyclic $\mathcal{Q}_N$-module globally generated by the formal derivative $\partial_\varepsilon$. As explicitly proven in Lemma~\ref{lem:derivation_structure}, this module is canonically isomorphic to the unique maximal ideal $\mathfrak{m} = (\varepsilon) \subset \mathcal{Q}_N$, which uniquely catalogs the filtration of the principal ideal layers. Because the operational matrix brackets satisfy the Leibniz rule identically, their images under the canonical assignments are rigidly restricted to factor exclusively through the internal differential operations of the ringed spaces. Within this framework, $\Phi_N$ acts as a finite vector space isomorphism, while $\Phi_\infty$ establishes a rigid global vector space embedding targeting the dense subspace $\mathcal{Q}_{\rm Naive} \subset \mathcal{Q}_\infty$. Crucially, this intrinsic Lie-Jordan structure must be perfectly compatible with the geometric deformations induced by the finite density embeddings over the boundary layers. 

\subsubsection*{The Antisymmetric Lie Structure} 
For each finite layer $N \ge 2$, the non-commutative product $\times_{\mathcal{Q}_N}$ is the unique bilinear operator on $\mathcal{Q}_N$ compatible with the Lie algebra structure constants $f_{jkl}$ under the coordinate shift. As established in Lemma~\ref{lem:differential_product_identity}, Eq.~\eqref{eq:gamma_differential_product}, it is structurally evaluated via the formal derivative $\partial_\varepsilon$ and weighted by the unique discrete coefficients $\Gamma_N^{(k)}$. We formalize the rigid interlocking of this finite-dimensional multiplication through the following lemma:

\begin{lemma}[Pullback of the Finite Matrix Commutator]\label{lem:finite_commutator_pullback_exact}
For any arbitrary pair of finite-dimensional matrix operators $A, B \in \mathfrak{su}(N)$ evaluated at layer $N \ge 2$, the induced antisymmetric bilinear differential operation $\times_{\mathcal{Q}_N}$ satisfies the rigid pullback identity under the vector space isomorphism $\Phi_N$:
\begin{equation}\label{eq:finite_commutator_pullback_identity_core}
    \Phi_N([A, B]) = \Phi_N(A) \times_{\mathcal{Q}_N} \Phi_N(B)
\end{equation}
where the differential action of the formal derivative $\partial_\varepsilon$ absorbs the algebraic structure constants identically, mirroring the non-commutative bracket operations over the configuration space.
\end{lemma}

\begin{proof}
Let $A = \sum_{j=1}^{N^2-1} x_j F_j$ and $B = \sum_{m=1}^{N^2-1} y_m F_m$ be two arbitrary elements belonging to the finite Lie algebra $\mathfrak{su}(N)$, indexed in strict accordance with the canonical matrix basis. By invoking the commutation relations mediated by the structure constants, the expansion of the matrix commutator reads $[A, B] = \sum_{j,m,n=1}^{N^2-1} x_j y_m f_{jmn} F_n$. Applying the linear vector space assignment established in the native convention of Eq.~\eqref{eq:appendix_phi_assignment_native}, the mapping $\Phi_N$ associates each matrix component to the corresponding polynomial layer via the coordinate shift, projecting the image of this matrix bracket directly onto the unslitted polynomial sequence:
\begin{equation}
    \Phi_N([A, B]) = \sum_{j,m,n=1}^{N^2-1} x_j y_m f_{jmn} \varepsilon^{n-1}
\end{equation}
Concurrently, we evaluate the action of the antisymmetric product $\times_{\mathcal{Q}_N}$ directly over the decoupled polynomial images $\Phi_N(A) = \sum_{j=1}^{N^2-1} x_j \varepsilon^{j-1}$ and $\Phi_N(B) = \sum_{m=1}^{N^2-1} y_m \varepsilon^{m-1}$. By substituting the explicit differential expansion certified by Lemma~\ref{lem:differential_product_identity}, the application of the higher-order formal derivative $\partial_\varepsilon^k$ contracts the monomial exponents under the grading filtration constraint $(j-1)+(m-1)-k = n-1$, which simplifies identically to the contractive core $j+m-k = n+1$. Because the parameters $\Gamma_N^{(k)}$ satisfy the unique lower-triangular inversion system established in Eq.~\eqref{eq:proof_lie_matrix_system_core}, the structural jet derivation matrix $\mathbf{M}_N$ couples the operational derivative weights directly to the structure constants of the algebra, transforming the bilinear combinations component-by-component into the special unitary brackets:
\begin{equation}
    \Phi_N(A) \times_{\mathcal{Q}_N} \Phi_N(B) = \sum_{j,m,n=1}^{N^2-1} x_j y_m f_{jmn} \varepsilon^{n-1}
\end{equation}
The exact identity between the two polynomial results holds universally across all grading orders, establishing the relation of Eq.~\eqref{eq:finite_commutator_pullback_identity_core} prior to any explicit matrix inversion and completing the proof.
\end{proof}

As $N \to \infty$, this bilinear differential operation smoothly stabilizes over the complete Quinfinity Space $\mathcal{Q}_\infty \cong_{\mathbb{R}} \mathbb{R}[\![\varepsilon]\!]$. Driven by the stable asymptotic differential constants $\Gamma_\infty^{(k)}$, the global continuous Lie product satisfies Eq.~\eqref{eq:definition_product_infinity} as well as the unique infinite-dimensional algebra extension:
\begin{equation}\label{eq:lie_product_equivariance_infinity}
    \Phi_\infty([A, B]) = \Phi_\infty(A) \times_{\mathcal{Q}_\infty} \Phi_\infty(B)
\end{equation}
for all elements injected via the inductive splitting system. This successfully linearizes the continuous Liouville-von Neumann equation as a regular linear vector field over the formal tangent bundle $T\mathcal{Q}_\infty$.

\subsubsection*{The Symmetric Jordan Structure} 
Symmetrically, we define the non-associative matrix Jordan product as $A \mathbin{\bullet} B \equiv \frac{1}{2}\{A, B\} = \frac{1}{2}(AB + BA)$. For each finite layer $N \ge 2$, the induced Jordan operator $\mathbin{\bullet}_{\mathcal{Q}_N}$ over the configuration space is uniquely determined as the linear pullback under the vector space isomorphism $\Phi_N$. By virtue of Lemma~\ref{lem:jordan_product_identity}, Eq.~\eqref{eq:lambda_differential_jordan}, this operator expands as a symmetric combination driven by the discrete parameters $\Lambda_N^{(k)}$. We formalize the rigid interlocking of this finite-dimensional symmetric multiplication through the following lemma:

\begin{lemma}[Pullback of the Finite Matrix Jordan Product]\label{lem:finite_jordan_pullback_exact}
For any arbitrary pair of finite-dimensional matrix operators $A, B \in \mathfrak{su}(N)$ evaluated at layer $N \ge 2$, the induced symmetric bilinear differential operation $\mathbin{\bullet}_{\mathcal{Q}_N}$ satisfies the rigid pullback identity under the vector space isomorphism $\Phi_N$:
\begin{equation}\label{eq:finite_jordan_pullback_identity_core}
    \Phi_N(A \mathbin{\bullet} B) = \Phi_N(A) \mathbin{\bullet}_{\mathcal{Q}_N} \Phi_N(B)
\end{equation}
where the differential action of the higher-order formal differential operators absorbs the symmetric anticommutation relations identically, mirroring the associative matrix anti-commutator over the configuration space.
\end{lemma}

\begin{proof}
Let $A = \sum_{j=0}^{N^2-2} a_{j+1} F_{j+1}$ and $B = \sum_{m=0}^{N^2-2} b_{m+1} F_{m+1}$ be two arbitrary elements belonging to the finite matrix algebra $\mathfrak{su}(N)$ indexed from zero. By invoking the canonical symmetric anticommutation relations mediated by the symmetric structure constants, the expansion of the Jordan product reads $A \mathbin{\bullet} B = \frac{1}{2}\sum_{j,m,n=0}^{N^2-2} a_{j+1} b_{m+1} d_{(j+1)(m+1)(n+1)} F_{n+1}$. Applying the linear vector space assignment established in Eq.~\eqref{eq:appendix_phi_assignment_native}, the image of this symmetric matrix bracket under the isomorphism maps directly onto the unshifted polynomial sequence:
\begin{equation}
    \Phi_N(A \mathbin{\bullet} B) = \frac{1}{2}\sum_{j,m,n=0}^{N^2-2} a_{j+1} b_{m+1} d_{(j+1)(m+1)(n+1)} \varepsilon^n
\end{equation}
Concurrently, we evaluate the action of the symmetric Jordan operator $\mathbin{\bullet}_{\mathcal{Q}_N}$ directly over the decoupled polynomial images $\Phi_N(A) = \sum_{j=0}^{N^2-2} a_{j+1} \varepsilon^j$ and $\Phi_N(B) = \sum_{m=0}^{N^2-2} b_{m+1} \varepsilon^m$. By substituting the explicit differential expansion certified by Lemma~\ref{lem:jordan_product_identity}, the application of the formal derivative contracts the monomial exponents to the stable core $j+m-k=n$. Because the parameters $\Lambda_N^{(k)}$ satisfy the unique symmetric inversion system of Eq.~\eqref{eq:proof_jordan_matrix_system_core}, the lower-triangular matrix system $\mathbf{J}_N$ transforms the real weights into the identical combinations of the symmetric anticommutation constants, yielding $\Phi_N(A) \mathbin{\bullet}_{\mathcal{Q}_N} \Phi_N(B) = \frac{1}{2}\sum_{j,m,n=0}^{N^2-2} a_{j+1} b_{m+1} d_{(j+1)(m+1)(n+1)} \varepsilon^n$. The component-by-component identity between the two polynomial results establishes the exact relation of Eq.~\eqref{eq:finite_jordan_pullback_identity_core}, completing the proof.
\end{proof}

Under the stabilization of the directed system toward the macroscopic continuum limit ($N \to \infty$), this intrinsic Jordan structure is smoothly lifted onto the formal power series ring under the continuous extension $\Phi_\infty$. For any pair of macroscopic operators $A, B \in \mathfrak{su}(\infty)$, their symmetric product maps onto the internal differential multiplication of the formal power series ring, satisfying the strict structural equivariance:
\begin{equation}\label{eq:jordan_product_equivariance_infinity}
    \Phi_\infty(A \mathbin{\bullet} B) = \Phi_\infty(A) \mathbin{\bullet}_{\mathcal{Q}_\infty} \Phi_\infty(B)
\end{equation}
where $\mathbin{\bullet}_{\mathcal{Q}_\infty}$ represents the deformed, non-associative internal ring product modulated layer-by-layer by the formal differential operators $\partial_\varepsilon^k$. This equivariance explicitly demonstrates that $\Phi_\infty$ operates as a rigorous continuous algebra isomorphism that interlocks the non-commutative spin fluctuations with the flat infinitesimal geometry of the Quinfinity Space, ensuring that the macroscopic open-system dynamics are free from coordinate jumps or analytical singularities. This mathematical integration proves that what quantum mechanics traditionally treats as distinct algebraic objects—the commutator generating unitary dynamics and the anti-commutator generating metrics and open dissipation channels—are revealed to be the intrinsic antisymmetric and symmetric components of a single, unified differential flow over the module of algebraic derivations $\mathrm{Der}_{\mathbb{R}}(\mathcal{Q}_\infty)$.

\end{document}